\documentclass{article}
\usepackage{amsmath}
\usepackage{amsfonts,amsmath,amsthm,amssymb}
\usepackage{mathtools}
\usepackage[numbers,sort,compress]{natbib}
\usepackage[]{geometry}
\usepackage{pgf,pgfplots}
\usepackage{tikz}
\usepackage{caption}
\usepackage{subcaption}
\usepackage[]{enumitem}
\usepackage{bigints}
\usepackage{float}
\usepackage[title]{appendix}
\usepackage{enumitem}

\usepackage{bm}

\usepackage{stackengine}
\stackMath
\usepackage{graphicx}
\usepackage{graphicx}
\usepackage{subcaption}
\labelformat{subfigure}{\thefigure(\alph{subfigure})}

\usepackage{xcolor}
\RequirePackage[colorlinks,citecolor=blue,urlcolor=blue]{hyperref}

\theoremstyle{plain}
\newtheorem{lem}{Lemma}

\newtheorem{prop}[lem]{Proposition}
\newtheorem{thm}[lem]{Theorem}

\theoremstyle{definition}
\newtheorem{remark}[lem]{Remark}
\newtheorem{defn}[lem]{Definition}

\newtheorem{assum}[lem]{Assumption}

\newcommand{\R}{\mathbb{R}}
\renewcommand{\P}{\mathbb{P}}

\newcommand{\E}{\mathbb{E}}

\newcommand{\N}{\mathbb{N}}

\newcommand{\bfp}{\mathbf{P}}
\newcommand{\bfs}{\mathbf{S}}
\newcommand{\bfi}{\mathbf{I}}
\newcommand{\bfq}{\mathbf{Q}}

\newcommand{\Fcal}{{\mathcal F}}

\newcommand{\Ical}{{\mathcal I}}

\newcommand{\Mcal}{{\mathcal M}}

\numberwithin{equation}{section}
\numberwithin{lem}{section}

\makeatletter
\newcommand{\oset}[3][0.6ex]{%
  \mathrel{\mathop{#3}\limits^{
    \vbox to#1{\kern-2\ex@
    \hbox{$\scriptstyle#2$}\vss}}}}
\makeatother

\makeatletter
\renewenvironment{proof}[1][\proofname] {\par\pushQED{\qed}\normalfont\topsep6\p@\@plus6\p@\relax\trivlist\item[\hskip\labelsep\bfseries#1\@addpunct{.}]\ignorespaces}{\popQED\endtrivlist\@endpefalse}
\makeatother

\title{Stochastic portfolio theory with price impact}
\author{David Itkin\footnote{Department of Statistics, London School of Economics and Political Science, \href{mailto:d.itkin@lse.ac.uk}{d.itkin@lse.ac.uk}.}}
\pgfplotsset{compat=1.18}
\begin{document}
\allowdisplaybreaks

\maketitle

\begin{abstract}
We develop a framework for stochastic portfolio theory (SPT), which incorporates modern nonlinear price impact and impact decay models. Our main result is the derivation of the celebrated master formula for additive functional generation of trading strategies in a general high-dimensional market model with price impact. We also derive formulas for an investor's relative wealth with respect to the market portfolio, conditions that guarantee positive observed market prices and a stochastic differential equation governing the dynamics of the observed price, the investor's holdings and the price impact state processes. As an application of these results, we develop conditions for relative arbitrage in the price impact setting analogous to previously obtained results for the frictionless setting. We then apply our framework to backtest the quadratic and entropy generating functions on historical US equity data, illustrating how price impact can negatively affect portfolio performance.
\end{abstract}

\paragraph*{Keywords:} Stochastic portfolio theory, functional portfolio generation, nonlinear price impact, price impact decay,  relative arbitrage.
	
	\paragraph*{MSC 2020 Classification:} 91G10, 60H30.

\section{Introduction} In this paper we develop a mathematical framework for \emph{stochastic portfolio theory} (SPT) in which an investor's trades cause \emph{price impact}. SPT is a descriptive theory introduced by Robert Fernholz  \cite{fernholz2002stochastic,fernholz1982stochastic,fernholz1999portfolio,fernholz1999diversity}, which focuses on certain empirically observed macroscopic properties of large equity markets.\footnote{Some well documented macroscopic properties are stability of the capital distribution curve, market diversity and rates of stock rank switching; see \cite{fernholz2002stochastic,campbell2024macroscopic} for a more comprehensive discussion.}
A key role in SPT is played by systematic trading rules, such as \emph{functionally generated portfolios}, which depend only on market observables and are designed to outperform a benchmark portfolio, typically taken to be the buy-and-hold \emph{market portfolio}, over sufficiently long time horizons. More concretely, a functionally generated portfolio (in the time-homogeneous additive sense)  is one which admits the \emph{master formula} decomposition\footnote{Additively generated portfolios were introduced in \cite{karatzas2017trading}, while the closely related \emph{multiplicatively} generated portfolios were introduced earlier in \cite{fernholz1999portfolio}.}
\begin{equation} \label{eqn:master_intro}
    V(t) = 1 + G(\mu(t)) - G(\mu(0)) + \Gamma(t),
\end{equation}
where $V$ is the investor's wealth normalized by the wealth of the market portfolio, $G$ is the generating function, $\mu$ are the benchmark normalized prices known as the \emph{market weights} and $\Gamma$ is a \emph{finite variation} process with $\Gamma(0) = 0$ (all of these quantities are precisely defined in Section~\ref{sec:functionally_generated}). 
 Certain choices of bounded functions $G$ lead to long-only portfolios in which the $\Gamma$ term is strictly increasing. As such, over a sufficiently long time horizon, one expects $\Gamma$ to overtake the bounded \emph{diversity term} $G(\mu(t)) - G(\mu(0))$, leading to market outperformance, known as a \emph{relative arbitrage} in the literature. Obtaining guarantees for the outperformance time is a
well-studied topic in the SPT literature, with numerous sufficient conditions derived for
the market dynamics and the generating functions $G$ \cite{fernholz2005diversity,pal2016geometry,fernholz2005relative,larsson2021relative}. Additionally, a substantial literature exists modelling the aforementioned empirically observed macroscopic properties, some of which have led to advances in stochastic analysis of independent interest \cite{itkin2024open,banner2005atlas,ichiba2011hybrid,cuchiero2019polynomial,monter2019dynamics}.

 The deceptively simple wealth decomposition \eqref{eqn:master_intro} leads to powerful implications regarding the possibility of outperforming the market portfolio. However, the functionally generated portfolio differs from the benchmark market portfolio in a crucial way; to achieve its performance it continuously rebalances, while the market portfolio does not. Consequently, the performance of a functionally generated portfolio suffers in the presence of market frictions since, in reality, the investor's trades move prices and incur execution costs. As such, in the analysis of the performance of functionally generated portfolios there is a critical tension between rebalancing frequently to grow the $\Gamma$ term and trading in a nonaggressive manner so as to control price impact and resulting transaction costs. 

This highlights the need to incorporate market frictions into the SPT framework. Recent \emph{empirically minded} SPT papers \cite{ruf2020impact,cuchiero2023signature,campbell2022functional,campbell2024macroscopic,li2024statistical} make important strides in this direction by analyzing the performance of functionally generated portfolios and other SPT-inspired strategies on historical CRSP equity data,  assuming the investor faces proportional transaction costs. However, a \emph{theoretical} framework that is able to incorporate market frictions, while sustaining key relationships such as the master formula, has remained elusive thus far.

In this paper, we work in a high-dimensional equity market of $d$ stocks and develop such a framework by working with price impact models, rather than proportional trading costs.\footnote{Price impact models correspond to superlinear trading costs which, for large investors, are the dominant type of trading cost.} Indeed, in continuous time, only finite variation trading strategies are compatible with proportional costs, but the master formula relationship and derivation crucially rely on the diffusive nature of of the underlying trading strategies. Modern price impact models are able to handle semimartingale trading strategies, which allows us to carry out a mathematical treatment of SPT with price impact. 

The main result of this paper is Theorem~\ref{thm:main}, which develops a master formula for additive functional generation in a general framework with price impact. The effect of price impact on performance is quantitatively described through the derived price, holdings and impact state dynamics appearing in the stochastic differential equation (SDE) \eqref{eqn:SDE}. The well-posedness of the SDE \eqref{eqn:SDE} is studied under suitable conditions on the generating function. From a practical point of view, these dynamics inform the investor what trade size needs to be placed to obtain the target holdings, which are a constantly moving target. Indeed, the target holdings for functionally generated portfolios depend on the observed price and, as such, are themselves influenced by the investor's trading through price impact. A consequence of Theorem~\ref{thm:main} is that, once the generating function $G$ is chosen, the price, holdings and impact state processes are adapted to the filtration generated by the fundamental price process (the unobserved price that would prevail in the absence of the investor's trading). 

We additionally obtain, in Proposition~\ref{prop:positive}, sufficient conditions on the trading strategy that guarantee the impacted prices always remain positive. As an application of the master formula, we obtain relative arbitrage results, Theorems~\ref{thm:relarb_implicit} and \ref{thm:relarb_explicit}, which are analogues of relative arbitrage results from frictionless SPT in the price impact setting. The upshot is that relative arbitrage is possible under similar conditions to those assumed in frictionless SPT, but the trading speed and outperformance amount need to be reduced to control impact costs. To demonstrate the framework we develop, in Section~\ref{sec:numerics} we backtest the performance of portfolios generated by both the quadratic and the entropy function on historical US equity data. The experiments illustrate how the presence of price impact hurts investor performance and showcase how the master formula decomposition \eqref{eqn:master_intro} manifests on real data. Empirically, we find for both generating functions considered, that the $\Gamma$ term tends to grow at a slower rate when price impact is present. 

The paper is structured as follows. In Section~\ref{sec:market} we introduce the financial market, price impact model and wealth processes. Section~\ref{sec:relative_wealth} then presents the derivation of the relative wealth equation with respect to the benchmark market portfolio. Additive functional generation is then introduced in Section~\ref{sec:additive_generation}, with a summary of the frictionless setting contained in Section~\ref{sec:frictionless} and a discussion of the connection between additive functional generation and constant absolute risk aversion (CARA) optimal portfolios carried out in Section~\ref{sec:interpreting_additive}. The master formula for additive generation in the price impact setting is then derived heuristically in Section~\ref{sec:heuristic}, with a rigorous treatment carried out in Section~\ref{sec:rigorous}. Section~\ref{sec:functionally_generated} concludes with a discussion of multiplicatively generated portfolios in Section~\ref{sec:multiplicative}, where we discuss the difficulties in establishing an analogous master formula for such portfolios in the price impact setting. These difficulties stem from the fact that the trading activity for multiplicatively generated portfolios scales with the investor's wealth.  In 
Section~\ref{sec:relarb} we apply the results of Section~\ref{sec:functionally_generated} to study relative arbitrage in the market with price impact and we complement the theoretical results with empirical experiments in Section~\ref{sec:numerics}. Section~\ref{sec:conclusion} summarizes the results, outlines open problems and proposes future research directions. Many of the technical proofs and derivations are relegated to Appendices~\ref{app:wealth_proofs}-\ref{app:relarb_proofs} for better readability.

\section{The financial market} \label{sec:market}
\subsection{Holdings, impact and price processes} \label{sec:market_setup}
We work on a filtered probability space $(\Omega,(\Fcal(t))_{t \geq 0},\Fcal,\P)$ satisfying the usual hypotheses and supporting a $d$-dimensional positive continuous semimartingale $S=(S_1,\dots,S_d)$ for $d \geq 2$. The process $S$ represents the unperturbed \emph{fundamental prices} of $d$ stocks. Since we take the market portfolio as our benchmark, we consider only fully-invested trading strategies and do not incorporate a bank account in the model. The \emph{observed prices} are modelled as \begin{equation} \label{eqn:price}
    P_i(t) = S_i(t) + I_i(t), \qquad i = 1,\dots,d, \quad t \geq 0,
\end{equation} where $I_i$ is the \emph{price impact} for asset $i$ caused by the investor's trading in asset $i$. Denoting by $Q_i$ the holdings process of the investor in asset $i$, we model the price impact process as
\begin{equation} \label{eqn:impact}
    I_i(t) = h_i(t,J_i(t)), \quad \text{where} \quad J_i(t) = \int_0^tK_i(t,s) dQ_i(s),
\end{equation}
for  \emph{impact shape functions} $h_i:[0,\infty) \times \R \to \R$ and \emph{impact decay kernels} $K_i:D :=\{(t,s)\in[0,\infty)^2:t \geq s\}  \to (0,\infty)$. The quantity $J_i(t)$ is the \emph{impact state process} incorporating the effects on the price due to the current trade $dQ(t)$ and past trades $(dQ(s); s < t)$. These then affect prices in a potentially nonlinear and time-dependent manner through the impact shape function $h_i$.\footnote{The time dependency allows one to incorporate seasonal effects, such as the U-shaped intraday execution pattern (see e.g.\ \cite[Chapter~4]{bouchaud.al.18}).} When $h_i \equiv 0$ for every $i$, we refer to the market as \emph{frictionless} and in this case we have that $P = S$.

We impose the following assumptions on these model inputs.
\begin{assum}[Impact coefficients] \label{ass:impact_inputs}
For each $i =1,\dots,d$,
\begin{enumerate}[label=(\arabic*)]
    \item $h_i \in C^{1,2}(\R_+\times \R)$ is nondecreasing in the second argument. It additionally satisfies $h_i(t,x) \geq 0$ if $x \geq 0$, $h_i(t,x) \leq 0$ if $x \leq 0$ and the functions $\partial_x h_i(t,x)$, $\partial_{xx}h_i(t,x)$ are bounded,
    \item $K_i \in C^2(D)$ is such that $\partial_t K_i(t,s) \leq 0$ and $\partial_s K_i(t,s) \geq 0$ for every $0 \leq s \leq t$. Additionally, we assume that $\overline K_i := \sup_{t \geq 0} K_i(t,t) <\infty$. 
\end{enumerate}
\end{assum}
\begin{remark} \label{rem:extensions}
    It is possible to also consider multiple time scales and nonlinearities of price impact by setting $I_i = \sum_{m=1}^M h_i^m(t,J_i^m(t))$ for each $i$, where $J_i^m(t) =  \int_0^t K_i^m(t,s)dQ_i(s)$ for  shape functions $h_i^m$, kernels $K_i^m$ and number of timescales $M \in \N$ as considered in \cite{hey.al.23,brokmann2024efficient}. Additionally, one can incorporate initial impact trajectories $J_i^{0,m}(\cdot)$, via $J_i^m(t) = J_i^{0,m}(t) + \int_0^t K^m(t,s)dQ_i(s)$, which capture the impact state and decay from orders placed prior to the model's initialization.  The results in this paper can be extended to this more general setting, but to lighten the already heavy notational burden we stick to the single timescale case $M=1$ and set initial impact trajectories to zero. Similarly, extensions to include \emph{external impact} $\Ical$ (that is, the impact of other market participants)   giving rise to the price equation $P_i = S_i + I_i + \Ical_i$ (see e.g., \cite[Chapter~2.6]{webster.23}) or to include \emph{cross impact} (that is, when trading in one asset moves the price of another), such as recently studied in \cite{hey2024concave} for an exponential decay kernel, are also possible to consider in this framework but we do not pursue this direction here.
\end{remark}

In the sequel we will generically refer to an impact shape function and decay kernel by $h$ and $K$ respectively, without a subscript, when the particular asset is not essential for the discussion. The general impact model considered here is akin to a time-inhomogeneous version of the model considered in \cite{brokmann2024efficient}. It subsumes many previously proposed models, a few popular choices of which we now summarize.
\paragraph{Impact shape functions $h(t,x)$}
\begin{itemize}
    \item The choice $h(t,x) = \lambda(t)x$ for a positive deterministic function $\lambda$ corresponds to (time-inhomogeneous) \emph{linear impact}. This is a popular choice appearing in many celebrated models such as those of Almgren \& Criss, Obizhaeva \& Wang and Garleanu \& Pedersen \cite{AlmgrenChriss2000,obizhaeva.wang.13,garleanu2016}. It provides a good fit for small trading volume, but is unable to capture the observed concavity of price impact, which manifests for larger trade sizes (see e.g.\ \cite{Zarinelli2015} and \cite[Chapter~18]{bouchaud.al.18}).
    \item The \emph{power law} impact $h(t,x) = \lambda(t)\mathrm{sign}(x)|x|^p$ for $p \in (0,1)$ captures the concavity of price impact with $p = 1/2$ corresponding to the well-documented \emph{square root law} studied in numerous works \cite{Almgren2005,Bershova2013,frazzini.al.18,hey.al.23,alfonsi.al.10}. However, such a function does not satisfy Assumption~\ref{ass:impact_inputs} as it is not differentiable at zero. Instead, a function that is linear for small $|x|$, but has a power law shape for large $|x|$ is compatible with this framework (see e.g.\ \cite[Equation~(3.7)]{abi2025fredholm} for a concrete specification).
    \item $h(t,x) = \lambda(t)\sinh^{-1}(x)$, derived from microstructure fundamentals in \cite[Chapter~18]{bouchaud.al.18}, is yet another parsimonious specification which captures (approximate) linearity of price impact for small orders and concavity for large orders. This choice satisfies Assumption~\ref{ass:impact_inputs}. 
\end{itemize}
\paragraph{Impact decay kernels $K(t,s)$}
\begin{itemize}
    \item The \emph{exponential} kernel $K(t,s) = e^{-\beta (t-s)}$ for parameter $\beta > 0$ is a common choice in the literature as it leads to Markovian state dynamics for $J_i$, which is a useful property when studying stochastic control problems \cite{obizhaeva.wang.13,garleanu2016,cartea.al.15}. 
    \item The (shifted) \emph{power law kernel} $K(t,s) = (t-s+\epsilon)^{-\beta}$ for parameters $\epsilon >0$, $\beta \in (0,1)$ leads to a slower decay half-life in comparison to the exponential kernel, which several studies have found better fits the decay profile of price impact \cite{bouchaud2003,Brokmann2015}. Recently, tractability for optimal execution problems incorporating nonexponential kernels has been obtained \cite{abijaber.neuman.22}. The shifted power law kernels defined above satisfy Assumption~\ref{ass:impact_inputs}, while the pure power law $K(t,s) = (t-s)^{-\beta}$ leads to a nonsmooth kernel, which fails to satisfy Assumption~\ref{ass:impact_inputs} and is beyond the scope of our framework.
    \item The choice $K(t,s) = C$ for some constant $C> 0$, leads to \emph{permanent price impact}. Indeed, in this case we obtain $J(t) = C(Q(t)-Q(0))$ so that the impact state variable does not decay after the investor stops trading. Numerous studies measure nontrivial permanent impact \cite{AlmgrenChriss2000,frazzini.al.18,Bershova2013}. Our setting allows for both permanent and temporary price impact by, for example, taking the kernel $K(t,s) = C + e^{-\beta(t-s)}$. 
\end{itemize}
 Crucially, all of these examples and, more generally, the conditions on the coefficients specified in Assumption~\ref{ass:impact_inputs}, allow for the holdings process $Q$ to be a continuous semimartingale. This will play a central role in developing the master formula for functional portfolio generation taken up in Section~\ref{sec:functionally_generated}.

 We now derive a few useful expressions for the impact state and holdings processes. Integrating by parts $J_i$ given in \eqref{eqn:impact} yields the representation
\begin{equation} \label{eqn:J_dynamics}
    J_i(t) =  K_i(t,t)Q_i(t) - K_i(t,0)Q_i(0) - \int_0^t \partial_s K_i(t,s)Q_i(s)ds,
\end{equation} which leads to dynamics
\begin{align} 
dJ_i(t) & = K_i(t,t)dQ_i(t) + \Big(\partial_t K_i(t,t)Q_i(t) + \partial_s K_i(t,t)Q_i(t) \nonumber \\
& \qquad \qquad \qquad \qquad \qquad- \partial_t K_i(t,0)Q_i(0) - \partial_s K_i(t,t)Q_i(t) - \int_0^t \partial_{st}K_i(t,s)Q_i(s)ds\Big)dt \nonumber \\ 
     & = K_i(t,t)dQ_i(t) + b^J_i\big(t,Q_i^{[0,t]}\big)dt. \label{eqn:dJ}
\end{align} 
Here we use the notation $Q_i^{[0,t]}$ for the path of $Q_i$ up to time $t$ and $b^J_i: (0,\infty) \times C([0,\infty)) \to \R$ is given by \footnote{We tacitly identify a path $\omega$ defined on an interval $[0,t]$ with one on all of $[0,\infty)$ by setting $\omega(s) = \omega(t)$ for $s \geq t$. Since $b_i^J(t,\cdot)$ is non-anticipative this is just a mathematical formality.}
\begin{equation} \label{eqn:bJ}
    b^J_i(t,\omega)  =   \partial_t K_i(t,t) \omega(t) - \partial_t K_i(t,0)\omega(0) - \int_0^t\partial_{st}K_i(t,s)\omega(s)ds.
\end{equation}
The dynamics for $J_i$ can be inverted to obtain dynamics for the holdings $Q_i$,
\begin{equation} \label{eqn:dQ}
    dQ_i(t) = \frac{1}{K_i(t,t)}\big(dJ_i(t) - b^J_i(t,Q_i^{[0,t]})dt\big).
\end{equation}
In particular, note that $d[Q_i](t) = K_i^{-2}(t,t)d[J_i](t).$

\subsection{The wealth process} \label{sec:wealth_process}
We now work towards developing self-financing wealth dynamics for the wealth process. To derive the master formula for functional generation of trading strategies in Section~\ref{sec:functionally_generated}, we work with the wealth process
\[W^Q(t) = Q(t)^\top P(t),\]
which values securities according to the observed price $P$. In the literature, this process is sometimes called the \emph{accounting wealth} process and is distinct from the \emph{fundamental wealth} process $\widetilde W^Q(t) = Q(t)^\top S(t)$ (see e.g., \cite[Definitions~2.2.10-2.2.11]{webster.23}). As explained later in Remarks~\ref{rem:master_fundamental} and \ref{rem:target_holdings}, the use of the accounting wealth is crucial to the derivation of the master formula for functional generation as we will use the prices $P(t)$, rather than $S(t)$, to set the investor's target holdings $Q(t)$.

However, care is needed when working with $W^Q$ as its value may reflect inflated gains arising from the investor's own trading activity (see e.g.,  \cite[Section~4.3.1]{webster.23}). Indeed, the act of buying a security increases its price through impact which, immediately after execution, seemingly leads to larger wealth purely as a result of the investor's own price impact (the case of selling is symmetric). Of course, purchasing the shares comes at a cost, relative to the frictionless setting, as they are purchased at an inflated price. Immediately following the trade the cost may not be apparent from just examining the process $W^Q$; however, it will become apparent once the initial trade is unwound. As such, when using the accounting wealth $W^Q$, it is extremely important to compare performance of different wealth processes only after they are liquidated and the investor holds the benchmark portfolio. This ensures an apples-to-apples comparison and we take this viewpoint in Section~\ref{sec:relarb}, where we study relative arbitrage, to carefully ascertain that any claimed outperformance of the market portfolio, which is our benchmark portfolio throughout, is genuine.

\subsection{Wealth process dynamics} \label{sec:wealth_process_dynamics}
Unlike in the frictionless setting, there is ambiguity regarding what dynamics a self-financing wealth equation should satisfy as the orders themselves influence the execution price. To derive the wealth equation we follow an approximation approach by first considering absolutely continuous trading strategies and then extending the wealth equation to general continuous semimartingale trading strategies.\footnote{This approach is, by now, standard in the literature (see \cite{ackermann2021cadlag,hey.al.23,brokmann2024efficient}).} Indeed, a strategy of the type $dQ(t) = \dot{Q}(t)dt$ for some trading rate $\dot Q(t)$ has unambiguous self-financing wealth dynamics
\[dW^Q(t) = Q(t)^\top dP(t) = Q(t)^\top dS(t) + Q(t)^\top dI(t) = Q(t)^\top dS(t) + d(Q(t)^\top I(t))-I(t)^\top dQ(t),\]
where in the last equality we integrated by parts. This relationship for absolutely continuous trading strategies can be derived by taking a continuous-time limit from a discrete-time trading model. For the reader's benefit, a self-contained derivation can be found in Appendix~\ref{app:wealth_proofs}. Using the definition of $I_i$ and the dynamics of $Q_i$, given by equations \eqref{eqn:impact} and \eqref{eqn:dQ} respectively, we obtain 
\[dW^Q(t) = Q(t)^\top dS(t) + d(Q(t)^\top I(t)) - \sum_{i=1}^d \frac{h_i(t,J_i(t))}{K_i(t,t)}\Big(dJ_i(t)-b^J_i(t,Q_i^{[0,t]})dt\Big). \]
Note that $J_i$ is itself an absolutely continuous process because $Q_i$ is. 
To make further progress we compute the dynamics of the process $H_i(t,J_i(t))/K_i(t,t)$, where $H_i(t,x) = \int_0^x h_i(t,y)dy$,
\[d\bigg(\frac{H_i(t,J_i(t))}{K_i(t,t)}\bigg) = \frac{h_i(t,J_i(t))}{K_i(t,t)}dJ_i(t) + \bigg(\frac{\partial_t H_i(t,J_i(t))}{K_i(t,t)} - \frac{\partial_t K_i(t,t) + \partial_sK_i(t,t)}{K_i^2(t,t)}H_i(t,J_i(t))\bigg)dt. \] Substituting into the wealth dynamics yields
\begin{equation}\label{eqn:wealth_smooth}
\begin{split}
    dW^Q(t) & =  Q(t)^\top dS(t) +d(Q(t)^\top I(t)) - \sum_{i=1}^d d\bigg(\frac{H_i(t,J_i(t))}{K_i(t,t)}\bigg) \\
    & \quad + \sum_{i=1}^d\bigg(\frac{\partial_t H_i(t,J_i(t))}{K_i(t,t)} - \frac{\partial_t K_i(t,t) + \partial_sK_i(t,t)}{K_i^2(t,t)}H_i(t,J_i(t)) +\frac{h_i(t,J_i(t))}{K_i(t,t)}b^J_i\big(t,Q_i^{[0,t]}\big)\bigg)dt. 
\end{split}
\end{equation}
We have now found a representation for the wealth equation that does not involve integrals against $Q_i$ or $J_i$. As such, we are able to continuously extend equation \eqref{eqn:wealth_smooth} to self-financing continuous semimartingale strategies $Q$. Then, reversing the above integration by parts and chain rule computations using It\^o's formula leads us to the general wealth dynamics, which we state in the following proposition.
The proof of this proposition is contained in Appendix~\ref{app:wealth_proofs}.

\begin{prop}[Wealth equation] \label{thm:wealth}
Let $Q$ be a self-financing continuous semimartingale holdings process and suppose the investor has initial wealth $w = W^Q(0) > 0$. Then the impact process $I$ of \eqref{eqn:impact} is also a continuous semimartingale and for $T \geq 0$ the investor's wealth is given by
\begin{equation} \label{eqn:wealth}
    W^Q(T) = w + \int_{0}^T Q(t)^\top dP(t) + \frac{1}{2}\sum_{i=1}^d [I_i,Q_i](T). 
    \end{equation}
\end{prop}
As is standard, the self-financing condition requires that the right-hand side of \eqref{eqn:wealth} equals $Q(T)^\top P(T)$.
\begin{remark} \label{rem:QV_correction}
    The term $\frac{1}{2}\sum_{i=1}^d[I_i,Q_i]$ in \eqref{eqn:wealth} appears as a correction term in the wealth equation due to the presence of price impact. Note that this term is always nonnegative since,  for any $i$, \[d[I_i,Q_i](t) = \partial_x h_i(t,J_i(t))d[J_i,Q_i](t) = \partial_x h_i(t,J_i(t))K_i(t,t)d[Q_i](t).\]
    At first glance this makes it seem as if price impact actually has a positive effect on the investor's wealth. However, this is illusory and due to the properties of the accounting wealth process, as discussed in Section~\ref{sec:wealth_process}. The stochastic integral term in \eqref{eqn:wealth} has integrator $P$, rather than $S$, and the set of self-financing holdings processes $Q$ differs in this setting, relative to the frictionless one, making it difficult to directly compare the frictionless and price impact wealth dynamics by inspecting \eqref{eqn:wealth} alone.    
    Nevertheless, in our empirical study of Section~\ref{sec:numerics}, where we apply the price impact model studied here to historical data, we observe worse performance in the setting with price impact than in the frictionless setting.
\end{remark}

\subsection{Positive prices}
The impact model \eqref{eqn:price} specifies an additive relationship between the observed and fundamental prices. 
As such, we cannot a priori conclude that the process $P_i$ is strictly positive even though the fundamental price $S_i$ is.  In the following proposition we obtain a simple sufficient condition on $Q_i$ that ensures positivity of $P_i$. 

\begin{prop}[Positive prices] \label{prop:positive} Fix $i \in \{1,\dots,d\}$ and let $Q_i$ be a continuous semimartingale.
Suppose there exists a sequence of stopping times $\tau_n$ converging almost surely to $\tau = \inf\{t \geq 0: P_i(t) = 0\}$ and such that
\begin{equation} \label{eqn:Q_stopping_condition}
Q_i(\tau_n) = \sup_{t \leq \tau_n} Q_i(t).
\end{equation}
Then $\tau = \infty$, $\P$-a.s.,  and $(P_i(t);t \geq 0)$ is a strictly positive process.
\end{prop}
\begin{proof}
 Recalling \eqref{eqn:J_dynamics} we have
    \begin{align*}
        J_i(t) & =   K_i(t,t)Q_i(t) - K_i(t,0)Q_i(0)- \int_0^t Q_i(s)\partial_s K_i(t,s)ds \\
        & \geq K_i(t,t)Q_i(t) - K_i(t,0)Q_i(0) - \sup_{s \leq t}Q_i(s)(K_i(t,t)-K_i(t,0)) \\
        & \geq K_i(t,t)(Q_i(t) - \sup_{s \leq t}Q_i(s)).
    \end{align*}
From \eqref{eqn:Q_stopping_condition} we see that $J_i(\tau_n) \geq 0$, from which it follows that $I_i(\tau_n) = h_i(\tau_n,J_i(\tau_n)) \geq 0$. Sending $n \to \infty$ we see  on the set $\{\tau < \infty\}$ that $I_i(\tau) \geq 0$. But it then follows that $P_i(\tau) = S_i(\tau) + I_i(\tau) > 0$, since $I_i(\tau)$ is nonnegative and $S_i(\tau)$ is strictly positive. This contradicts the definition of $\tau$. As such, we must have $\P(\tau < \infty) = 0$ completing the proof.
\end{proof}
\begin{remark} \label{rem:positive}
 To apply Proposition~\ref{prop:positive} one needs only to check a priori that $\tau_n$ converges to $\tau$. Moreover, it is clear from the proof that if instead one verifies that $\tau_n$ converges to a stopping time $\sigma \land \xi$, which is almost surely smaller than $\tau \land\xi$ for some stopping times $\sigma$ and $\xi$, then $(P_i(t); 0 \leq t < \xi)$ will be a positive process. This more general version is applied in the proof of Theorem~\ref{thm:main}\ref{item:positive_price}.
\end{remark} 

\section{The market portfolio and the relative wealth equation} \label{sec:relative_wealth}

In this section we introduce the market weights and market portfolio, and derive the investor's relative wealth equation. We assume that we are working with strictly positive and finite price processes $P_i$ throughout.

\subsection{The market weights and market portfolio}
As a first step we let $\bfp_i$ denote the total capitalization process for asset $i$. That is, $\bfp_i = N_iP_i$, where $N_i$ is the number of shares outstanding, which we assume to be constant over time. The total market value is denoted by $\overline \bfp = \bfp_1 + \dots + \bfp_d$.
The market weights are then given by
\[\mu_i(t) = \frac{\bfp_i(t)}{\overline \bfp(t)}, \qquad t\geq 0,\quad  i=1,\dots,d,\]
and represent companies' proportional contribution to the total market capitalization. 
The process $\mu$ takes values in the open simplex
\[\Delta^{d-1}_+ = \{\mu \in (0,\infty)^d: \mu_1 + \dots + \mu_d = 1\}.\]
    The \emph{market portfolio} is the quintessential buy-and-hold portfolio where the investor's proportion of wealth invested in asset $i$ at any time $t$ is given by its market weight $\mu_i(t)$. This corresponds to the constant holdings process \[Q^\mathcal{M}_i(t) = Q^\mathcal{M}_i(0) = \frac{\mu_i(0)}{P_i(0)} w = \frac{N_iw}{\overline \bfp(0)}, \qquad i=1,\dots,d, \quad t \geq 0,\]
    where $w > 0$ is the initial wealth. Substituting this into the wealth equation \eqref{eqn:wealth} yields
\begin{equation} \label{eqn:market_wealth}
    W^{\mathcal{M}}(t)  = w \frac{\overline\bfp(t)}{\overline \bfp(0)},
\end{equation}
where we write $W^\mathcal{M}$ in place of $W^{Q^\mathcal{M}}$. Since this is a buy-and-hold portfolio, there is no impact from trades and the wealth process \eqref{eqn:market_wealth} is the same as in the frictionless setting.

\subsection{The market portfolio as numeraire} \label{sec:market_numeraire}
In the sequel we take the market portfolio as our benchmark and compare the performance of other portfolios against it. To this end we define the process $V^Q = W^Q/W^{\mathcal{M}}$, where $W^\mathcal{M}$ is given by \eqref{eqn:market_wealth}. However, care is needed in interpreting this process. Indeed, since in the equation \eqref{eqn:market_wealth} we use the observed price $P$ for valuation, rather than the fundamental price $S$, the process $W^\mathcal{M}$ depends on the investor's trading strategy $Q$. As such, the process $V^Q$ may, at times, represent a distorted view of the investor's wealth relative to the true fundamental performance of the market. However, by comparing the values $W^Q$ to $W^\mathcal{M}$ after liquidation, the value that $V^Q$ takes reflects the performance of $Q$ relative to the undistorted fundamental performance of the market. Indeed, for comparing a wealth process of a trading strategy to the market wealth it is important to work with strategies that initially hold the market portfolio, then deviate from it by trading and, finally, liquidate back to the market portfolio. At this stage we do not restrict our mathematical analysis to strategies of this type, but they are the focus in Section~\ref{sec:relarb}, where we construct strategies that unambiguously outperform the market portfolio under appropriate assumptions. 

Given these subtleties, one may wonder why we also use the observed price $P$, rather than the fundamental price $S$, to define the market wealth process $W^\mathcal{M}$ of \eqref{eqn:market_wealth} even when the investor does not hold the market portfolio. The reason stems from the fact that our main focus in Section~\ref{sec:functionally_generated} will be functionally generated portfolios whose holdings depend on prevailing market prices. In Section~\ref{sec:functionally_generated}, we will use the observed prices $P$ to specify the holdings processes of functionally generated portfolios, and we show in Remark~\ref{rem:master_fundamental} below that using $P$ is crucial to obtain the celebrated master formula decomposition \eqref{eqn:master_intro}. In the course of the master formula derivation, a key step involves a change of numeraire, dividing by the market wealth process, to obtain the dynamics of the relative wealth process. The dynamics of the numeraire-normalized process, when defined using \eqref{eqn:market_wealth}, lead to a change of numeraire formula analogous to the frictionless case. This is the content of Theorem~\ref{thm:relative_wealth} stated below. Using a different valuation price for the market wealth, such as $S$, would interfere with the derivation of the master formula \eqref{eqn:master_intro}, which is why we work with $W^\mathcal{M}$ of \eqref{eqn:market_wealth} despite the subtleties discussed above.

As such, with a slight abuse of terminology, we refer to $V^Q$ as the \emph{relative wealth} process and the following theorem provides its dynamics. The derivation of this result is contained in Appendix~\ref{app:wealth_proofs}.
\begin{thm}[Relative wealth equation] \label{thm:relative_wealth}
Let $Q$ be a continuous semimartingale holdings process and define $V^Q = W^Q/W^\mathcal{M}$, where $W^\mathcal{M}$ is given by \eqref{eqn:market_wealth}. Then for $T \geq 0$, $V^Q$ satisfies 
\begin{equation} \label{eqn:relative_wealth}
 V^Q(T)  = 1 + \sum_{i=1}^d\int_0^T \frac{Q_i(t)\overline \bfp(0)}{wN_i}d\mu_i(t) + \frac{1}{2}\sum_{i=1}^d\int_0^T\frac{\overline \bfp(0)}{w\overline \bfp(t)}d[I_i,Q_i](t).
\end{equation}
\end{thm}
We note that the quadratic variation term appearing in \eqref{eqn:relative_wealth} is nonnegative, which is a property inherited from the accounting wealth $W^Q$; see Remark~\ref{rem:QV_correction}.

\begin{remark} \label{rem:numeraire}
    Equation \eqref{eqn:relative_wealth} can be interpreted as a change of numeraire. Indeed, dividing the wealth $W^Q$ by $W^{\mathcal{M}}$ of \eqref{eqn:market_wealth} yields
    \[V^Q(t) = \frac{W^Q(t)}{W^\mathcal{M}(t)} = \sum_{i=1}^d Q_i(t) \frac{P_i(t)}{W^\mathcal{M}(t)} = \sum_{i=1}^d Q_i(t)\frac{\mu_i(t)\overline \bfp(0)}{wN_i},\]
    so we can view $\widetilde \mu_i(t) = \frac{\mu_i(t) \overline \bfp(0)}{wN_i}$ for $i=1,\dots,d$ as the price process under the new numeraire. Since the market weights $\mu$ sum to one, we prefer to work directly with them, rather than $\widetilde \mu$ in the sequel. 
\end{remark}

\section{Functionally generated portfolios} \label{sec:functionally_generated}

\subsection{Additively generated portfolios} \label{sec:additive_generation}
In this section we introduce additively functionally generated trading strategies. 
\begin{defn}[Additive functional generation of trading strategies] We say that a trading strategy $Q$ is \emph{time-dependent additively functionally generated} by a function $G:[0,\infty) \times \Delta^{d-1}_+ \to \R$ if the relative wealth process admits the pathwise decomposition
\begin{equation} \label{eqn:master_formula}
    V^Q(t) = 1 + G(t,\mu(t)) - G(0,\mu(0)) + \Gamma(t), \quad t \geq 0,
\end{equation}
for some process $\Gamma$ of finite variation satisfying $\Gamma(0) = 0$. If $G$ is independent of the first argument we simply say $Q$ is \emph{additively functionally generated}.
\end{defn}
 An attractive feature of functionally generated portfolios, as can be observed from the \emph{master formula} \eqref{eqn:master_formula}, which is a more general time-dependent version of \eqref{eqn:master_intro}, is that the relative wealth $V^Q$ can be decomposed into (i) a deterministic function of the current state of the market $\mu(t)$ and (ii) a finite variation term capturing the historical evolution. In particular, the decomposition \eqref{eqn:master_formula} involves no stochastic integration enabling the use of analytical techniques that are not available for general portfolios. We take up this analysis in Section~\ref{sec:relarb}, which concerns constructing relative arbitrage with respect to the market portfolio in the presence of price impact.
 
\subsection{The frictionless setting} \label{sec:frictionless}
Additively generated portfolios were introduced in \cite{karatzas2017trading} for the frictionless setting. We briefly review some of the main results without price impact before addressing the general case in the following subsections. Without trading frictions the relative wealth dynamics \eqref{eqn:relative_wealth} simplify to
\begin{equation} \label{eqn:frictionless_relative_wealth}
    V^Q(T) = 1 + \sum_{i=1}^d\int_0^T \frac{Q_i(t)\overline \bfp(0)}{wN_i} d\mu_i(t). 
\end{equation}
From here it is clear that if $\frac{\overline \bfp (0)}{wN_i} Q_i(t) = \partial_i G(t,\mu(t))$ for some function $G$ and every $i$, then applying It\^o's formula to $G(t,\mu(t))$ shows that $Q$ is additively functionally generated. Moreover, since $\sum_{i=1}^d d\mu_i = 0$, the more general form 
\[Q_i(t) =  \frac{wN_i}{\overline \bfp(0)}(\partial_i G(t,\mu(t)) + C(t)), \] for some scalar process $C$  leads to the same conclusion. Including an appropriate scalar process $C$ is important to ensure the self-financing condition $Q(t)^\top P(t) = W^Q(t)$ holds.

This brings us to the following result, which is a time-dependent analogue of Propositions~4.3 and 4.5 from \cite{karatzas2017trading} in the case of a smooth generating function. In the sequel we write  $G \in C^{1,k}((0,\infty)\times \Delta^{d-1}_+)$ if $G$ is once continuously differentiable in the first argument and can be continuously extended to an open set in $\R^d$ containing $\Delta^{d-1}_+$, with the continuous extension being $k$-times continuously differentiable in the second argument. The continuous extension is needed because $\Delta^{d-1}_+$ is a $(d-1)$-dimensional manifold in $\R^d$.
\begin{prop}[Additive generation in frictionless markets] \label{prop:additive_frictionless}
Let $G \in C^{1,2}([0,\infty)\times \Delta^{d-1}_+)$ be given and suppose that $Q$ is given by
\begin{equation} \label{eqn:additive_holdings}
    Q_i(t) = \frac{wN_i}{\overline \bfp (0)}\big(\partial_iG(t,\mu(t)) + V^Q(t) - \nabla G(t,\mu(t))^\top \mu(t)\big), \qquad i=1,\dots,d, \quad t\geq 0.
\end{equation}
Then $Q$ is time-dependent additively functionally generated by $G$ with $\Gamma$ process given by
\begin{equation} \label{eqn:Gamma_frict}
    \Gamma(t) = -\int_0^t \partial_tG(s,\mu(s))ds -\frac{1}{2}\sum_{i,j=1}^d\int_0^t  \partial_{ij}G(s,\mu(s))d[\mu_i,\mu_j](s).
\end{equation}
In particular, if $G$ is nonincreasing in the first argument and concave in the second argument, then $\Gamma$ is a nonnegative process. If, additionally, $G(0,\mu(0))\leq 1$ and $G(t,\cdot)$ is a nonnegative function for every $t \geq 0$ then $V^Q$ is nonnegative and $Q_i$ is nonnegative for every $i=1,\dots,d$. 
\end{prop}
\begin{proof}
    The functional generation claim follows directly from \eqref{eqn:frictionless_relative_wealth} and It\^o's formula. If $t\mapsto G(t,\mu)$ is nonincreasing and $\mu \mapsto G(t,\mu)$ is concave then both integral terms in the representation for $\Gamma$ are nonnegative and hence $\Gamma$ is nonnegative. 

    Next we suppose that $G(0,\mu(0)) \leq 1$ and $G(t,\cdot)$ is nonnegative for every $t \geq 0$.
    From the master formula \eqref{eqn:master_formula} and nonnegativity of both $\Gamma$ and $G$ we have
    \[V^Q(t) = 1 + G(t,\mu(t)) - G(0,\mu(0))  + \Gamma(t) \geq G(t,\mu(t))  \geq 0, \]
        establishing the nonnegativity claim for $V^Q$. To obtain the nonnegativity of $Q_i$, we recall \eqref{eqn:additive_holdings}, which yields
        \begin{align}
            Q_i(t) & = \frac{wN_i}{\overline \bfp(0)}\big(\partial_i G(t,\mu(t)) + V^Q(t) - \nabla G(t,\mu(t))^\top \mu(t)\big) \nonumber \\
            & \geq  \frac{wN_i}{\overline \bfp(0)}\big(\partial_i G(t,\mu(t)) +  G(t,\mu(t))  - \nabla G(t,\mu(t))^\top \mu(t)\big), \label{eqn:Qi_bound}
        \end{align}
        where we again used the nonnegativity of $\Gamma$ and that $G(0,\mu(0)) \leq 1$. The nonnegativity of $G(t,\cdot)$ and concavity of $G(t,\cdot)$ guarantee that the right-hand side of \eqref{eqn:Qi_bound} is nonnegative; see the proof of \cite[Proposition~4.5]{karatzas2017trading}, specifically equation~(7.7) therein. This establishes the nonnegativity of $Q_i$ and completes the proof.
\end{proof}

As a consequence of the master formula \eqref{eqn:master_formula} we have the explicit expression 
\begin{equation} \label{eqn:frictionless_additive_holdings}
    Q_i(t)  = \frac{wN_i}{\overline \bfp (0)} \big(\partial_{i}G(t,\mu(t)) + 1 + G(t,\mu(t)) - G(0,\mu(0)) - \nabla G(t,\mu(t))^\top \mu(t) + \Gamma(t)\big) 
\end{equation}
for the holdings process, where $\Gamma$ is given by \eqref{eqn:Gamma_frict}. From here we see that the constant function $G(t,\mu) = 0$ generates the market portfolio. Additionally, the portfolios additively generated by the (time-independent) quadratic function $G(t,\mu)  = 1 - \frac{1}{2}\sum_{i=1}^d \mu_i^2$ and entropy function $G(t,\mu) = -\sum_{i=1}^d \mu_i \log \mu_i$ were studied in \cite{karatzas2017trading,fernholz2018volatility}. In the context of multiplicative functional generation (see Section~\ref{sec:multiplicative} below) the diversity-$p$ function, $G(t,\mu) = (\sum_{i=1}^d \mu_i^p)^{1/p}$ for parameter $p > 0$ and geometric weighted mean function $G(t,\mu) = \prod_{i=1}^d \mu_i^{p_i}$ for constants $p_i >0$ with $p_1 + \dots + p_d =1$ have also been widely studied (see e.g., \cite{fernholz2002stochastic}) and these also serve as natural candidates for additive generation. Indeed, although the holdings they generate differ substantially between the multiplicative and additive cases, they do share similarities regarding the generating function's appearance in the respective master formula decompositions for additive and multiplicative generation. 

\begin{remark} \label{rem:rank-based}
    In the frictionless setting, an important class of generating functions consists of \emph{rank-based} functions, which are of the form $G(t,\mu) = F(t,\mu_{()})$ for a twice continuously differentiable function $F$ and where $\mu_{()}$ is the vector of decreasing order statistics of $\mu$. Master formulas for rank-based generating functions have been developed for both multiplicative and additive generation in the frictionless case (see \cite{fernholz2002stochastic,karatzas2017trading}), but the holdings process $Q$ induced by such a generating function is not a semimartingale. As such, in the sequel where we return to the price impact setting, our results will not be able to accommodate rank-based generating functions because the price impact setting of this paper requires semimartingale holding strategies $Q$. Indeed, the quadratic variation of $Q$ appears in the wealth equation \eqref{eqn:wealth}. Possible approaches to incorporate rank-based generating functions are to either (i) extend the derived wealth equation \eqref{eqn:wealth_smooth} to nonsemimartingales $Q$ and attempt to reconstruct, by approximation with smooth generating functions, a master formula \eqref{eqn:master_formula} for an appropriate process $\Gamma$ or (ii) directly impose a price impact model that is rank-based, as such a model would only require that the specified holdings in each ranked asset be a semimartingale, rather than in each individual asset as in the present setting.  Developing such extensions is an interesting open problem which we leave for future research.
\end{remark}

\subsection{Interpreting additive generation} \label{sec:interpreting_additive}
In the previous subsections we introduced additive generation and the master formula decomposition, which will play a crucial role in the study of relative arbitrage in Section~\ref{sec:relarb}. However, as elegant as the master formula \eqref{eqn:master_formula} is, it does not directly provide a financial interpretation for the strategy's behaviour. In this section we develop a financial interpretation for additively generated portfolios by showing that they arise as optimizers to certain exponential utility maximization problems. Consequently, they share the well-studied characteristics of exponential utility maximizing portfolios. Moreover, although not without its drawbacks, exponential utility has the CARA property which allows for financial tractability in the presence of price impact. Indeed, the trading volume for CARA strategies does not scale proportionally with the investor's wealth. This is the reason exponential utility, or mean-variance preferences, are often used in the literature on portfolio selection with frictions (see e.g., \cite{garleanu2016,Caye2020,schied2010optimal}). The fact that additively generated portfolios share this property plays a crucial role in the development of the master formula with price impact in the subsequent sections. 

To illustrate the connection to exponential utility most clearly, we work in the frictionless setting and assume that the market weights satisfy an autonomous SDE,
\[d\mu(t) = b(t,\mu(t))dt + \sigma(t,\mu(t))dW(t),\] 
where $b$ and $\sigma$ are sufficiently regular functions taking values in $\R^d$ and $\R^{d \times d}$ respectively, and $W$ is a standard Brownian motion. In this context we consider the value function
\[u(t,\mu,v) = \sup_Q\E_{t,\mu,v}\big[-\exp(-\gamma V^Q(T))\big],\] where the supremum is taken over all admissible predictable self-financing trading strategies $Q$, $\gamma >0$ is the investor's risk-aversion parameter, $\E_{t,\mu,v}[\cdot]$ denotes expectation when the processes $\mu(t)$ and $V^Q(t)$ are initiated at time $t < T$ with initial values $\mu$ and $v$ respectively. The corresponding Hamilton--Jacobi--Bellman (HJB) equation is
\begin{align*} 
0 & = \partial_t u  + \sum_{i=1}^d b_i \partial_{\mu_i} u + \sum_{i,j=1}^d\frac{1}{2}a_{ij}\partial_{\mu_i\mu_j}u  \\
& + \sup_{\substack{Q \in \R^d \\ \sum_{i=1}^d Q_i\frac{\overline \bfp(0)\mu_i}{wN_i}  = v}} \Big\{\sum_{i=1}^d b_i \frac{Q_i\overline \bfp(0)}{wN_i}\partial_v u + \sum_{i,j=1}^d \frac{Q_i\overline \bfp(0)}{wN_i}a_{ij}\partial_{v\mu_j}u + \frac{1}{2}\sum_{i,j=1}^d \frac{Q_i\overline \bfp(0)}{wN_i}a_{ij}\frac{Q_j\overline \bfp(0)}{wN_j}\partial_{vv}u\Big\} , 
\end{align*}
where $a = \sigma \sigma^\top$ and we omit function evaluations for brevity.
The translation invariance property of exponential utility allows us to write $u(t,\mu,v) = -\exp(-\gamma v +f(t,\mu))$ in terms of a reduced value function $f$. Substituting this form and solving for the optimal $Q^*$ in feedback form yields 
\begin{equation} \label{eqn:Q*}
    Q^*_i(t,\mu) = \frac{wN_i}{\overline \bfp(0)}\frac{1}{\gamma}\big(\partial_i f(t,\mu) + (a^{-1}b)_i(t,\mu) - C(t,\mu)\big), \qquad i=1,\dots,d 
\end{equation}
for an explicit scalar term $C(t,\mu)$ that enforces the self-financing condition. In the case that $a^{-1}b(t,\cdot)$ is the gradient of some function $g(t,\cdot)$, we see that $Q^*$ is additively functionally generated by $G(t,\mu) = \frac{1}{\gamma}(f(t,\mu)+g(t,\mu))$.\footnote{Models with the special structure that $a^{-1}b(t,\cdot)$ is a gradient appear frequently and have previously been studied in the literature. For instance, it is known that this is the structural form for worst-case models in robust growth-maximization problems \cite{kardaras2012robust,kardaras2021ergodic}.} Even beyond this special case, due to the appearance of the gradient of the reduced value function $f$, the strategy $Q^*$ shares similarities to additively generated portfolios. In particular, as the investor's wealth increases, the profits are reinvested in the market portfolio, leaving the magnitude of the residual trading activity unchanged --- a key property of CARA utility optimal portfolios. 

In the context of price impact, which we consider in the next section, this is a desirable property as the impact an investor generates scales with the trading size and intensity. If the holdings do not scale with wealth, then neither does the impact, which makes the analysis tractable. However, the holdings $Q^*$ do depend on the risk-aversion parameter $\gamma$, and in the context of utility maximization, the investor's choice of $\gamma$ regulates the trading intensity through \eqref{eqn:Q*}. In our context of additive functional generation and price impact, one can instead view scaling a generating function $G$ as tuning the typical trade volume so that the strategy operates within its capacity constraints. Indeed, as we will see in Section~\ref{sec:relarb} and, particularly, in Theorem~\ref{thm:relarb_explicit}, the trading speed constant $\nu$ introduced there plays an analogous role to the inverse risk-aversion parameter $1/\gamma$ here.

\subsection{Heuristic approach with price impact} \label{sec:heuristic}
We now return to the setting with price impact. The discussion in this section is heuristic. Letting $Q$ be given by \eqref{eqn:additive_holdings} leads from \eqref{eqn:relative_wealth} to the relative wealth dynamics
\begin{equation} \label{eqn:relative_wealth_gradient} 
dV^Q(t) = \sum_{i=1}^d \partial_i  G(t,\mu(t))d\mu_i(t) + \frac{1}{2}\sum_{i=1}^d \frac{\overline \bfp (0)}{w\overline \bfp(t)}d[I_i,Q_i](t), 
\end{equation}
where, as before, we used the fact that $\sum_{i=1}^d d\mu_i(t) = 0$ to eliminate the scalar term $V^Q(t)-\nabla G(t,\mu(t))^\top \mu(t)$. Using It\^o's formula for $G(t,\mu(t))$ we obtain 
    \begin{equation} \label{eqn:master_addtive_partial}
    dV^Q(t) = dG(t,\mu(t))- \partial_t G(t,\mu(t))dt - \frac12\sum_{i,j=1}^d \partial_{ij}G(t,\mu(t))d[\mu_i,\mu_j](t) + \frac{1}{2}\sum_{i=1}^d\frac{\overline \bfp (0)}{w\overline \bfp(t)}d[I_i,Q_i](t).
    \end{equation}
Using \eqref{eqn:impact} and \eqref{eqn:dJ}, we have that 
\begin{equation} \label{eqn:dIQ}
    d[I_i,Q_i](t) = \partial_xh_i(t,J_i(t))d[J_i,Q_i](t) = \partial_x h_i(t,J_i(t))K_i(t,t)d[Q_i](t).
\end{equation} From \eqref{eqn:additive_holdings} we see that
\begin{align*} 
    d[Q_i] & =\frac{w^2N_i^2}{\overline \bfp ^2(0)}d\big[\partial_iG(\cdot,\mu) + V^Q-\nabla G(\cdot,\mu)^\top \mu\big] = \frac{w^2N_i^2}{\overline \bfp^2 (0)}d\big[\partial_iG(\cdot,\mu) + G(\cdot,\mu)-\nabla G(\cdot,\mu)^\top \mu\big],
\end{align*}
where in the final equation we used that $V^Q - G(\cdot,\mu)$ is of finite variation by \eqref{eqn:master_addtive_partial}.
Hence, plugging into \eqref{eqn:master_addtive_partial} we obtain
\begin{equation} \label{eqn:additive_master_impact}
\begin{aligned}
    dV^Q(t) & = dG(t,\mu(t))- \partial_t G(t,\mu(t))dt - \frac{1}{2}\sum_{i,j=1}^d \partial_{ij}G(t,\mu(t))d[\mu_i,\mu_j](t) \\ 
     & \qquad + \frac{1}{2}\sum_{i=1}^d \frac{wN_i^2}{\overline \bfp(0)\overline \bfp(t)}\partial_xh_i(t,J_i(t))K_i(t,t)d\big[\partial_iG(\cdot,\mu) + G(\cdot,\mu)-\nabla G(\cdot,\mu)^\top \mu\big](t).
\end{aligned}
\end{equation}
The upshot is that $Q$ is also additively functionally generated in the general price impact setting. Indeed, we see that $V^Q$ satisfies \eqref{eqn:master_formula} with
\begin{equation} \label{eqn:Gamma}
\begin{aligned}
    \Gamma(t) & = -\int_0^t \partial_t G(s,\mu(s))ds - \frac{1}{2}\int_0^t\sum_{i,j=1}^d \partial_{ij}G(s,\mu(s))d[\mu_i,\mu_j](s) \\
& \qquad + \frac{1}{2}\int_0^t \sum_{i=1}^d \frac{wN_i^2}{\overline \bfp(0) \overline \bfp(s)}\partial_xh_i(s,J_i(s))K_i(s,s)d\big[\partial_iG(\cdot,\mu) + G(\cdot,\mu)-\nabla G(\cdot,\mu)^\top \mu\big](s).
\end{aligned} 
\end{equation}
Substituting the derived wealth \eqref{eqn:additive_master_impact} into the holdings equation \eqref{eqn:additive_holdings} establishes that \eqref{eqn:frictionless_additive_holdings} continues to hold, with $\Gamma$ now given by \eqref{eqn:Gamma} rather than by \eqref{eqn:Gamma_frict}.

 In comparison to the frictionless setting, the holdings process has additional terms that arise owing to the presence of price impact. Notably, these terms depend only on the historical market weight trajectory $\mu$ and the historical impact state trajectory $J$. In the case of linear price impact $h_j(t,x) = \lambda_j(t)x$, the dependence on $J$ vanishes. It is also worth emphasizing that although the first two terms in \eqref{eqn:Gamma} have the same functional form as those in \eqref{eqn:Gamma_frict} for the frictionless case, these will not have the same trajectories. This is because the market weight process $\mu$ differs from the frictionless case, as the investor's trades affect prices through the impact model. In Section~\ref{sec:numerics} we will see that, empirically, in the frictionless case the $\Gamma$ process tends to be larger than in the case with frictions despite the additional positive term appearing in the second line of \eqref{eqn:Gamma}.   

 Next we make two remarks concerning the importance of using the observed prices $P$ in the derivation of the master formula \eqref{eqn:master_formula}.

\begin{remark} \label{rem:master_fundamental}
    We highlight that to obtain the equation \eqref{eqn:master_addtive_partial} for the relative wealth, which is devoid of stochastic integration, from \eqref{eqn:relative_wealth_gradient} it was crucial to define $V^Q$ using the accounting wealth $W^Q$ and the numeraire $W^\mathcal{M}$.
    Indeed, if one instead uses the fundamental price $S$ to mark-to-market the investor's wealth, $\widetilde W^Q(t) = \sum_{i=1}^d Q_i(t)S_i(t)$, it is possible to show  (see for example \cite[Equation~(2.5)]{webster.23}) that the self-financing dynamics of $\widetilde W^Q$ are
    \[d\widetilde W^Q(t) = Q(t)^\top dS(t) - I(t)^\top dQ(t) - \frac{1}{2}d[I,Q](t).\]
    The additional stochastic integral $I(t)^\top dQ(t)$, absent from the accounting wealth dynamics \eqref{eqn:wealth}, is the key difficulty in directly establishing a master formula for the fundamental relative wealth process $\widetilde V^Q(t) = \widetilde W^Q/\widetilde W^\mathcal{M}$, where $\widetilde W^{\mathcal{M}} = \sum_{i=1}^d Q^\mathcal{M}_i(t)S_i(t)$ is the fundamental market wealth. An analogous calculation to the proof of Theorem~\ref{thm:relative_wealth} leads to the dynamics
    \begin{equation} \label{eqn:tildeVQ} 
    d\widetilde V^Q(t) = \sum_{i=1}^d \frac{Q_i(t)\overline \bfp(0)}{wN_i}d\mu^S_i(t) - \frac{1}{\widetilde W^\mathcal{M}(t)}I(t)^\top dQ(t) - \sum_{i=1}^d\frac{1}{2\widetilde W^\mathcal{M}(t)}d[I_i,Q_i](t),  
    \end{equation}
    where 
    \begin{equation} \label{eqn:boldS}
    \mu^S_i(t) = \frac{\bfs_i(t)}{\overline \bfs(t)}, \qquad \text{with } \qquad \bfs_i(t) = N_iS_i(t) \qquad \text {and} \qquad \overline \bfs(t) = \bfs_1(t) + \dots +\bfs_d(t) 
    \end{equation}
    for $i=1,\dots,d$ and  $t \geq 0$.
     Hence, even if one chooses $Q$ to be the gradient of a function of $\mu^S$ and applies It\^o's formula, akin to how we obtained \eqref{eqn:master_addtive_partial} from \eqref{eqn:relative_wealth_gradient}, the stochastic integral $I(t)^\top dQ(t)$ would remain, hindering the development of a master formula. This highlights the importance of using market weights derived from the observed prices $P$. It is for this reason that we work with the accounting wealth throughout this paper.
     \end{remark} 
     \begin{remark} \label{rem:target_holdings}
         Using the observed price $P$, rather than the fundamental price $S$, to specify the holdings $Q$ (through the observed market weights $\mu$) admits a clear financial interpretation. Since functionally generated portfolios systematically prescribe holdings that change with prices, the investor should take their own impact into account when making their trading decisions. This ensures that the target holdings are genuinely achieved and is analogous in spirit to the handling of proportional transaction costs in \cite[Section~2.1]{ruf2020impact} for their empirical experiments. Moreover, this point of view provides a financial explanation for why the master formula \eqref{eqn:additive_master_impact} can be derived, while using the fundamental price as in Remark~\ref{rem:master_fundamental} leads to the additional term $-\frac{1}{2\widetilde W^\mathcal{M}(t)}I(t)^\top dQ(t)$ in the relative wealth equation \eqref{eqn:tildeVQ}. Indeed, this extra term tracks the change in fundamental relative wealth arising from the systematic error of not reaching the target holdings when the investor ignores their own price impact.   
     \end{remark}
     
      Finally, we conclude this section by noting, for future reference, that the initial holdings can easily be deduced from \eqref{eqn:additive_holdings} and are given by
\begin{equation} \label{eqn:q0}
    q_i^0 := Q_i(0) = \frac{wN_i}{\overline \bfp (0)} \big(\partial_i G(0,\mu(0)) + 1 - \nabla G(0,\mu(0))^\top \mu(0)\big).
\end{equation}

\subsection{Dynamics of the price, holdings and impact state processes} \label{sec:rigorous}
The discussion in Section~\ref{sec:heuristic} was heuristic. Indeed, the existence of the price process $P$ and the semimartingality of the holdings process $Q$ were taken for granted. Unlike the frictionless case, the existence and semimartingality properties are not a priori guaranteed. Indeed, functionally generated strategies lead to holdings processes $Q$ that depend on the price process $P$ (through $\mu$). However, $Q$ affects the impact state $J$, which, in turn, affects $P$ through the price impact term \eqref{eqn:impact}, so that the three processes $(P,Q,J)$ are coupled. Additionally,  
% although in markets with frictions price processes need not be semimartingales to prohibit scalable arbitrage opportunities (see e.g. \cite{guasoni2015hedging}), 
the computations deriving the wealth equation \eqref{eqn:wealth} and the master formula \eqref{eqn:additive_master_impact} required $Q$ and $\mu$ to both be semimartingales. 

In this section, we state precise conditions on the generating function $G$ under which the computations in Section~\ref{sec:heuristic} are rigorous. We also provide the explicit coupled SDE \eqref{eqn:SDE} that the price process $P$, impact process $J$ and holdings process $Q$ satisfy. A formidable amount of notation needs to be introduced to compactly present the SDE. As such, to improve readability the lengthy derivations and proofs are contained in Appendix~\ref{app:SDE_derivation}. 

In the sequel we will make use of the following notation. Where not stated explicitly, all indices $i,j,k,\ell,m$ range from $1$ to $d$. For any vector $\varphi \in \R^d$, matrix $M\in \R^{d\times d}$, tensor $T \in \R^{d \times d \times d}$ and vector $x \in \R^d$, we use the notation $\overrightarrow{\varphi}^x$, $\overleftrightarrow M^x$ and $\overset{\Longleftrightarrow}{T^x}$ for the vector, matrix and tensor, respectively, with entries
\begin{align}
\overrightarrow \varphi_i^x & = \varphi_i - \varphi^\top x, \nonumber \\
 \overleftrightarrow {M}^{x}_{ij} & = M_{ij} - (Mx)_i - (Mx)_j + x^\top M x. \label{eqn:arrowM} \\ 
  \overset{\Longleftrightarrow}{T^x}_{ijk} & = T_{ijk} - \sum_{\ell=1}^d (T_{\ell j k} + T_{i\ell k} + T_{ij\ell})x_\ell + \!\sum_{\ell,m=1}^d (T_{i\ell m} + T_{\ell j m} + T_{\ell m k})x_\ell x_m - \!\sum_{\ell,m,n=1}^d T_{\ell m n}x_\ell x_mx_n. \nonumber
\end{align} 
For $p \in (0,\infty)^d$ we define the functions 
\[\overline \bfp (p) = N_1p_1+\dots+N_dp_d, \qquad \mu_i(p) = \frac{N_ip_i}{\overline \bfp(p)}.\]
Next, given initial wealth $w > 0$, an initial price $P(0) = p^0 \in (0,\infty)^d$ and  a tuple $(t,p,J) \in [0,\infty) \times (0,\infty)^d \times \R^d$, we set
\begin{align}
A^P_{ij}(t,p,J)  & = \delta_{ij} - \frac{wN_iN_j\partial_x h_i(t,J_i)K_i(t,t)}{\overline \bfp(p^0)\overline \bfp(p)}\overleftrightarrow{\nabla^2 G}^{\mu(p)}_{ij}(t,\mu(p)), \label{eqn:AP} \\
A^Q_{ij}(t,p) & = \frac{wN_iN_j}{\overline \bfp(p^0)\overline \bfp(p)}\overleftrightarrow{\nabla^2 G}^{\mu(p)}_{ij}(t,\mu(p)), \label{eqn:AQ}\\
& \hspace{-2.05cm}  B_{jk}(t,p,J) =  - 2\partial_{jk}G(t,\mu(p)) + \sum_{\ell=1}^d \frac{wN_\ell^2}{\overline \bfp(p^0)\overline \bfp(p)}\partial_x h_\ell(t,J_\ell) K_\ell(t,t)\big((\overrightarrow{\nabla \partial_jG})^{\mu(p)}_\ell(\overrightarrow{\nabla \partial_kG})^{\mu(p)}_\ell\big)(t,\mu(p)) , \nonumber \\
\Upsilon_{i\ell m}(t,p,J) & = \frac{wN_iN_\ell N_m}{\overline \bfp(p^0)\overline \bfp^2(p)}\Big(-\overleftrightarrow{\nabla^2 G}^{\mu(p)}_{im}(t,\mu(p)) + \frac{1}{2} \oset{\xLeftrightarrow{\qquad}}{\nabla^3 G^{\mu(p)}_{i\ell m}}(t,\mu(p)) + \frac{1}{2}\overleftrightarrow{B}^{\mu(p)}_{\ell m}(t,p,J)\Big). \label{eqn:upsilon}
\end{align}
We also recall the quantity $b_i^J(t,\omega)$ defined by \eqref{eqn:bJ}.
To formulate the coefficients that appear in the dynamics of $P$, we will need invertibility of the matrix $A^P$ given by \eqref{eqn:AP}. In the setting with (nontrivial) price impact, this is equivalent to a condition on $G$ that is akin to directional concavity. This is the content of the following proposition, whose proof is deferred to  Appendix~\ref{sec:proof_AP}.
\begin{prop}[Invertibility of $A^P$]
\label{prop:AP}Suppose $G(t,\cdot) \in C^2(\Delta^{d-1}_+)$ for every $t \geq 0$ and satisfies
\begin{equation}
\label{eqn:quadratic_form}
    x^\top \nabla^2 G(t,\mu) x \leq 0, \qquad \text{whenever} \qquad x^\top {\bf 1}_d = 0, \qquad \text{ for all } \quad t \geq 0, \quad \mu \in \Delta^{d-1}_+.
\end{equation} Then for any $(t,p,J) \in [0,\infty) \times (0,\infty)^d \times \R^d$, all of the eigenvalues of $A^P(t,p,J)$ are real and greater than or equal to one. In particular, $A^P(t,p,J)$ is invertible.

Conversely, suppose $G(t^*,\cdot)$ does not satisfy \eqref{eqn:quadratic_form} for some $t^* \geq 0$. If there exists a $J^* \in \R^d$ such that $\partial_x h_i(t^*,J^*_i) > 0$ for every $i=1,\dots,d$ then there exists a $p^* \in (0,\infty)^d$ such that $A^P(t^*,p^*,J^*)$ is not invertible.
\end{prop}
Clearly a sufficient condition for \eqref{eqn:quadratic_form} to hold is that $G(t,\cdot)$ be concave for every $t \geq 0$.
With the previous result in hand, we can define the SDE coefficients under the assumption that $G$ satisfies \eqref{eqn:quadratic_form}.
 Below we use the notation $u \circ v$ for the componentwise product of vectors, $(u\circ v)_i = u_iv_i$, and write $\mathrm{diag}(u)$ for the $d \times d$ matrix with the components of $u$ on the diagonal and zeros elsewhere. We also use the notation $u_{\cdot jk}$ for a $d$-dimensional vector with components $(u_{ijk};i=1,\dots,d)$. For $(t,p,J,\omega) \in [0,\infty) \times (0,\infty)^d \times \R^d \times (C(0,\infty);\R^d)$ we define
\begin{subequations} \label{eqn:SDE_coefficients}
\begin{align}
    &\beta^P(t,p,J)  && = (A^P)^{-1}(t,p,J), \\
    &\alpha^P(t,p,J,\omega)  &&= \beta^P(t,p,J)\Big(\partial_t h(t,J) + \partial_x h(t,J) \circ b^J(t,\omega)  \\
    &&&\qquad \qquad \quad 
    + \frac{w}{\overline \bfp(p^0)}\big(\partial_x h(t,J) \circ K(t,t) \circ N \circ (\overrightarrow{\nabla \partial_t G})^{\mu(p)}(t,\mu(p))\big)\Big), \\
    &\gamma^P_{\cdot jk}(t,p,J)  &&= \beta^P(t,p,J)\widetilde\gamma^P_{\cdot jk}(t,p,J), \\
    &\widetilde \gamma^P_{ijk}(t,p,J)   &&=  \partial_x h_i(t,J_i) K_i(t,t)\sum_{\ell,m=1}^d \Upsilon_{ilm}(t,p,J)\beta^P_{j\ell}(t,p,J)\beta^P_{km}(t,p,J) \\
     &&&\qquad + \frac{1}{2}\partial_{xx}h_i (t,J_i)K_i^2(t,t)(A^Q\beta^P)_{ij}(t,p,J)(A^Q\beta^P)_{i k}(t,p,J), \\
    &\beta^Q(t,p,J)  &&=  A^Q(t,p)\beta^P(t,p,J), \\
    &\alpha^Q(t,p,J,\omega)  &&= A^Q(t,p)\alpha^P(t,p,J,\omega) + \frac{w}{\overline \bfp(p^0)}N \circ (\overrightarrow{\nabla \partial_tG})^{\mu(p)}(t,\mu(p)), \\
    &\gamma^Q_{ijk}(t,p,J) && = (A^Q(t,p)\gamma^P_{\cdot jk}(t,p,J))_i  + \sum_{\ell,m=1}^d \Upsilon_{ilm}(t,p,J)\beta^P_{j\ell}(t,p,J)\beta^P_{km}(t,p,J),\\
   & \beta^J(t,p,J)  &&= \mathrm{diag}(K(t,t)) \beta^Q(t,p,J), \\
    &\alpha^J(t,p,J,\omega)  &&= \mathrm{diag}(K(t,t)) \alpha^Q(t,p,J,\omega) + b^J(t,\omega), \\
    &\gamma^J_{\cdot jk}(t,p,J)  &&= \mathrm{diag}(K(t,t)) \gamma_{\cdot jk}^Q(t,p,J).
\end{align}
\end{subequations}
The corresponding coupled SDE for $(P,Q,J)$ is then given by 
\begin{subequations} \label{eqn:SDE}
    \begin{align} 
 dP_i(t) &  = \alpha^P_i(t,P(t),J(t),Q^{[0,t]}_i)dt + \sum_{j=1}^d\beta^P_{ij}(t,P(t),J(t))dS_j(t) + \!\sum_{j,k=1}^d\!\gamma^P_{ijk}(t,P(t),J(t))d[S_j,S_k](t), \raisetag{0.3cm}\label{eqn:SDEP} \\
 dQ_i(t) &  = \alpha^Q_i(t,P(t),J(t),Q^{[0,t]}_i)dt + \sum_{j=1}^d\beta^Q_{ij}(t,P(t),J(t))dS_j(t) +\!\sum_{j,k=1}^d\!\gamma^Q_{ijk}(t,P(t),J(t))d[S_j,S_k](t),\raisetag{0.3cm}\label{eqn:SDEQ}  \\
 dJ_i(t) &  = \alpha^J_i(t,P(t),J(t),Q^{[0,t]}_i)dt + \sum_{j=1}^d\beta^J_{ij}(t,P(t),J(t))dS_j(t) + \!\sum_{j,k=1}^d \!\gamma^J_{ijk}(t,P(t),J(t))d[S_j,S_k](t). \raisetag{0.3cm}
\label{eqn:SDEJ}
\end{align}
\end{subequations}
We consider this SDE with initial condition $(P(0),Q(0),J(0)) = (p^0,q^0,0) \in (0,\infty)^d \times \R^d \times \R^d$, where $q^0_i$ is given by \eqref{eqn:q0},  consistent with the initial holdings derived in Section~\ref{sec:heuristic}.

We can now state our main theorem. The derivation of the SDE is contained in Appendix~\ref{app:SDE_derivation}, while the proof of this theorem is contained in Appendix~\ref{sec:main_proof}.
\begin{thm}[Additive portfolio generation with price impact] \label{thm:main} Let Assumption~\ref{ass:impact_inputs} hold and let $G \in C^{1,3}((0,\infty)\times \Delta^{d-1}_+)$ be given. Additionally assume that the third derivatives $\partial_{ijk}G$ are locally Lipschitz continuous and that the mixed derivatives $\partial_{it} G$ exist and are locally Lipschitz continuous. We further assume that $G(t,\cdot)$ satisfies \eqref{eqn:quadratic_form} for every $t \geq 0$. 
Then 
\begin{enumerate}
    \item \label{item:SDE} the SDE \eqref{eqn:SDE} has a pathwise unique strong solution on the stochastic time interval $[0,\xi)$, where  \begin{equation} \label{eqn:xi}
        \xi = \inf\{t \geq 0: \mu_i(t) = 0 \text{ for some } i \in \{1,\dots,d\} \text{ or } \overline \bfp(t) = 0\}
    \end{equation}
     is the explosion time,
    \item \label{item:right_process}the processes $P$, $J$ and $Q$ of part \ref{item:SDE} are the unique semimartingale processes satisfying $P(t) = S(t) + I(t)$ for $t \in [0,\xi)$, where $I$ is given by \eqref{eqn:impact} and $Q$ by \eqref{eqn:frictionless_additive_holdings} with $\Gamma$ given by \eqref{eqn:Gamma},
    \item \label{item:Q_additive} $Q$ is time-dependent additively functionally generated by $G$ on $[0,\xi)$ and its wealth process $V^Q$ satisfies \eqref{eqn:additive_master_impact},
    \item \label{item:long-only} if $G(\cdot, \mu)$ is nonincreasing for every $\mu \in \Delta^{d-1}_+$ and $G(t,\cdot)$ is concave for every $t \geq 0$, then the process $\Gamma$ of \eqref{eqn:Gamma} is nondecreasing. If, additionally, $G(0,\mu(p^0)) \leq 1$ and $G$ is nonnegative then on the lifetime $[0,\xi)$ of the solution, $V^Q$ is nonnegative and $Q_i$ is nonnegative for every $i=1,\dots,d$,
    \item \label{item:positive_price} 
    Suppose $G$ is time-homogeneous and concave. Set $F_i(\mu) = \partial_i G(\mu) + G(\mu) - \nabla G(\mu)^\top \mu$ and $\overline F_i = \sup_{\mu \in \Delta^{d-1}_+}F_i(\mu)$. If for every $i \in \{1,\dots,d\}$ we have
    \begin{equation} \label{eqn:positive_condition}
    \lim_{n \to \infty} \mu_i^n = 0 \implies 
         \lim_{n \to \infty} F_i(\mu^n) = \overline F_i
    \end{equation} for any sequence $(\mu^n)_{n \in \N} \subset \Delta^{d-1}_+$ then $\xi = \inf\{t \geq 0: \overline \bfp(t) = 0 \text{ or } \overline \bfp(t) = \infty\}$.
\end{enumerate}
\end{thm} The examples introduced in Section~\ref{sec:frictionless} all satisfy the assumptions of Theorem~\ref{thm:main}. Moreover, the diversity-$p$, entropy and geometric weighted mean functions additionally satisfy \eqref{eqn:positive_condition}.
Several remarks are in order.
\begin{remark}
    The pathwise unique strong solution to \eqref{eqn:SDE} is entirely driven by the fundamental price $S$. That is, once a generating function $G$ is chosen, the processes $P,Q$ and $J$ can be viewed as functions of the path of $S$. This representation is particularly useful for backtests and trading simulations. Indeed, given a (discretized) trajectory of $S$, the Euler discretization of \eqref{eqn:SDE} provides a straightforward way to simulate the evolution of the price, impact and holdings processes and, as a result, also the investor's relative wealth process $V^Q$ given by \eqref{eqn:additive_master_impact}. We present the results of empirical backtest experiments in Section~\ref{sec:numerics}.
\end{remark}

\begin{remark}
    In the case of time-independent linear impact, $h_i(t,x) = \lambda_ix$ for every $i$, the SDE coefficients \eqref{eqn:SDE_coefficients} become independent of $J$. As such,  the SDE reduces to a $2d$-dimensional system \eqref{eqn:SDEP}--\eqref{eqn:SDEQ} for $P$ and $Q$, with $J$ given by the right-hand side of \eqref{eqn:SDEJ}.
\end{remark}
\begin{remark} \label{rem:C3}
    Depending on the choice of function $G$, the explosion time $\xi$ for the SDE \eqref{eqn:SDE} may not be maximal. For example, if $G(t,\cdot)$ is a function which is actually $C^{3}$  with Lipschitz continuous derivatives on all of $\R^d$, not just on $\Delta^{d-1}_+$, then the SDE solution can continue to exist even if $P_i$ (and hence $\mu_i$) becomes negative. This process, however, loses the interpretation of being a stock price, so we stick with the explosion time $\xi$ given by \eqref{eqn:xi}. 
\end{remark}
    By examining the SDE dynamics in \eqref{eqn:SDE}, as well as the computations leading to its derivation contained in Appendix~\ref{app:SDE_derivation}, it is clear that if the set $\{(t,p,J): A^P(t,p,J) \text{ is not invertible}\}$ is nonempty, which by Proposition~\ref{prop:AP} is equivalent to $G$ failing to satisfy \eqref{eqn:quadratic_form}, then the SDE coefficients would not be defined everywhere. Even if they are defined almost everywhere, they would no longer be locally Lipschitz continuous, leading to difficulties studying the well-posedness of the SDE \eqref{eqn:SDE}. Moreover, if this set of noninvertible points is sufficiently large and visited by the process with positive probability, then it is clear that the price process $P$ cannot be a semimartingale. This is in contrast to the frictionless setup, where generic $C^2$ generating functions are easily handled.

Regarding smoothness requirements, compared to the frictionless setting, additional regularity on $G$ is required. This is because the final term of \eqref{eqn:Gamma}, which does not appear in the frictionless case, contains the quadratic variation of $\partial_i G(\cdot,\mu)$.  The assumption that $\partial_i G$ is a $C^{1,2}$ function guarantees semimartingality of $\partial_i G(\cdot,\mu)$ and hence the existence of its quadratic variation. Furthermore, assuming the third derivatives are locally Lipschitz continuous allows us to establish well-posedness of the SDE \eqref{eqn:SDE}.  

In the sequel, to simplify the exposition and avoid making statements up to an explosion time, we will at times invoke the following assumption.
\begin{assum}[Nonexplosive market] \label{ass:nonexplosion} In the setup of Theorem~\ref{thm:main} assume that $\P(\xi < \infty) = 0$.
\end{assum}
We briefly mention a few additional sufficient conditions that ensure Assumption~\ref{ass:nonexplosion} holds in the case that $G$ satisfies \eqref{eqn:positive_condition}. If $h_i(t,x) = 0$ for some $i$, then $\overline \bfp(t) \geq S_i(t) > 0$. Additionally, if each $h_i$ is bounded by some constant $C_i$ then $\overline \bfp(t) = \sum_{i=1}^d N_i S_i(t) + \sum_{i=1}^d N_ih_i(t,J_i(t)) \leq \sum_{i=1}^d N_i S_i(t) + \sum_{i=1}^d N_iC_i < \infty$. We will not impose these stronger conditions directly, but will instead enforce the weaker Assumption~\ref{ass:nonexplosion} where needed, specifying each some instance explicitly. In Section~\ref{sec:relarb_explicit} we will establish that Assumption~\ref{ass:nonexplosion} holds under additional assumptions on the frictionless market and by imposing certain bounds on $G$ and its derivatives. Studying the SDE \eqref{eqn:SDE} and understanding whether a market explosion is possible in the general case considered in this section remain interesting problems beyond the scope of this paper.

\subsection{Multiplicative generation} \label{sec:multiplicative}
In this section we discuss multiplicative functional generation. As the name suggests, multiplicative functional generation is also a systematic way to prescribe the investor's holdings through a generating function, but, differently from additive generation, the natural parametrization of the strategy is through the proportions of wealth $\pi$ invested in the assets, rather than the number of shares $Q$. Restricting to strategies with strictly positive wealth, these are in one-to-one correspondence via the relationships
\begin{equation} \label{eqn:pi_Q_relationship}
\pi_i(t) = \frac{Q_i(t)P_i(t)}{W^Q(t)} \iff Q_i(t) = \frac{\pi_i(t)W^Q(t)}{P_i(t)}.
\end{equation} 
In the frictionless case, for a sufficiently smooth and positive generating function $G$, the following (time-homogeneous) master formula holds (see \cite[Theorem~3.1.5]{fernholz2002stochastic}),
\begin{equation} \label{eqn:master_mult}
\log V^Q(t) = \log G(\mu(t)) - \log G(\mu(0)) + \Theta(t),
\end{equation}
where $d\Theta(t) = -\frac{1}{2}\sum_{i,j=1}^d \frac{\partial_{ij} G(\mu(t))}{G(\mu(t))}d[\mu_i,\mu_j](t)$ and the strategy is given by
\begin{equation} \label{eqn:pi_fg}
\pi_i(t) = \mu_i(t)\big(\partial_i \log G(\mu(t)) + 1 - \nabla \log G(\mu(t))^\top \mu(t)\big), \qquad i=1,\dots,d.
\end{equation}
Multiplicative generation leads to a master formula decomposition for the logarithmic wealth, rather than the wealth itself. This reflects the fact that the proportions of wealth, rather than the number of shares,  are systematically specified by a function of the market weights. As in the discussion of Section~\ref{sec:interpreting_additive} for additive generation, multiplicatively generated portfolios can be shown to arise from utility maximization, where the utility function is of constant relative risk aversion (CRRA) type; i.e.,  power or logarithmic utility. In particular, as the investor's wealth increases, multiplicative generation prescribes increasingly large absolute positions $Q$ and, consequently, increasingly large trades.  

This behaviour is in contrast to additive generation, as discussed in Section~\ref{sec:interpreting_additive}, and leads to complications in establishing an analogue of \eqref{eqn:master_mult} in the price impact setting. Indeed, we will heuristically show that the relationship \eqref{eqn:master_mult} can be derived with price impact present, but the resulting process $\Theta$ depends on $V^Q$ so that the right-hand side does not yield a closed-form formula for the left-hand side. 

To see this, we let $\pi$ be given by \eqref{eqn:pi_fg}, deduce $Q$ from \eqref{eqn:pi_Q_relationship}, and substitute into \eqref{eqn:relative_wealth} to obtain, after some simplification, 
\[\frac{dV^Q(t)}{V^Q(t)} = \nabla \log G(\mu(t))^\top d \mu(t) + \frac{1}{2}\sum_{i=1}^d \frac{\overline \bfp(0)}{w\overline \bfp(t)}d[I_i,Q_i](t).\] Using It\^o's formula for $\log V^Q$ on the left-hand side and $\log G(\mu)$ on the right-hand side leads to the relationship
\begin{equation} \label{eqn:master_mult_implicit}
d\log V^Q(t) = d\log G(\mu(t)) - \frac{1}{2}\sum_{i,j=1}^d \frac{\partial_{ij}G(\mu(t))}{G(\mu(t))}d[\mu_i,\mu_j](t) + \frac{1}{2}\sum_{i=1}^d \frac{\overline \bfp(0)}{w\overline \bfp(t)}d[I_i,Q_i](t).
\end{equation}
As in the additive case, we now seek to simplify the quadratic covariation term. Using the identity \eqref{eqn:dIQ} it suffices to compute $d[Q_i](t)$. Setting $F_i(\mu) = 1 + \partial_i \log G(\mu) - \nabla \log G(\mu)^\top \mu$, we directly compute from \eqref{eqn:pi_Q_relationship} and \eqref{eqn:pi_fg} that
\[d[Q_i](t) = d\Big[\frac{\mu_i}{P_i} W^QF_i(\mu)\Big](t) = d\Big[N_i \frac{W^Q}{\overline \bfp}F_i(\mu)\Big](t) = d\Big[\frac{wN_i}{\overline \bfp(0)}V^QF_i(\mu)\Big](t) =\frac{w^2N_i^2}{\overline \bfp^2(0)}d[V^Q F_i(\mu)](t),\]
where in the third equality we recalled the form of the market wealth process given in \eqref{eqn:market_wealth}. Next, we can use the product rule to simplify the remaining quadratic variation term,
\begin{align*}
    d[V^QF_i(\mu)](t) & = (V^Q(t))^2d[F_i(\mu)](t) + 2V^Q(t)F_i(\mu(t))d[V^Q,F_i(\mu)](t) + F_i^2(\mu(t))d[V^Q](t) \\
    & = (V^Q(t))^2\big(d[F_i(\mu)](t) + 2F_i(\mu(t))d[\log V^Q,F_i(\mu)](t) + F_i^2(\mu(t))d[\log V^Q](t)\big) \\
    & = (V^Q(t))^2\big(d[F_i(\mu)](t) + 2F_i(\mu(t))d[\log G(\mu),F_i(\mu)](t) + F_i^2(\mu(t))d[\log G(\mu)](t)\big) \\
    & = (V^Q(t))^2d\Big[F_i(\mu) + \int_0^\cdot F_i(\mu)d\log G(\mu)\Big](t).
\end{align*}
where in the second equality we used the fact that $d[\log V^Q]= d[V^Q]/(V^Q)^2$; in the third equality we deduced from \eqref{eqn:master_mult_implicit} that $\log V^Q$ is given by $\log G(\mu)$ plus a finite variation term, which does not contribute to the quadratic variation; and in the final equality we used the properties of quadratic variation to write the terms more compactly. Combining these computations yields the relationship
\begin{equation} \label{eqn:master_mult_explicit}
\begin{aligned}  d\log V^Q(t) & = d\log G(\mu(t)) - \frac{1}{2}\sum_{i,j=1}^d \frac{\partial_{ij}G(\mu(t))}{G(\mu(t))}d[\mu_i,\mu_j](t)\\
& \quad  + \frac{(V^Q(t))^2}{2}\sum_{i=1}^d \frac{\overline \bfp(0)}{w\overline \bfp (t)}\partial_xh_i(t,J_i(t))K_i(t,t)d\Big[F_i(\mu) + \int_0^\cdot F_i(\mu)d\log G(\mu)\Big](t). 
\end{aligned}
\end{equation}
The first two terms coincide with the frictionless master formula \eqref{eqn:master_mult}, while the final term is new and accounts for the distortion in wealth caused by the presence of price impact. However, unlike the case of additive generation, in the multiplicative case this additional term scales with $V^Q$, reflecting the fact that trade sizes, and hence price impact, increase with one's wealth for multiplicatively generated portfolios. 

The presence of $V^Q$ on the right-hand side of \eqref{eqn:master_mult_explicit} makes the rigorous study of multiplicative generation substantially more difficult in the price impact setting. Indeed, it is far from clear whether the well-posedness of the coupled system $(P,Q,J,V^Q)$, analogous to what was established in Theorem~\ref{thm:main} for the additive case, can be established. From the mathematical point of view, one might expect the coupled system to exist locally, but for sufficiently large values of $V^Q$ the investor's trading volume would give rise to a large price distortion and could lead to finite-time blowup in the equations. In essence, for additively generated portfolios the capacity constraints of the strategy are handled mechanically due to profits being reinvested in the market portfolio as discussed in Section~\ref{sec:interpreting_additive}. This is not the case for multiplicative generation. For this reason our main focus in this paper is on additive generation, which allows us to develop a clean and rigorous theory of relative arbitrage in the next section. The study of multiplicative generation with price impact is left as an open problem for future research.

\section{Relative arbitrage} \label{sec:relarb}
In this section we discuss relative arbitrage in price impact models. That is, we investigate whether or not outperformance of the market portfolio \emph{with probability one} over a sufficiently long time horizon is possible. As previously discussed in Section~\ref{sec:relative_wealth}, care is needed when comparing wealth processes of two portfolios to ensure any outperformance is genuine and not attributable to inflated mark-to-market prices caused by the investor's own trades. This motivates the following definition of relative arbitrage

\begin{defn}[Relative arbitrage] \label{defn:relarb} For a time horizon $T > 0$, a trading strategy $Q$ is called a \emph{relative arbitrage} with respect to the market portfolio $Q^\mathcal{M}$ on $[0,T]$ if the following two conditions hold.
\begin{enumerate}
\item \label{item:liquidate}
$Q_i(0) = Q^\mathcal{M}_i(0) = \frac{wN_i}{\overline \bfp(0)}$ and  $Q_i(t) = V^Q(T)Q^\mathcal{M}_i(t) = \frac{W^Q(T)N_i}{\overline \bfp(T)}$ for all $i\in\{1,\dots,d\}$ and $t \geq T$,
\item \label{item:relarb} $\P(V^Q(T) \geq 1) = 1$ and  $\P(V^Q(T) >1) > 0$.
\end{enumerate}
If $\P(V^Q(T) > 1) = 1$ then we call $Q$ a \emph{strong relative arbitrage}.
\end{defn}
Condition \ref{item:liquidate} requires the strategy $Q$ to initially hold the market portfolio and then to hold it again from time $T$ onward. By this condition,  we have the equality of events,
\begin{align*} 
\{V^Q(T) \ge  1\} & = \Big\{V^Q(T) \sum_{i=1}^d Q^\mathcal{M}_i(T)S_i(T) \geq  \sum_{i=1}^d Q^\mathcal{M}_i(T)S_i(T)\Big\} \\
& = \Big\{\sum_{i=1}^d Q_i(T)S_i(T) \geq \sum_{i=1}^d Q^\mathcal{M}_i(T)S_i(T)\Big\}  = \{\widetilde W^Q(T) \geq  \widetilde W^\mathcal{M}(T)\},
\end{align*}
 where $\widetilde W^Q$ and $\widetilde W^\mathcal{M}$, introduced in Remark~\ref{rem:master_fundamental}, denote the investor's fundamental wealth when trading $Q$ and the market portfolio respectively. The strict inequality case $\{V^Q(T) > 1\} = \{\widetilde W^Q(T) > \widetilde W^\mathcal{M}(T)\}$ follows analogously. As such, we see that outperformance measured with respect to the accounting wealth, as in item \ref{item:relarb}, implies genuine fundamental outperformance as long as the relative wealth is measured after liquidation, as required by item \ref{item:liquidate}.  
 In essence, a relative arbitrage is a trading strategy which, before initiation and after liquidation, holds the benchmark market portfolio and upon liquidation is guaranteed to outperform the market portfolio. In the case of frictionless markets the first condition \ref{item:liquidate} need not be explicitly stated, as there is a unique undistorted price with which all portfolios are valued.
\subsection{Implicit conditions for relative arbitrage}

In \cite[Chapter~3]{fernholz2002stochastic}, it was shown that under the conditions of \emph{market nondegeneracy} and \emph{market diversity}, relative arbitrage exists in a frictionless market.\footnote{For other sufficient conditions and sharp time horizon bounds for relative arbitrage see \cite{fernholz2018volatility,larsson2021relative}.} We state these conditions here and show that they imply relative arbitrage in the present setting of price impact as well.
\begin{defn}[Market nondegeneracy and market diversity] \label{defn:market_nondegen_diverse}
Given a stopping time $T > 0$ the market is called
    \begin{enumerate}
        \item \emph{nondegenerate} on $[0,T]$ with nondegeneracy constant $\epsilon > 0$ if the log price has absolutely continuous quadratic variation process, $d[\log P](t) = a(t)dt$, and $a(\cdot)$ satisfies the uniform ellipticity condition,
        \[x^\top a(t) x \geq \epsilon \|x\|^2, \qquad \P\text{-a.s.}\]
        for every $t \leq T$ and every $x \in \R^d$,
        \item \emph{diverse} on $[0,T]$ with diversity constant $\delta > 0$ if \[\max_{i\in \{1,\dots,d\}}\mu_i(t) \leq 1-\delta, \qquad \P\text{-a.s.}\]
        for every $t \in [0,T]$.\end{enumerate}
        If $T = \infty$, we simply say the market is nondegenerate and diverse respectively.
\end{defn}
    In nondegenerate and diverse markets, relative arbitrage opportunities can be constructed using functionally generated portfolios. Here we will construct relative arbitrage using a time-dependent additively generated portfolio. Unlike the frictionless setting, where changes in holdings are cost-free, care is needed to systematically initiate and liquidate in a way that controls the evolution of the relative wealth process. 
    
    To facilitate this, given a target time horizon $T > 0$ and time-homogeneous generating function $H(\mu)$, we will use the time-dependent generating function $G(t,\mu) = \psi(t)H(\mu)$, where $\psi$ is a $C^1$ function that satisfies $\psi(0) =\psi(T) = 0$, $\psi(t) = 1$ for $t \in [T_0,T_1]$, $\psi'(t) > 0$ for  $t \in [0,T_0]$ and $\psi'(t) < 0$ for $t \in [T_1,T]$ with $0 < T_0 \leq T_1 < T < \infty$. The period $[0,T_0]$ is the \emph{initiation period}, where the investor gradually goes from holding the market portfolio to the target portfolio generated by $H$, while the period $[T_1,T]$ is the \emph{liquidation period}, where the investor gradually trades back to the market portfolio. The period $[T_0,T_1]$ will then be referred to as the \emph{main trading period}. For concreteness, given times $T_0,T_1$ and $T$, we use the following choice of $\psi$,
\begin{equation} \label{eqn:psi}
    \psi(t) = 
    \begin{cases}
    \chi\Big(\frac{t}{T_0}\Big), & 0 \leq t \leq T_1,\\
    1-\chi\Big(\frac{t-T_1}{T-T_1}\Big), & t \geq T_1,
    \end{cases}\qquad \text{where} \qquad \chi(u) = \begin{cases}
        3u^2-2u^3, & u \leq 1, \\
        1, & u \geq 1.
    \end{cases}
\end{equation}
A schematic representation of $\psi$, along with its derivative, is depicted in the left panel of Figure~\ref{fig:psi}. Typically, we will take $T_0,T_1$ and $T$ to be deterministic constants, but in Theorem~\ref{thm:relarb_implicit} we take $T_1$ to be a bounded stopping time.

We now describe how the diversity term $G(t,\mu(t)) = \psi(t)H(\mu(t))$ and the $\Gamma(t)$ term of the master formula decomposition \eqref{eqn:master_formula} evolve when using such a strategy. The right panel of Figure~\ref{fig:psi} illustrates their behaviour when using the quadratic generating function on real data in 2016; see Section~\ref{sec:numerics} for a more detailed description of the dataset and specifically chosen parameters.  Over the initiation period $[0,T_0]$, the diversity term $G(t,\mu(t))$ rapidly increases from its initial value of $G(0,\mu(0)) = 0$ to the positive value $H(\mu(T_0))$.  To compensate for this rapid increase, the term $\Gamma(t)$ decreases rapidly over $[0,T_0]$ with the term $-\int_0^t \partial_t G(s,\mu(s))ds$ being the dominant contribution in \eqref{eqn:Gamma} over this period.  Over the main trading period $[T_0,T_1]$ the diversity process $H(\mu(t))$ experiences fluctuations due to fundamental price shocks and, to a lesser extent, the  trades placed by the investor since they impact the market weight process $\mu$. Meanwhile, the $\Gamma$ process is increasing over the interval $[T_0,T_1]$ as the first term on the right-hand side of \eqref{eqn:Gamma} is zero and the remaining terms are increasing. Empirically, the main contribution in this period comes from the second derivative term in \eqref{eqn:Gamma} with the price impact term, which appears on the second line of \eqref{eqn:Gamma}, contributing only a small percentage --- typically at most 5\% of the total increase in $\Gamma$ over the main trading period. However, the total accumulation of the $\Gamma$ process tends to be smaller when price impact is present than in the frictionless case. We defer further discussion of this point to Section~\ref{sec:backtest_results}, where we also depict the growth of the $\Gamma$ processes over the main trading period, $(\Gamma(t) - \Gamma(T_0); t \in [T_0,T_1])$, across different price impact assumptions in the top right panels of Figures~\ref{fig:Quadratic} and \ref{fig:Entropy}. Finally, over the liquidation period $[T_1,T]$ the diversity and $\Gamma$ processes reverse course with $G(t,\mu(t))$ rapidly decreasing to $G(T,\mu(T)) = 0$, while the $\Gamma$ process increases monotonically to compensate. 

\begin{figure}
    \centering
    \includegraphics[width=\linewidth]{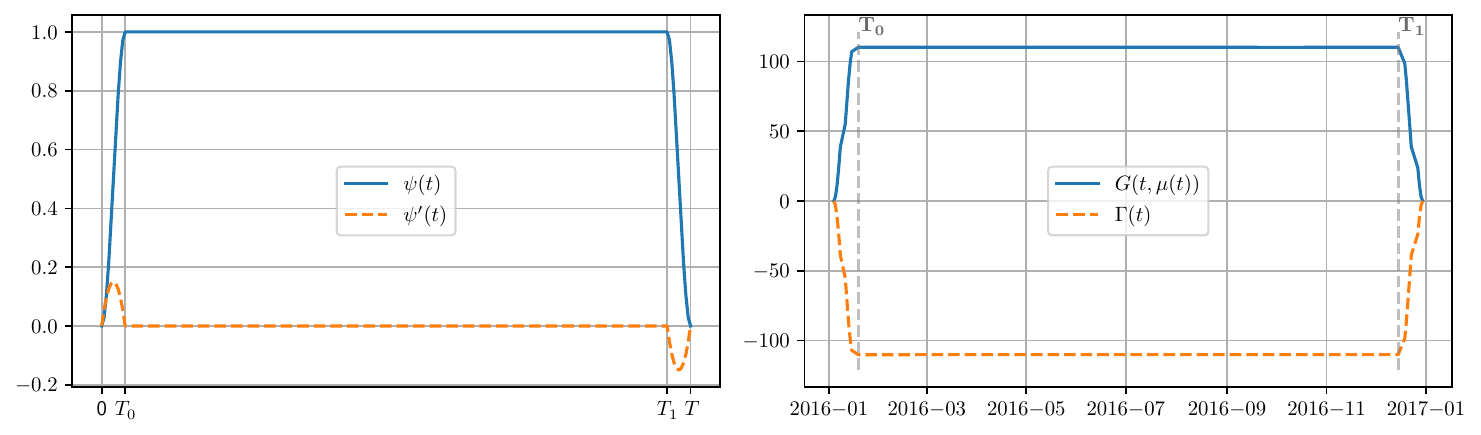}
    \caption{The left panel displays $\psi(t)$ and $\psi'(t)$ where $T_0 = 10$, $T_1 = 242$ and $T = 252$ are in units of trading days. The right panel displays the evolution of the diversity term $G(t,\mu(t))$ and the $\Gamma(t)$ term of the master formula decomposition \eqref{eqn:master_formula} during the 2016 backtest period when using the quadratic generating function.}
    \label{fig:psi}
\end{figure}

We are now ready to state a first result on relative arbitrage in the setting with price impact.
\begin{thm}[Relative arbitrage; implicit version] \label{thm:relarb_implicit} Let Assumption~\ref{ass:impact_inputs} hold. Set
\[
    G(t,\mu) = \psi(t)H(\mu), 
\]
where $H\in C^{3}(\Delta^{d-1}_+)$ is a positive, strongly concave function with locally Lipschitz third derivatives, bounded from above by a constant $\overline H > 0$
and $\psi$ is given by \eqref{eqn:psi}. We assume the resulting market satisfies Assumption~\ref{ass:nonexplosion}, is nondegenerate and diverse, with respective constants $\epsilon$ and $\delta$ independent of the time horizon.

Then there exists a strong relative arbitrage on $[0,T]$ for any $T > T^* := \frac{2d\overline H}{m\epsilon \delta^2}$, where $m$ is the strong concavity constant associated with $H$; that is,  $-\nabla^2 H(\mu) \succeq m I_{d \times d}$ for every $\mu \in \Delta^{d-1}_+$, in the sense of positive-definite matrices. Concretely, letting $T_0 > 0$ be arbitrary and setting $T_1 = T^* + T_0$ and $T = T_1 + T_0$, we have that the time-dependent additively generated strategy with generating function $G$ is a strong relative arbitrage on $[0,T]$.
\end{thm}
\begin{proof}
    We start by noting that $G$ satisfies the assumptions of Theorem~\ref{thm:main}.
    Theorem~\ref{thm:main} together with Assumption~\ref{ass:nonexplosion} guarantees that this function $G$ gives rise to valid and globally defined price, holdings and impact states processes $(P,Q,J)$ given by the SDE \eqref{eqn:SDE}. Moreover, $Q$ is time-dependent additively functionally generated by $G$ and the wealth process $V^Q$ satisfies \eqref{eqn:master_formula} with $\Gamma$ given by \eqref{eqn:Gamma}.
    Next, note that $G(\cdot,\mu)$ is nondecreasing on $[0,T_1]$ and nonincreasing on $[T_1,T]$ for each $\mu \in \Delta^{d-1}_+$. As such, we see from \eqref{eqn:Gamma} that $\Gamma$ is strictly increasing on $[T_0,T]$ and the only component that contributes to a decrease on $[0,T_0]$ is the term $t \mapsto -\int_{0}^t\partial_t G(s,\mu(s))ds$. Using these observations together with the fact that $G(0,\mu(0)) = G(T,\mu(T)) = 0$, we obtain the bound
    \begin{align}
        V^Q(T)  = 1+ \Gamma(T) & > 1-  \int_{0}^{T_0}\partial_t G(t,\mu(t))dt - \frac{1}{2}\sum_{i,j=1}^d\int_{T_0}^{T_1}\partial_{ij}G(t,\mu(t))d[\mu_i,\mu_j](t) \nonumber \\
        & = 1 -\int_{0}^{T_0}\psi'(t)H(\mu(t))dt -\frac{1}{2}\sum_{i,j=1}^d\int_{T_0}^{T_1}\partial_{ij}H(\mu(t))d[\mu_i,\mu_j](t) \nonumber  \\
        & \geq 1 - \overline H + \frac{m}{2}\sum_{i=1}^d\int_{T_0}^{T_1}d[\mu_i](t).  \label{eqn:VQ_lower_bound}
    \end{align}
In the final inequality we used the bound $H \leq\overline H$ together with $\psi(0) = 0$ and $\psi(T_0) = 1$ to bound the time integral, and strong concavity of $H$ to bound the quadratic variation term. 

It now remains to show that $ \frac{m}{2}\sum_{i=1}^d\int_{T_0}^{T_1}d[\mu_i](t) \geq \overline H$ with probability one. Applying It\^o's formula to $\mu_i$, it is easily verified that 
\[
    d[\mu_i](t) = \frac{1}{\overline \bfp^2(t)}(d[\bfp_i](t) - 2\mu_i d[\bfp_i,\overline \bfp](t) + \mu_i^2d[\overline \bfp](t)).
\] 
It then follows that
\begin{align*}
\frac{m}{2}\sum_{i=1}^d\int_{T_0}^{T_1}d[\mu_i](t) &  = \frac{m}{2}\sum_{i,j=1}^d \int_{T_0}^{T_1}\frac{1}{\overline \bfp^2(t)}(\delta_{ij} - 2\mu_i(t) + \|\mu(t)\|^2)d[\bfp_i,\bfp_j](t) \\
& =  \frac{m}{2}\sum_{i,j=1}^d\int_{T_0}^{T_1} \mu_i(t)\mu_j(t)(\delta_{ij} - 2\mu_i(t) + \|\mu(t)\|^2)\frac{d[\bfp_i,\bfp_j](t)}{\bfp_i(t)\bfp_j(t)} \\
& =  \frac{m}{2}\int_{T_0}^{T_1}\mathrm{Tr}\big(\Psi(\mu(t))d[\log \bfp ](t)\big),
\end{align*}
% where $\Psi_{ij}(t) = \mu_i(t)(\delta_{ij} -\mu_j(t))$ 
where $\Psi_{ij}(\mu) = \mu_i\mu_j(\delta_{ij}-2\mu_i+\|\mu\|^2)$
and $\delta_{ij}$ is the Kronecker delta. By market nondegeneracy we have that $d[\log \bfp](t) = d[\log P](t) \succeq \epsilon I_{d \times d} dt$. Hence, recalling the definition of $T_1$ and $T^*$ we obtain
\begin{equation}
\begin{aligned} \label{eqn:QV_estimates}
    \frac{m}{2}\sum_{i=1}^d\int_{T_0}^{T_1}d[\mu_i](t) & \geq \frac{m\epsilon}{2}\int_{T_0}^{T_1}\mathrm{Tr}\big(\Psi(\mu(t))\big)dt = \frac{m\epsilon}{2}\int_{T_0}^{T_1}\sum_{i=1}^d\mu_i^2(t)(1-2\mu_i(t)+\|\mu(t)\|^2)dt \\
    & \geq \frac{m\epsilon}{2}\sum_{i=1}^d \int_{T_0}^{T_1}\mu_i^2(t)(1-\mu_i(t))^2dt \geq \frac{m\epsilon\delta^2}{2d}(T_1-T_0) = \overline H,
\end{aligned}
\end{equation}
where in the final inequality we used the market diversity property $1-\mu_i(t) \geq \delta$ and the bound $\sum_{i=1}^d \mu_i^2(t) \geq 1/d$, the latter following from the Cauchy--Schwarz inequality. This completes the proof.
\end{proof}
\begin{remark}
    The proof only required diversity and nondegeneracy to hold up to the time $T_1$.
\end{remark}

Two examples of generating functions $H$ satisfying the conditions of the theorem are the quadratic function $H(\mu) =  1 - \frac{1}{2}\sum_{i=1}^d \mu_i^2$ and the entropy function $H(\mu) = -\sum_{i=1}^d \mu_i \log \mu_i$. In Section~\ref{sec:numerics} we investigate the performance of these two generating functions on historical data.

\subsection{Explicit conditions for relative arbitrage} \label{sec:relarb_explicit}
In the previous section we obtained relative arbitrage in the market with price impact, but under \emph{implicit} assumptions on the market. Indeed, we directly assumed that the observed market was nondegenerate, diverse and nonexplosive. However, in our framework, the market with price impact is derived from the  frictionless market prices $S$, the price impact specifications $h$ and $K$, and the choice of generating function $G$. In this section we develop explicit conditions on these inputs that guarantee the existence of relative arbitrage. To this end, we make the following definitions.  
\begin{defn}[Bounded volatility]
The market is said to have \emph{bounded volatility} if the almost sure bound $\sum_{i,j=1}^d [\log P_i,\log P_j](t) \leq \sigma^2 t$ holds
for every $t \geq 0$ and some $\sigma^2 >0$.
\end{defn}

\begin{defn}[Floor constraints] \label{defn:floor_constraints}
\begin{enumerate} 
    \item \! The total market capitalization satisfies a \emph{floor constraint} with bound $\kappa > 0$ on $[0,T]$ if $\overline \bfp(t) \geq \kappa$ almost surely for every $t \in [0,T]$,
    \item The market weights are said to satisfy a \emph{floor constraint} with bound $\ell > 0$ on $[0,T]$ if $\min_{i\in\{1,\dots,d\}}\mu_i(t) \geq \ell$ almost surely for every $t \in [0,T]$.
\end{enumerate}
\end{defn}
\begin{remark} \label{rem:floor_bounds} \begin{enumerate}
\item Most markets have a minimum capitalization requirement for the continued listing of a company's stock. For instance, as of January 2026, NASDAQ requires companies to have at least \$15M in the market value of their publicly held shares and a minimum bid price of \$1 per share \cite{nasdaq_continued_listing_2026}. As such, the floor constraints of Definition~\ref{defn:floor_constraints} are expected to hold in practice,
    \item  When $d = 2$, market diversity implies a market weight floor constraint since $\min_{i\in\{1,2\}}\mu_i(t) = 1 - \max_{i \in \{1,2\}}\mu_i(t) \geq \delta$, where $\delta$ is the diversity constant.
    
\end{enumerate}
   
\end{remark}
In the sequel, we will say that the \emph{frictionless market} is diverse, nondegenerate, has bounded volatility, or satisfies a floor constraint if the corresponding definition holds with $S$ in place of $P$. 
To simplify the exposition of what is to come, we will assume that the impact functions are of separable form \begin{equation} \label{eqn:h_separable}
    h_i(t,x) = \phi_i(t)\widehat h_i(x), \qquad i=1,\dots,d, \quad t \geq 0, \quad  x \in \R,
\end{equation} for some function $\widehat h_i$ and a nonnegative function $\phi_i$ satisfying $\overline \phi_i = \sup_{t \geq 0} \phi_i(t) < \infty$ for each $i$.

We will now need a technical lemma, whose proof is deferred to Appendix~\ref{app:relarb_proofs}.
\begin{lem} \label{lem:relarb}
Let $H \in C^3(\Delta^{d-1}_+)$ be concave with Lipschitz continuous third derivatives and set $F_i(\mu) = \partial_i H (\mu) + H(\mu) -\nabla H(\mu)^\top \mu$ for every $i$. We assume that there exist constants $\underline H, \overline H, \underline F_i, \overline F_i > 0$ such that 
\[\underline H \leq H(\mu) \leq \overline H, \quad \text{and} \quad \underline F_i \leq F_i(\mu) \leq \overline F_i \quad \text{ for all } \mu \in \Delta^{d-1}_+ \text{ and } i=1,\dots,d.\]
We additionally assume that the second derivatives of $H$ are all bounded; that is, $|\partial_{ij}H(\mu)| \leq \overline D$ for all $i,j \in \{1,\dots,d\}$, $\mu \in \Delta^{d-1}_+$ and some constant $\overline D > 0$.
Moreover, we assume that $H$ is strongly concave so that there exists $m > 0$ such that $-\nabla^2 H(\mu) \succeq  mI_{d \times d}$ for all $\mu \in \Delta^{d-1}_+$ in the sense of positive semidefinite ordering. 

 Let Assumption~\ref{ass:impact_inputs} hold and further assume that the impact shape functions are of the form \eqref{eqn:h_separable}.
Suppose that the frictionless market is diverse, nondegenerate, has bounded volatility and satisfies capitalization and market weight floor constraints with respective constants $(\delta_S,\epsilon_S,\sigma_S^2,\kappa_S,\ell_S)$. Define the constants 
% Set  $\delta_P,\epsilon_P,T^*$}
\begin{align*} \kappa_P = \kappa_S\frac{1-\delta_S}{1-\delta_S/2}, \qquad  \epsilon_P = \epsilon_S \frac{ \ell_S^2\kappa_S^2}{2(\kappa_S + d\delta_S\kappa_P/4)^2}\Big(\frac{ \min_{i \in \{1,\dots,d\}}N_i}{\max_{i\in\{1\dots,d\}}N_i}\Big)^2 ,
\qquad \delta_P = \frac{\delta_S}{4}, \\
 T^* =  \frac{2\overline Hd}{m\epsilon_P\delta_P^2}, \qquad \ell_P = \frac{\ell_S\kappa_P}{\kappa_S + d\delta_S\kappa_P/4}.
\end{align*}
Now set
\begin{equation}  \label{eqn:time_dep_quadratic} 
G(t,\mu) = \nu \psi(t)H(\mu),\end{equation}
where we choose $\nu \in (0,\overline \nu)$,
with $\overline \nu$
  defined in \eqref{eqn:bar_nu} of Appendix~\ref{app:relarb_proofs}, and $\psi$ is given by \eqref{eqn:psi} with arbitrary $T_0>0$,  $T_1 = \inf\{t \geq T_0: \Gamma(t) \geq 0\}$
being a stopping time and $T = T_0 + T^*$.  Then the observed market
\begin{enumerate}[noitemsep]
    \item is globally nonexplosive satisfying Assumption~\ref{ass:nonexplosion},
    \item is nondegenerate on $[0,T_1]$ with constant $\epsilon_P$, \label{item:ellipticity}
    \item is diverse on $[0,T_1]$ with constant $\delta_P$,
    \item satisfies market weight floor bounds on $[0,T_1]$ with constant $\ell_P$,
    \item satisfies global capitalization floor constraints with constant $\kappa_P$.
    \end{enumerate}
    Moreover, we have $\P(T_1 < T) = 1$.
\end{lem}
The previous lemma imposes additional requirements on the generating function beyond those of Theorem~\ref{thm:relarb_implicit}. The most stringent new constraint is that the second derivatives of $H$ be bounded. The quadratic generating function satisfies this additional requirement, but the entropy function does not. We are now ready to state the main result of this section.
\begin{thm}[Relative arbitrage; explicit  version]\label{thm:relarb_explicit} 
    Let the assumptions and setup of Lemma~\ref{lem:relarb} hold. Then the trading strategy additively generated by the time-dependent generating function $G$ of \eqref{eqn:time_dep_quadratic} is a strong relative arbitrage on $[0,T]$. In particular, strong relative arbitrage exists for any $T > T^*$.
\end{thm}
\begin{proof} 
We note that since $G(T,\mu(T)) = G(0,\mu(0)) = 0$ we have from the master formula \eqref{eqn:master_formula} that $V^Q(T) = 1 + \Gamma(T)$. It therefore suffices to show that $\Gamma(T) > 0$, $\P$-a.s. To this end, note that $\P(T > T_1)=1$ courtesy of Lemma~\ref{lem:relarb}  
Additionally, note that $\Gamma$ is strictly increasing on $[T_1,T]$, which follows from the fact that the first term in \eqref{eqn:Gamma} is strictly increasing on this interval while the remaining two are nondecreasing.
This gives $\Gamma(T) > \Gamma(T_1) \geq 0$, where the nonnegativity of $\Gamma(T_1)$ follows from the definition of $T_1$, thereby establishing the strong relative arbitrage claim. Moreover, since $T_0$ was arbitrary, we see from this analysis that strong relative arbitrage exists for any $T > T^*$.
\end{proof}
Theorem~\ref{thm:relarb_explicit} establishes the existence of relative arbitrage in the price impact setting. However, to establish this we needed to impose additional assumptions on the frictionless market, namely the capitalization and market weight floors.
Moreover, the proof is considerably more challenging, as one needs the technical Lemma~\ref{lem:relarb}, which requires establishing nondegeneracy and diversity of the observed market, as well as nonexplosion of the SDE \eqref{eqn:SDE}. Many estimates are employed in the proof and, as such, the constants $\epsilon_P,\delta_P,\kappa_P,\ell_P$ and $T^*$ obtained in Lemma~\ref{lem:relarb} may not be sharp.

Nevertheless, these quantitative bounds, in terms of the fundamental input variables, provide financial intuition regarding the nature of the outperformance. Indeed, by inspecting the constants $\delta_P,\epsilon_P$ and $T^*$, we see that they are entirely independent of the impact shape functions $h$ and decay kernels $K$. The impact coefficients do, however, appear in the bound $\overline \nu$ on the \emph{trading speed} parameter $\nu$, defined in \eqref{eqn:bar_nu}. As such, the less liquid the market is, as measured by $h$ and $K$, the slower the investor needs to trade to achieve relative arbitrage. The outperformance time $T^*$ is worse in the case with price impact, since $\delta_P < \delta_S$ and $\epsilon_P < \epsilon_S$, but, if the investor appropriately regulates their trading speed $\nu$, the state of illiquidity need not further lengthen the outperformance time.

The decrease in the trading speed, however, does directly impact the \emph{size} of the achievable outperformance. To see this most clearly, consider the frictionless market and suppose that a strategy $Q$ generated by some function $G$ achieves relative arbitrage on some time horizon $T$. Then the function  $\nu G$, which generates a strategy denoted by $Q^\nu$, leads from the master formula \eqref{eqn:master_formula} and \eqref{eqn:Gamma_frict} to the relative outperformance amount
\[V^{Q^\nu}(T) - 1 = \nu (G(T,\mu(T)) - G(0,\mu(0)) + \Gamma(T)) = \nu(V^Q(T)-1). \]
  As such, a slower trading speed $\nu$ linearly reduces the original outperformance amount $V^Q(T) - 1$. In the market with price impact, the relationship is more complicated since the choice of $\nu$ itself affects the market weight process $\mu$, but a qualitatively similar relationship is expected to hold.

It follows that an investor seeking relative arbitrage may need to balance outperformance time with outperformance amount. This does not arise in the frictionless setting, as one could always suitably scale the generating function to achieve the desired outperformance level. This tradeoff is consistent with previous results on arbitrage in markets with frictions, where arbitrage opportunities may exist but cannot be scaled arbitrarily (see e.g., \cite{guasoni2015hedging}). 

\section{Historical backtests} \label{sec:numerics}
 In this section we apply the price impact framework developed in this paper to historical US equity data, with the goal of illustrating how it manifests on real data and to explore its consequences for the behaviour of additively functionally generated portfolios. As the framework requires a fixed stock universe, we preprocess the data as described in Section~\ref{sec:data} to ensure that all securities selected are available for investment throughout the entire backtesting period. A more extensive rank-based backtest analysis, akin to those in \cite{fernholz2002stochastic} and \cite{ruf2020impact}, that allows the investment universe to fluctuate over time is left for future work. Such an analysis would also require extending the price impact framework to handle rank-based generating functions; see Remark~\ref{rem:rank-based}. In Section~\ref{sec:parameters} we describe how the model parameters are selected and in Section~\ref{sec:backtest_results} we discuss the results of our backtest experiments.\footnote{The code used to produce the backtest results and all of the figures in this paper is publicly available GitHub: \url{https://github.com/ditkin25/Stochastic-portfolio-theory-with-price-impact}. \label{footnote:github}}.

\subsection{Description of the data} \label{sec:data}
We use US equity data provided by the Center for Research in Securities Prices (CRSP) for our empirical analysis. We preprocess the data using the publicly available code \cite[Notebook~6]{ruf2023github} developed by Ruf, which selects only common stock and deals in a systematic way with missing or anomalous returns and other data irregularities to ensure data quality for backtesting purposes. For our study we  focus on two separate backtesting periods:
\begin{enumerate}[noitemsep]
    \item The most recent 5 years of available daily data spanning Jan 1, 2021--Dec 31, 2025, 
    \item The single calendar year Jan 1, 2016--Dec 31, 2016.
\end{enumerate}
The first period has $T = 1255$ trading days, while the second has $T = 252$. 
For each period, we initially select the 500 largest stocks by market capitalization available on the first trading day. To ensure the availability of all assets for the entire backtesting period, we then remove  securities that either (a) delist or (b) fall below $\$1$B in market capitalization at any point during the backtest period, leaving us with $d=454$ and $d=474$ securities in periods (i) and (ii) respectively. This, of course, creates survivorship bias, which is expected to benefit both the benchmark market portfolio and the additively generated portfolios we consider. Since the goal here is to illustrate the developed theoretical price impact framework, we prefer removing the securities that delist or become too small, and acknowledging the survivorship bias this may cause, rather than attempting to adjust the SDE \eqref{eqn:SDE} in an ad hoc way that allows for the number of securities to decrease stochastically over time. 

The reason we consider these two separate backtesting periods is that over short and medium time horizons the key driver of performance in functionally generated portfolios is the diversity term $G(t,\mu(t))$ appearing in the master formula decomposition \eqref{eqn:master_formula}, while the $\Gamma(t)$ term starts to dominate over longer time horizons. It is well-documented that market diversity has decreased in recent years due to a high concentration of capital in several key technology and AI companies, known colloquially as the ``magnificent seven''. This is further exacerbated by the fact that for a fixed investment universe, diversity empirically tends to decrease over time. This behaviour is in contrast to when the ticker set is allowed to refresh according to entrances and exits of securities, where the diversity level remains more stable (see \cite[Figure~5]{campbell2024macroscopic}). As such, the quadratic and entropy additively generated portfolios we consider in the sequel are profitable over the period 2021--2025, but underperform the market portfolio even without market frictions; see the left panels of Figure~\ref{fig:frictionless}. To illustrate the effects of price impact in an increasing market diversity environment we perform an additional backtest over the calendar year 2016, where market diversity increased and both the quadratic and entropy additively generated portfolios outperformed the market portfolio, as depicted in the right panels of Figure~\ref{fig:frictionless}.
\begin{figure}
    \centering
    \includegraphics[width=0.75\linewidth]{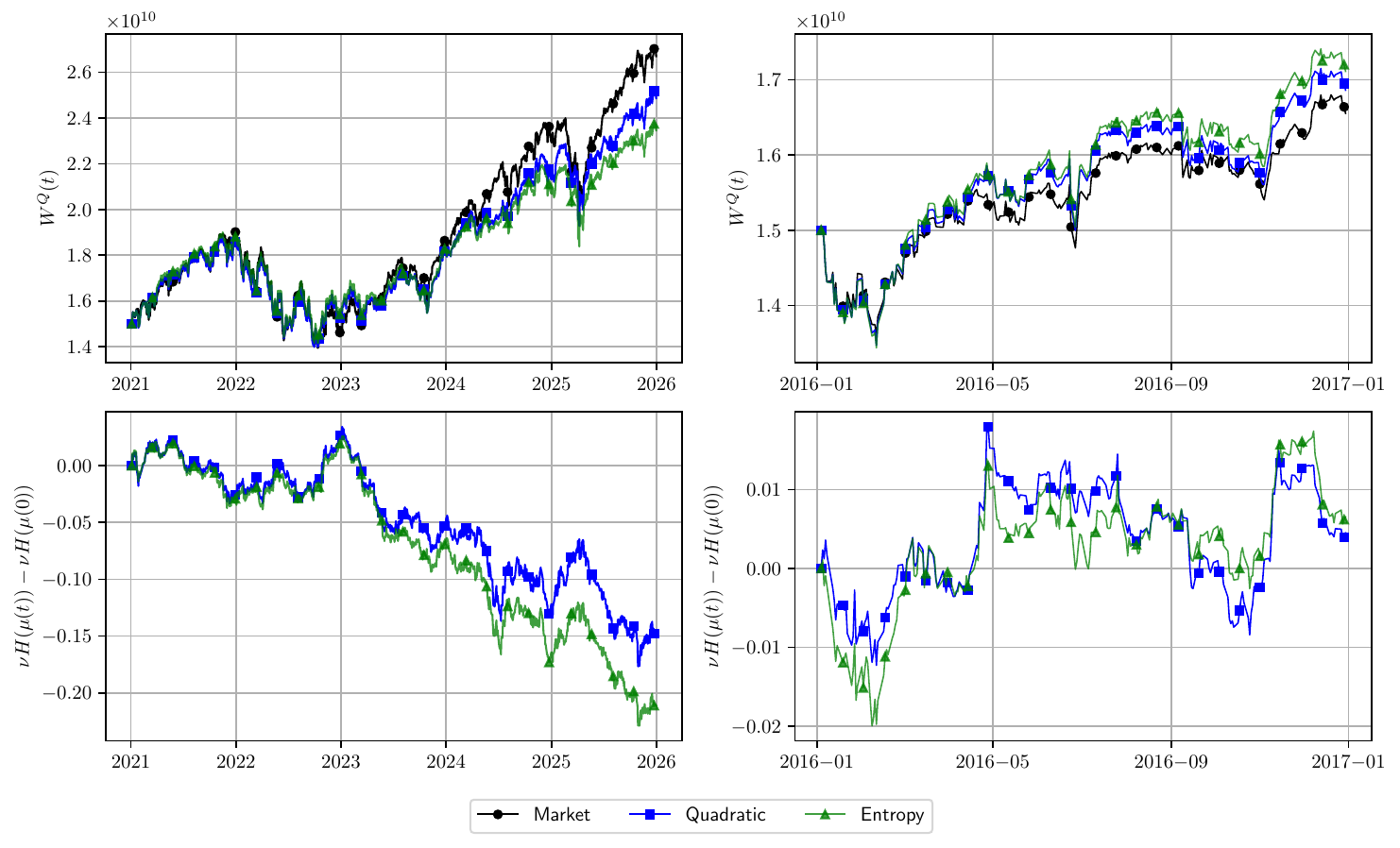}
    \caption{The top panels depict the performance of the market, quadratic and entropy portfolios without trading frictions over the periods 2021--2025 (left) and 2016 (right) respectively. The bottom panels display the change in market diversity $\nu H(\mu(t)) - \nu H(\mu(0))$ over the corresponding periods for quadratic and entropy generating functions $H$ defined in \eqref{eqn:H} with scaling parameters $\nu$ given by \eqref{eqn:nus}. }
    \label{fig:frictionless}
\end{figure}

\subsection{Specifying the model inputs} \label{sec:parameters}
\subsubsection{Impact parameters}
In this section we discuss our choice of impact shape  functions and decay kernels, along with their corresponding parameters. To simplify the illustration, we work with time-homogeneous linear price impact and exponential decay,
\[h_i(t,x) = \lambda_i x, \qquad  K_i(t,s) = e^{-\beta_i(t-s)}, \qquad i=1,\dots,d.\]  For the impact decay parameter, we take $\beta_i =2$ for all securities, corresponding to a decay half-life of $1/3$ of a trading day as reported in \cite{hey.al.23}. For the impact shape parameters $\lambda_i$, we choose each one so that a time weighted average price (TWAP) strategy that trades 1\% of average daily volume (ADV) results in a price impact of $c$ basis points (bps) at the end of the trading day. In line with the average cost estimates reported in \cite[Table~X]{frazzini.al.18}, we select $c \in \{0,5,10,15\}$ to represent frictionless, low, medium and high price impact market regimes respectively. The TWAP strategy corresponds to $dQ_i(t) = 0.01\mathrm{ADV}_idt$, where $\mathrm{ADV}_i$ is the average daily trade volume of security $i$, which we estimate for each security by averaging the security's trade volume over the calendar year prior to the start of the backtesting period; namely, 2020 for period (i) and 2015 for period (ii), hereafter referred to as the \emph{parameter selection period}.\footnote{If a security was not present for the entire  year prior to the backtest start year, we average the daily volume over the days of the previous year in which it traded.} This leads to the relationship
\[\overline S_i \times c \times 0.0001 = \lambda_i^{(c)} I_i(1) = \lambda_i^{(c)}\int_0^1 e^{-\beta_i(1-s)}\times 0.01 \times \mathrm{ADV}_i\, ds, \]
where $\overline S_i$ is the average price of the stock over the parameter selection period and we explicitly indicate the dependence of $\lambda_i$ on the impact size parameter $c$.  Isolating for $\lambda_i^{(c)}$ gives
\begin{equation} \label{eqn:lambda_raw}
\lambda_i^{(c)} = \frac{0.01\beta_i \overline S_i}{(1-e^{-\beta_i})\mathrm{ADV}_i}c.
\end{equation}

\subsubsection{Generating functions}
We study the performance of the quadratic and entropy generating functions. In accordance with the analysis of Section~\ref{sec:relarb}, where the trading rate $\nu$ was an important parameter and the initiation and liquidation periods are key components of the strategy, we will take 
\begin{equation} \label{eqn:G_data}
    G(t,\mu) = \nu \psi(t)H(\mu)
\end{equation} with either
\begin{equation} \label{eqn:H}
    H(\mu) = 1-\frac{1}{2}\sum_{i=1}^d \mu_i^2 \qquad \text{or} \qquad H(\mu) = -\sum_{i=1}^d \mu_i\log \mu_i,
\end{equation} specifying the quadratic and entropy generating functions respectively. The function $\psi$ is given by \eqref{eqn:psi} with a two-week initiation and liquidation period corresponding to $T_0 = 10$ and $T_1 = T-10$. 

The parameter $\nu$ plays an important role and needs to be chosen in a manner that leads to realistic portfolio turnover. To obtain reasonable values for $\nu$ we look to the frictionless setting and consider the time-homogeneous generating function $\nu H(\mu)$. We take the largest 500 stocks at the start of the parameter selection period and remove any stocks that delist or drop below $\$1$B in capitalization from the investment universe during the parameter selection period, following the same procedure as for the backtesting periods themselves. We compute the total dollar portfolio turnover (TO) of the frictionless strategy over the parameter selection period, defined as the minimum of the absolute total buy and total sell dollar volume. Since the signed dollar amount traded on day $t_j$ in asset $i$ is $S_i(t_j)(Q^\nu_i(t_j)-Q^\nu_i(t_{j-1}))$, where $Q^\nu_i$ is given by the right hand side \eqref{eqn:frictionless_additive_holdings} with $\nu H(\mu(t))$ in place of $G(t,\mu(t))$, it follows that 
\begin{equation} \label{eqn:TO}
\mathrm{TO} = \nu \min\Big\{\underbrace{\sum_{i=1}^d\sum_{j=1}^T S_i(t_j)(Q^1_i(t_j)-Q^1_i(t_{j-1}))^+}_{\mathrm{Buys}}, \underbrace{\sum_{i=1}^d\sum_{j=1}^TS_i(t_j)(Q^1_i(t_{j-1})-Q^1_i(t_{j}))^+\Big\}}_\mathrm{Sells},
\end{equation}
where we used the fact that $\nu \mapsto Q^\nu(t_j) - Q^\nu(t_{j-1})$ is linear in the frictionless setting. We also use the notation $x^+ = \max\{x,0\}$.  Active equity managers typically target their turnover to be a certain fraction $\alpha$ of their managed wealth, leading to the relationship $\mathrm{TO} = \alpha w$. Noting that the initial wealth $w$ also has a linear effect on $Q$, we obtain
\[\nu  = \frac{\alpha}{\min\{\mathrm{Buys},\mathrm{Sells}\}},\]
where $\mathrm{Buys}$ and $\mathrm{Sells}$ are computed as in \eqref{eqn:TO} with initial wealth $w = 1$. In our backtesting experiments we take $\alpha \in \{0.6,0.8,1\}$ in line with the proportions found for active managers in the seminal study  \cite[Table~1]{carhart1997on}. However, since the qualitative conclusions of our backtest are similar for the three different choices of $\alpha$ (its first order effect in this setting is scaling the performance of the strategies), we only present the case $\alpha = 0.8$ here and refer the interested reader to the publicly available code on GitHub (linked in footnote \ref{footnote:github}) for the results with other choices of $\alpha$. For this choice of $\alpha = 0.8$, our calibration procedure leads to the values
\begin{equation} \label{eqn:nus}
\nu_Q^{\text{2015}} = 110.7, \qquad \nu_Q^{\text{2020}} = 32.1, \qquad \nu_E^{\text{2015}} = 0.79, \qquad \nu_E^{\text{2020}}  = 0.47, 
\end{equation}
where the subscripts $Q$ and $E$ refer to the quadratic and entropy functions respectively, while the superscripts 2015 and 2020 refer to the respective year of the parameter selection period. The lower values of $\nu^\text{2020}_\cdot$ relative to $\nu^\mathrm{2015}_\cdot$ can be explained by the larger market volatility in 2020 due to the COVID-19 pandemic. Finally, we choose an initial wealth of $w = \$15$B corresponding to a mid- to large-sized institutional  equity manager.

\subsubsection{Using capitalizations rather than prices} \label{sec:caps_as_prices}

Sections~\ref{sec:relative_wealth}--\ref{sec:relarb} assume that the number of shares $N_i$ of each company is constant over time. However, in practice this assumption fails due to corporate actions such as stock splits, buybacks and issuances. As such, it is preferable to use capitalizations, rather than prices, to directly drive the SDE \eqref{eqn:SDE} in our backtest experiments. Only the impact parameters $\lambda^{(c)}$ of \eqref{eqn:lambda_raw} were calibrated using volume data expressed in units of raw share counts. These parameters therefore require adjustment when using the SDE \eqref{eqn:SDE} with capitalizations $\bfs$ of equation \eqref{eqn:boldS} in place of $S$ and ${\bf 1}_d$ in place of $N$.

To this end, we first note that the holdings processes $Q_i$ scale by $1/N_i$ when expressed in units of capitalizations, since
\begin{equation} \label{eqn:caps_holdings}
W^Q(t) = \sum_{i=1}^d P_i(t)Q_i(t) = \sum_{i=1}^d N_iP_i(t)\frac{Q_i(t)}{N_i} = \sum_{i=1}^d \bfp_i(t)\bfq_i(t),
\end{equation}
where we set $\bfq_i(t) = Q_i(t)/N_i$. It follows that, in these units and with linear impact, the prices satisfy
$P_i(t) = S_i(t) + \lambda_i^{(c)} \int_0^t K_i(t,s)dQ_i(s) =  S_i(t) + N_i\lambda_i^{(c)} \int_0^t K_i(t,s)d\bfq_i(s)$. Multiplying both sides by $N_i$, we find that the capitalizations satisfy
\begin{equation} \label{eqn:bold_lambda} 
\bfp_i(t) = \bfs_i(t) + {\boldsymbol{ \lambda}}_i^{(c)} \int_0^tK_i(t,s)d\bfq_i(s), \qquad \text{where} \qquad \boldsymbol{\lambda}_i^{(c)} = N_i^2 \lambda^{(c)}_i.
\end{equation}
As such, we compute $\boldsymbol{\lambda}_i^{(c)}$ using the formula \eqref{eqn:bold_lambda}, where $\lambda_i^{(c)}$ is given by \eqref{eqn:lambda_raw} and $N_i$ is taken to be the number of shares outstanding of asset $i$ on the first day of the corresponding backtesting period. Figure~\ref{fig:lambdas} displays $\boldsymbol{\lambda}^{(5)}$ on a log scale, sorted by the initial market capitalizations of the securities.\footnote{Note that the map $c \mapsto \boldsymbol{\lambda}_i^{(c)}$ is linear so that a plot depicting the parameters $\boldsymbol{\lambda}_i^{(c)}$ for other values $c>0$ would simply be a shifted version of this one.} That is, the values towards the left of the plot correspond to impact parameters for the largest securities, while the values towards the right correspond to impact parameters for the smaller securities. We observe a clear trend of increasing $\boldsymbol{\lambda}^{(c)}_i$ as market capitalization decreases, capturing that smaller securities have a tendency to be less liquid. As a result, trading activity in smaller securities is penalized more than trading activity in larger securities. This is, of course, in contrast to the frictionless setting, where trading in any security, small or large, incurs no cost.

\begin{figure}
    \centering
    \includegraphics[width=0.5\linewidth]{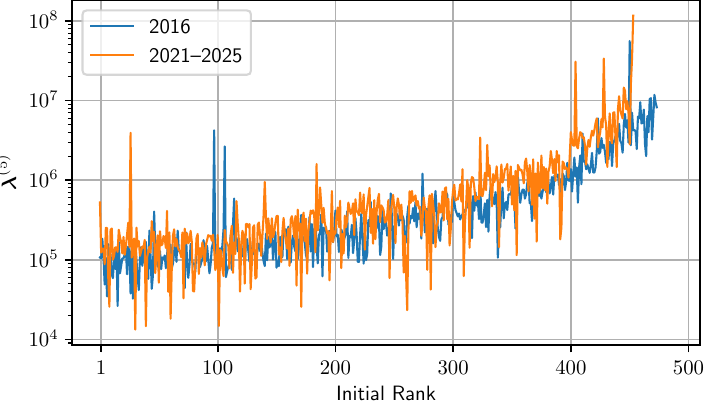}
    \caption{The figure displays $\boldsymbol{\lambda}^{(5)}$ of equation \eqref{eqn:bold_lambda}  on a log scale for both backtest periods. The values are sorted by initial market capitalization rank for each backtest period.}
    \label{fig:lambdas}
\end{figure}

\subsection{Backtest results} \label{sec:backtest_results} For each backtest, the inputs were the parameters described above as well as the time series of daily market capitalizations $(\bfs(t_j))_{j=0}^T$. To obtain the performance of each strategy we discretized the SDE \eqref{eqn:SDE}, using the observed time series of the fundamental capitalizations $\bfs$ as the driving process, and with generating function $G$ of the form \eqref{eqn:G_data}, with $H$ given by either the quadratic or entropy function of equation \eqref{eqn:H}. Figure~\ref{fig:Quadratic} displays the results for the quadratic generating function, while Figure~\ref{fig:Entropy} displays the results for the entropy generating function. In both cases, subfigure (a) corresponds to the backtesting period (i) of 2021--2025, while subfigure (b) corresponds to the backtesting period (ii) of 2016.

%%%%% Quadratic 
\begin{figure}
    \centering
    \begin{subfigure}{\textwidth}
    \centering
\includegraphics[width=0.75\linewidth]{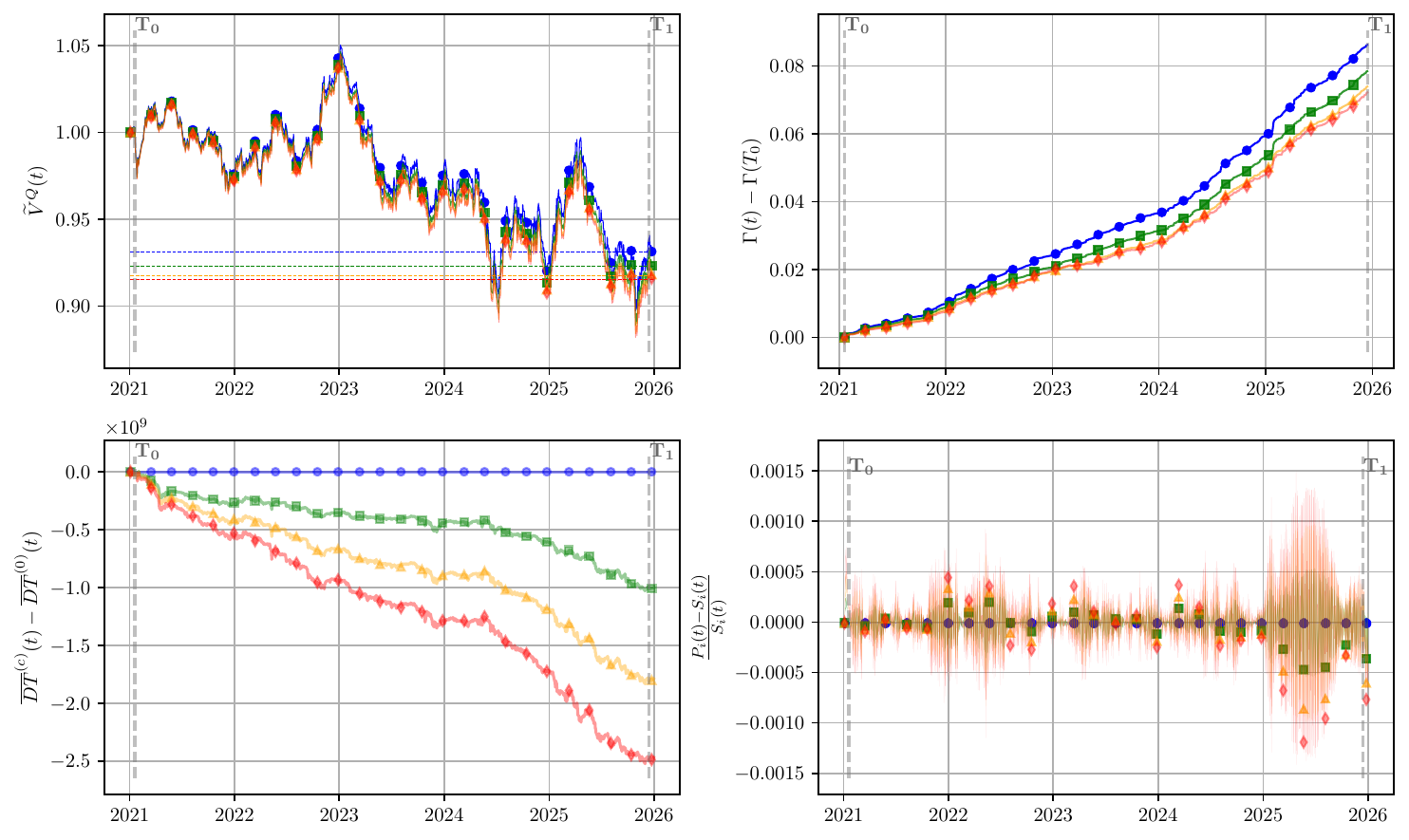}
\caption{2021--2025 quadratic generating function backtest}
    \end{subfigure}
    \vspace{0.25cm}

    \begin{subfigure}{\textwidth}
    \centering
\includegraphics[width=0.75\linewidth]{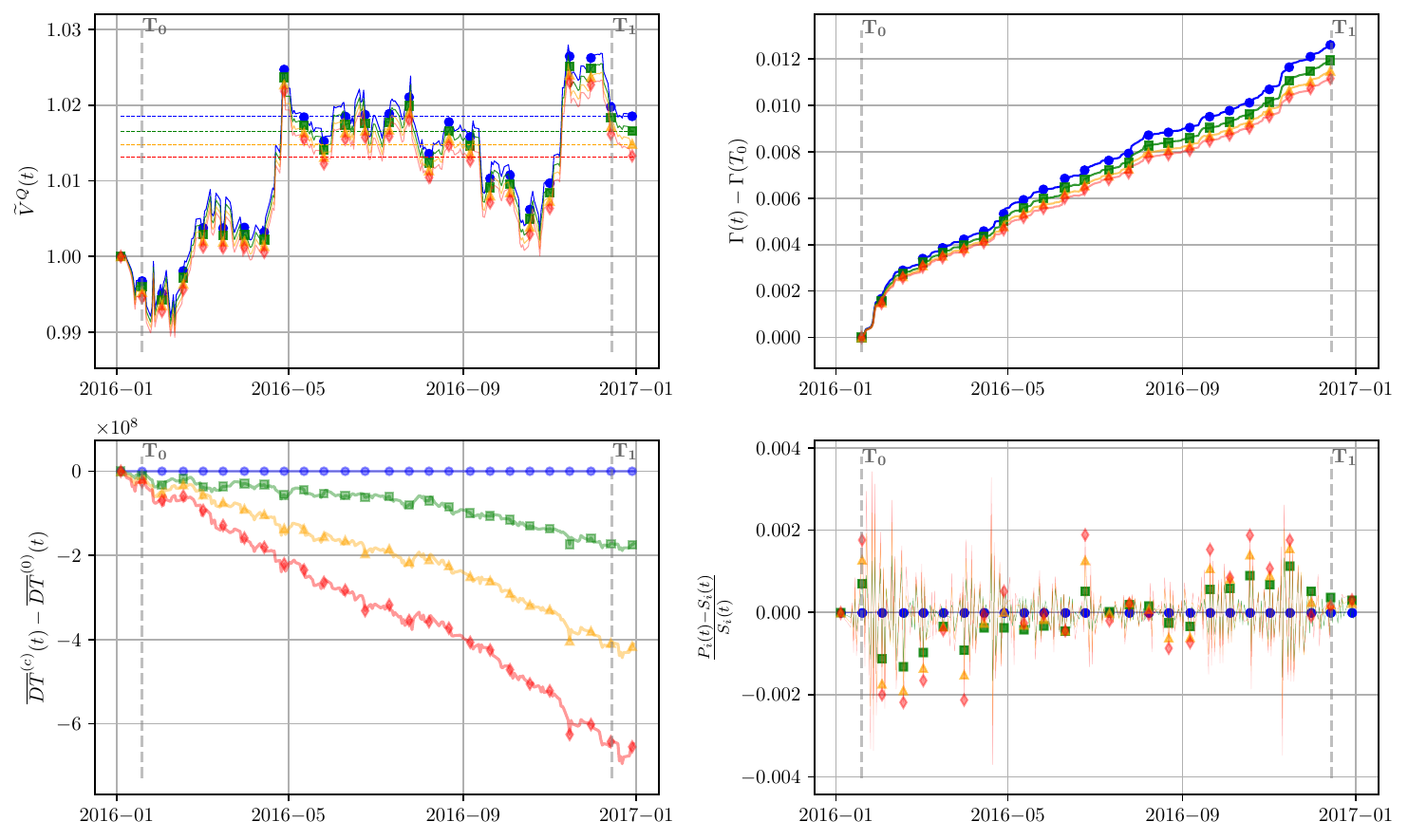}
    \caption{2016 quadratic generating function backtest}
    \end{subfigure}

\includegraphics[width=0.3\linewidth]{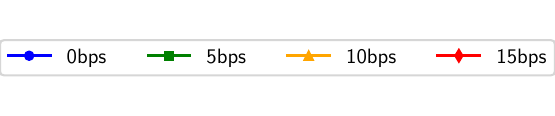}

    \caption{The figures display the performance of the quadratic generating function with $\alpha = 0.8$ and $c \in \{0,5,10,15\}$ in the two respective backtest periods. Depicted is the fundamental relative wealth $\widetilde V^Q(t)$ (top left panel), the $\Gamma$ process in the main trading period, $(\Gamma(t) - \Gamma(T_0); t \in [T_0,T_1])$ (top right panel), the difference in the cumulative absolute dollar volume traded $\overline{\mathrm{DT}}^{(c)}(t) - \overline {\mathrm{DT}}^{(0)}(t)$ (bottom left panel), and the relative impact $\frac{P_i(t)-S_i(t)}{S_i(t)}$ for the stock IBM (bottom right panel).}
    \label{fig:Quadratic}
\end{figure}

%%%%% Entropy 
\begin{figure}
    \centering
    \begin{subfigure}{\textwidth}
    \centering
\includegraphics[width=0.75\linewidth]{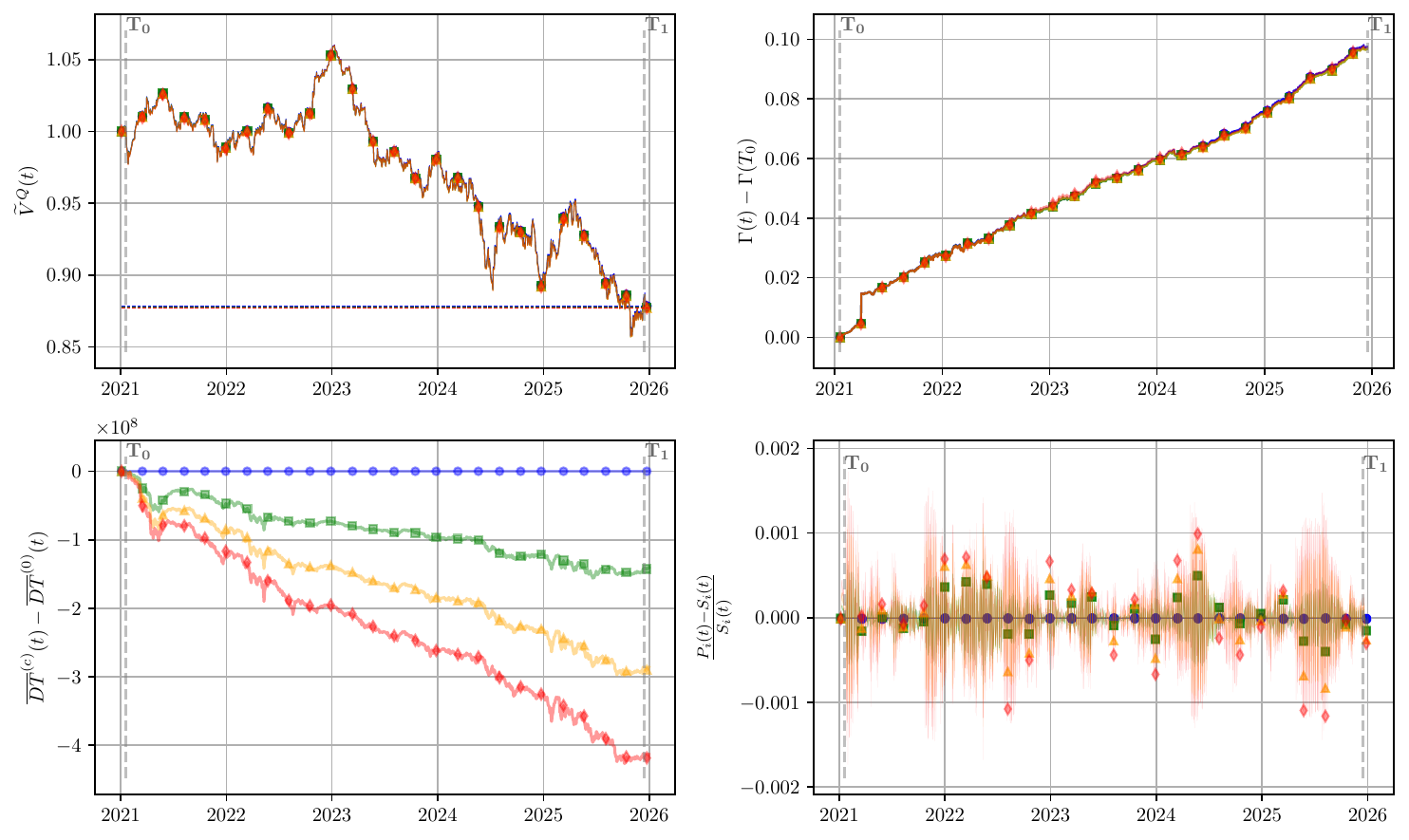}
\caption{2021--2025 entropy generating function backtest}
\label{fig:entropy_2021-2025}
    \end{subfigure}
    \vspace{0.25cm}

    \begin{subfigure}{\textwidth}
    \centering
\includegraphics[width=0.75\linewidth]{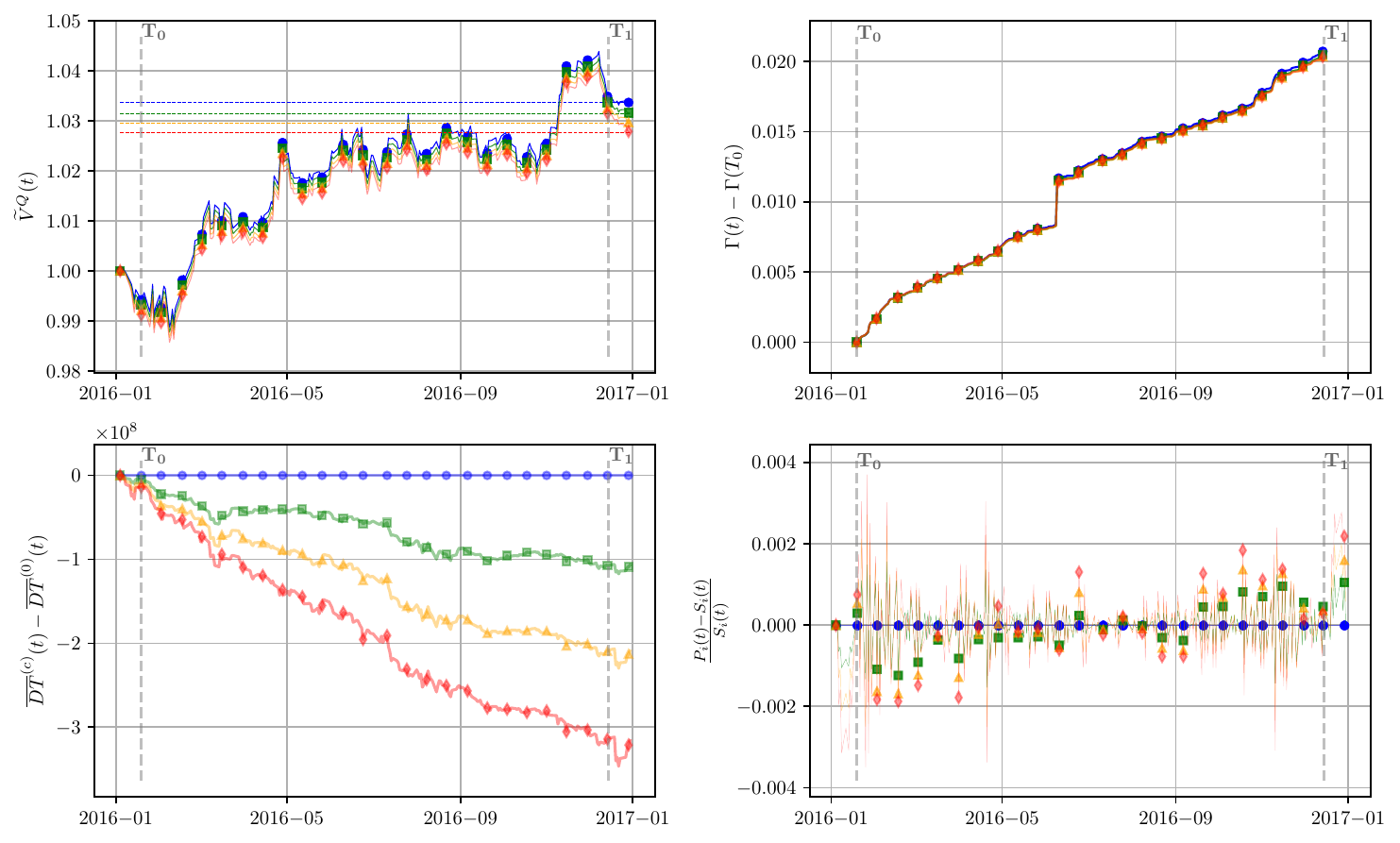}
    \caption{2016 entropy generating function backtest}
    \label{fig:entropy_2016}
    \end{subfigure}

\includegraphics[width=0.3\linewidth]{legend.pdf}

    \caption{The figures display the performance of the entropy generating function with $\alpha = 0.8$ and $c \in \{0,5,10,15\}$ in the two respective backtest periods. Depicted is the fundamental relative wealth $\widetilde V^Q(t)$ (top left panel), the $\Gamma$ process in the main trading period, $(\Gamma(t) - \Gamma(T_0); t \in [T_0,T_1])$ (top right panel), the difference in the cumulative absolute dollar volume traded $\overline{\mathrm{DT}}^{(c)}(t) - \overline {\mathrm{DT}}^{(0)}(t)$ (bottom left panel), and the relative impact $\frac{P_i(t)-S_i(t)}{S_i(t)}$ for the stock IBM (bottom right panel).  }
    \label{fig:Entropy}
\end{figure}

The upper left panel displays the relative fundamental wealth $\widetilde V^Q = \widetilde W^Q/\widetilde W^\Mcal$, defined in Remark~\ref{rem:master_fundamental}, for each strategy and choice of cost parameter $c \in \{0,5,10,15\}$.\footnote{The fundamental relative wealth $\widetilde V^Q$ is displayed to avoid distortion that the accounting wealth $V^Q$ may experience, as discussed in Section~\ref{sec:wealth_process}, although we note that for the experiments done here the relative wealth processes $V^Q$ are nearly indistinguishable to the eye from $\widetilde V^Q$.}  
In all cases, performance is best in the frictionless setting and price impact reduces the achieved wealth. In the case of 2016, where the frictionless portfolios outperformed the market portfolio, even with a cost of $c=15$bps, both the quadratic and entropy generating functions outperformed the market portfolio.
Moreover, since $G(T,\mu(T)) = G(0,\mu(0)) = 0$, courtesy of the fact that $\psi(T) = \psi(0) = 0$, we see that the difference in the final relative performance levels can be attributed to the value of $\Gamma(T)$. The upper right panel plots $\Gamma(t)-\Gamma(T_0)$ for $t \in [T_0,T_1]$, which captures the accumulation of this term during the main trading period where $\partial_t G(t,\cdot) = 0$. We see that, in the case of the quadratic generating function, the gap in the accumulated value of the $\Gamma$ process over the main trading period for different values of $c$ widens as time progresses for both backtest periods. In contrast, for the entropy portfolio, the amount $\Gamma$ accumulates over the main trading period remains in a very tight range, although the frictionless one does achieve a mildly larger accumulation value at time $T_1$ for both backtest periods. For the entropy portfolio, the accumulation period $[0,T_0]$ and, especially, the liquidation period $[T_1,T]$ have a larger effect on the trajectory of $\widetilde V^Q$ for different values of $c$ than the main trading period $[T_0,T_1]$. This is most clearly visible in the top left panel of Figure~\ref{fig:entropy_2016}, where the relative wealth time series experience a wider separation after the time $T_1$. Due to the dominance of the term $\int_{T_1}^\cdot \partial_t G(t,\mu(t))dt$ during this period over the other terms that make up the process $\Gamma$, we omit the liquidation period, and also the initiation period $[0,T_0]$ where $\int_0^\cdot \partial_t G(t,\mu(t))dt$ dominates, from the top right panels (cf.\ the right panel of Figure~\ref{fig:psi}).

Next, we investigate the effect of the price impact framework on the investor's trading activity and on observed prices. To this end, we compute the time series of cumulative total dollar volume traded, measured using the fundamental capitalizations, 
\[\overline{\mathrm{DT}}^{(c)}(t_k) = \sum_{j=1}^k\sum_{i=1}^d \bfs_i(t_j)|\bfq^{(c)}_i(t_j)-\bfq^{(c)}_i(t_{j-1})|, \qquad k=1,\dots,T,\]
where $\bfq_i^{(c)}$ is defined in \eqref{eqn:caps_holdings}, with the superscript $(c)$ indicating the dependence of the holdings process on the impact parameter $c$. The bottom left panels display  $\overline{\mathrm{DT}}^{(c)}(t) - \overline{\mathrm{DT}}^{(0)}(t)$, the difference between this quantity and its frictionless counterpart, for $c \in \{0,5,10,15\}$. We see, across both backtest periods and both generating functions, that the curves are decreasing with $c$, indicating that the additively generated portfolios trade less as price impact increases. 

To understand why this happens, we examine the key terms on the right-hand side of \eqref{eqn:frictionless_additive_holdings} specifying the holdings during the main trading period $[T_0,T_1]$. The key asset-dependent term influencing asset $i$'s holdings is $\partial_i G(t,\mu(t)) = \nu \partial_i H(\mu(t))$. The concave and additive structure of the generating function implies that in both cases the map $\mu_i \mapsto \partial_i H(\mu)$ is decreasing. Hence, heuristically, an increase in the market weight of asset $i$ typically leads to a sell order in asset $i$, which partially offsets the increase through price impact. Since the framework developed in this paper incorporates price impact when determining the investor's holdings process (see Remark~\ref{rem:target_holdings}), this results in a smaller prescribed trade relative to the frictionless setting, and explains the ordering of the curves in the lower left panels of Figures~\ref{fig:Quadratic} and \ref{fig:Entropy}. Moreover, since smaller securities tend to have higher price impact as showcased in Figure~\ref{fig:lambdas}, this affects small stocks more than large stocks. As such, in the absence of price impact the investor would place relatively large orders in small securities, but because of price impact this is penalized more heavily than for large stocks and the investor's realized trades in small securities tend to decrease.
Consequently, we see that the framework developed here, which uses the impacted prices $P$ to determine the target holdings, systematically leads to a reduction in trading volume when price impact is present.

Turning to prices, the bottom right panel in both subfigures of Figures~\ref{fig:Quadratic} and \ref{fig:Entropy} depicts the relative impact,
\[\frac{P_i(t) - S_i(t)}{S_i(t)} = \frac{\bfp_i(t) - \bfs_i(t)}{\bfs_i(t)}\]
of IBM, which serves as a representative security. As expected, the impact tends to increase with $c$, and when $c=0$ there is no price impact. Moreover, the sign of the investor's price impact can change even over short time periods. This, in line with the previous discussion regarding dollar volume traded, is due to the fact that the additively generated strategies typically prescribe buying more units of an asset when its price falls, while selling some of the holdings when its price rises. As such, even on consecutive trading days, the strategy may require the investor to switch their trading direction in any given asset, consistent with what is observed in the bottom right panels of Figures~\ref{fig:Quadratic} and \ref{fig:Entropy}. This additionally highlights the importance of considering semimartingale strategies $Q$ in the theoretical sections of this paper, so that the holdings process $Q$ can exhibit the same type of diffusive fluctuations as prices.

\section{Conclusion} \label{sec:conclusion}
We have introduced a framework for stochastic portfolio theory that is able to incorporate an investor's price impact. In this setting we derived the relative wealth equation \eqref{eqn:relative_wealth} with respect to the market portfolio and established a master formula for additive functional generation of trading strategies given by \eqref{eqn:master_formula} and \eqref{eqn:Gamma}. Unlike the frictionless case, care is needed to ensure that the trading strategy gives rise to a well-defined market with semimartingale prices; this is the content of Theorem~\ref{thm:main}, which establishes well-posedness of the SDE \eqref{eqn:SDE}. As an application of the master formula, we developed results on relative arbitrage in the price impact setting.  Theorems~\ref{thm:relarb_implicit} and \ref{thm:relarb_explicit} provide versions under implicit and explicit assumptions on the model inputs respectively. Finally, we conducted an empirical backtest study of the model on real data to validate and complement the theoretical results. We find that smaller securities tend to have higher impact parameters and, since functionally generated trading strategies constantly transact, this contributes to reduced performance of these strategies relative to the frictionless setting.

Our framework allows for many core tenets of SPT to be analyzed in a market with price impact, but extending certain results from frictionless SPT to the price impact setting remains open. For instance, rank-based generating functions are widely studied in SPT and allow one to incorporate the effects of \emph{leakage} into the analysis, which affects portfolio performance when stocks enter and exit the investment universe. As discussed in Remark~\ref{rem:rank-based}, rank-based generating functions lead to holdings processes $Q$ that are not semimartingales and, as such, are not covered by the results of this paper. Additionally, as highlighted in Section~\ref{sec:multiplicative}, a master formula decomposition, where the right-hand side is independent of the relative wealth process itself, does not seem to hold for multiplicatively generated portfolios. As such, extending the analysis to incorporate other types of functional generation, including rank-based functions and multiplicative generation, is an important topic for future research. Further open problems include optimal portfolio selection in the SPT price impact setting, the study of trading's effect on the capital distribution curve, the developing of equilibrium models for price formation, and the establishment of robustness guarantees for long-term performance of functionally generated portfolios.

\paragraph{Acknowledgments.} The author acknowledges Martin Larsson, Paul Mangers Bastian, Johannes Muhle-Karbe, Johannes Ruf and two anonymous referees for providing numerous comments that improved the manuscript.

\appendix

\section{Wealth equation derivations} \label{app:wealth_proofs}
The purpose of this section is to prove Proposition~\ref{thm:wealth} and Theorem~\ref{thm:relative_wealth} establishing the wealth and relative wealth equations respectively. For the benefit of the reader, we start with a heuristic derivation of the wealth equation \eqref{eqn:wealth_smooth}, for absolutely continuous strategies, which in the main text was our starting point for deriving the general wealth equation \eqref{eqn:wealth}.\footnote{We refer to \cite[Chapter~2.2]{webster.23} for a more detailed derivation.}  To this end, we look at a discrete time trading increment from time $t-\Delta t$ to $t$. We will write $\Delta Q(t)$ for $Q(t) - Q(t-\Delta t)$ and use analogous notation for $\Delta P(t)$ and $\Delta W^Q(t)$. At time $t-\Delta t$ the investor starts off with holdings $Q(t-\Delta t)$ and, by time $t$, ends up with the holdings $Q(t)$ after rebalancing. The self-financing condition stipulates that the rebalancing trades do not require an inflow or outflow of cash. In our current setting of fully invested trading strategies this leads to the condition $\Delta Q(t)^\top \widetilde P(t) = 0$, where $\widetilde P(t)$ is the effective execution price, which is typically different from $P(t)$ due to price impact. The self-financing condition can be equivalently rewritten as
\begin{equation} \label{eqn:discrete_self_financing}
    0 = P(t-\Delta t)^\top \Delta Q(t) + \Delta P(t)^\top \Delta Q(t) - \Delta Q(t)^\top \big( P(t)-\widetilde P(t)\big).
\end{equation}
Since $Q$ is absolutely continuous, we have that $\Delta Q(t) = \dot Q(t-\Delta t)\Delta t$, from which we see that the final term in \eqref{eqn:discrete_self_financing} is of higher order than the other two, since $ P(t) - \widetilde P(t) = O(\Delta Q(t)) = O(\Delta t)$. As such, in the limit as $\Delta t \to 0$ we recover the usual self-financing condition 
\begin{equation} \label{eqn:cont_self_financing}
    0 = P(t)^\top dQ(t) + d[P,Q](t).
\end{equation}
We can now derive \eqref{eqn:wealth} for absolutely continuous strategies,
\begin{align*} 
    \Delta W^Q(t) &  = Q(t)^\top P(t) - Q(t-\Delta t)^\top P(t-\Delta t)  \\
    & = Q(t-\Delta t)^\top \Delta P(t) + P(t-\Delta t)^\top \Delta Q(t) + \Delta Q(t)^\top \Delta P(t). 
\end{align*} 
Sending $\Delta t \to 0$ and using \eqref{eqn:cont_self_financing} we obtain $dW^Q(t) = Q(t)^\top dP(t)$, which is precisely \eqref{eqn:wealth}.

Our next goal is to prove the wealth equation \eqref{eqn:wealth} for semimartingale strategies.
We start by establishing the continuous extension claim for \eqref{eqn:wealth_smooth} in the following lemma. \begin{lem} \label{lem:wealth_approx}
    Let $Q$ be a continuous semimartingale.  Then there exists a sequence of predictable absolutely continuous holdings processes $(Q^n;n \in \N)$ such that the pointwise limit $W^Q(T) := \lim_{n \to \infty} W^{Q^n}(T)$ exists almost surely for every $T \geq 0$ and satisfies 
    \begin{equation} \label{eqn:wealth_predictable}  
    \begin{aligned}
    W^Q(T) & = w + \int_0^T Q(t)^\top dS(t) + I(T)^\top Q(T)-I(0)^\top Q(0) - \sum_{i=1}^d \bigg(\frac{H_i(T,J_i(T))}{K_i(
    T,T)} - \frac{H_i(0,J_i(0))}{K_i(0,0)}\bigg) \\
    &  + \sum_{i=1}^d\int_0^T\bigg(\frac{\partial_t H_i(t,J_i(t))}{K_i(t,t)} - \frac{\partial_t K_i(t,t) + \partial_sK_i(t,t)}{K_i^2(t,t)}H_i(t,J_i(t)) +\frac{h_i(t,J_i(t))}{K_i(t,t)}b^J_i(t,Q_i^{[0,t]})\bigg)dt.
\end{aligned}
\end{equation}
\end{lem}
\begin{proof} Fix a $\psi \in C_c^\infty(\R)$, which is a nonnegative function satisfying  $\int_{-\infty}^\infty \psi(x)dx = 1$ and $\mathrm{supp}(\psi) = [0,1]$. Put $\psi^n(t) = \frac{n\psi(nt)}{\int_0^{nt} \psi(s)ds}$ for $n \in \N$ and define $Q^n(t) = \int_0^t Q(s)\psi^n(t-s)ds$.
Clearly, $Q^n$ is predictable, since it only depends on $Q(s)$ for $s \leq t$. Additionally, it is absolutely continuous so that the wealth process corresponding to that strategy is unambiguously given by \eqref{eqn:wealth_smooth}. 
By basic properties of mollifiers we have almost surely that $Q^n(t) \to Q(t)$ pointwise and in $L^1_{\mathrm{Loc}}(\R_+;\R^d)$ as $n \to \infty$. Moreover, from \eqref{eqn:J_dynamics} we see that almost surely $J^n(t) \to J(t)$ in $L^1_{\mathrm{Loc}}(\R_+;\R^d)$, where $J^n$ is the impact state process corresponding to $Q^n$. Finally, note that since $|Q^n(T)| \leq \sup_{t \leq T}|Q(t)|$ we have by the stochastic dominated convergence theorem \cite[Theorem~IV.2.12]{revuz1999continuous} that $\int_0^T Q^n(t)^\top dS(t) \to \int_0^T Q(t)^\top dS(t)$ as $n \to \infty$. As such, we see that sending $n \to \infty$ in \eqref{eqn:wealth_smooth} yields \eqref{eqn:wealth_predictable}.
\end{proof}
This process $W^Q$ is the investor's wealth process when using the trading strategy $Q$. 
We are now ready to prove Proposition~\ref{thm:wealth}.

\begin{proof}[Proof of Proposition~\ref{thm:wealth}] We will reformulate the right-hand side of \eqref{eqn:wealth_predictable}. Using It\^o's formula on $H_i(t,J_i(t))/K_i(t,t)$ we get
\begin{align*}
	\frac{H_i(T,J_i(T))}{K_i(T,T)} - 	\frac{H_i(0,J_i(0))}{K_i(0,0)}  & = \int_{0}^T\frac{h_i(t,J_i(t))}{K_i(t,t)}dJ_i(t) + \int_0^T\frac{\partial_x h_i(t,J_i(t))}{2K_i(t,t)}d[J_i](t) \\
	& + \int_0^T\bigg(\frac{\partial_t H_i(t,J_i(t))}{K_i(t,t)} - \frac{(\partial_t K_i(t,t) + \partial_s K_i(t,t))H_i(t,J_i(t))}{K_i(t,t)^2}\bigg)dt. 
	% & + \sum_{0 < t \leq T} \sum_{i=1}^d \left(\frac{\Delta H_i(t,J_i(t))}{K_i(t,t)} - \frac{h_i(t,J(t-))}{K_i(t,t)}\Delta J _i(t)\right).
\end{align*}
    Substituting this into \eqref{eqn:wealth_predictable} yields
\begin{equation} \label{eqn:wealth_intermediate} 
\begin{split} \raisetag{1cm}
	W^Q(T) - w  & = \int_0^T Q(t)^\top dS(t) + I(T)^\top Q(T) - I(0)^\top Q(0) + \sum_{i=1}^d\int_0^T\frac{h_i(t,J_i(t))}{K_i(t,t)}b^J_i(t,Q_i^{[0,t]})dt \\
    & \hspace{0.5cm}    -\sum_{i=1}^d \int_{0}^T \frac{h_i(t,J_i(t))}{K_i(t,t)}dJ_i(t) -\sum_{i=1}^d\int_0^T \frac{\partial_x h_i(t,J_i(t))}{2K_i(t,t)}d[J_i](t). 
\end{split} 
\end{equation}
Now recalling the relationship \eqref{eqn:dQ} and that $I_i(t) = h_i(t,J_i(t))$, we obtain 
\[\sum_{i=1}^d\int_0^T\frac{h_i(t,J_i(t))}{K_i(t,t)}b^J_i(t,Q_i^{[0,t]})dt -\sum_{i=1}^d \int_{0}^T \frac{h_i(t,J_i(t))}{K_i(t,t)}dJ_i(t) = -\sum_{i=1}^d \int_{0}^T I_i(t)dQ_i(t) .\]
Moreover, we have that 
\[\int_0^T\frac{\partial_x h_i(t,J_i(t))}{2K_i(t,t)}d[J_i](t) = \frac{1}{2}\int_0^T\partial_x h_i(t,J_i(t))d[J_i,Q_i](t) = \frac{1}{2}[I_i,Q_i](T).\]
Hence, substituting these identities into \eqref{eqn:wealth_intermediate} and integrating $I(T)^\top Q(T)$ by parts gives
\begin{align*}
	W^Q(T) - w  = &  \int_0^T Q(t)^\top dS(t) + I(T)^\top Q(T)-I(0)^\top Q(0) -\sum_{i=1}^d \int_{0}^T I_i(t)dQ_i(t) - \frac{1}{2}\sum_{i=1}^d[I_i,Q_i](T)  \\
	 = & \int_0^T Q(t)^\top dS(t)  + \sum_{i=1}^d\int_{0}^TQ_i(t) dI_i(t) + \sum_{i=1}^d [I_i,Q_i](T) - \frac{1}{2}\sum_{i=1}^d [I_i,Q_i](T)  \\
	 = & \int_{0}^T Q(t)^\top dP(t) + \frac{1}{2}\sum_{i=1}^d [I_i,Q_i](T).  
\end{align*}
This completes the proof.
\end{proof}

Next, we establish the relative wealth equation \eqref{eqn:relative_wealth}.
\begin{proof}[Proof of Theorem~\ref{thm:relative_wealth}] We start by noting the following three relationships,
\begin{equation} \label{eqn:numeraire_identities}
\begin{aligned}
    \frac{dW^\mathcal{M}(t)}{W^\mathcal{M}(t)} & = \frac{d\overline \bfp(t)}{\overline \bfp(t)}, \\
    \frac{d\mu_i(t)}{\mu_i(t)} & = \frac{d\bfp_i(t)}{\bfp_i(t)} -  \frac{d\overline \bfp(t)}{\overline \bfp(t)} + \frac{d[\overline \bfp](t)}{\overline \bfp^2(t)} - \frac{d[\bfp_i,\overline \bfp](t)}{ \bfp_i(t)\overline \bfp(t)}, &    i=1,\dots,d, \\
    [W^Q,W^\mathcal{M}](T) & = \int_0^T\sum_{i=1}^d \frac{w}{N_i\overline \bfp(0)}Q_i(t)d[\bfp_i,\overline \bfp](t), &    i=1,\dots,d.
\end{aligned} 
\end{equation}
 These three identities follow directly from the wealth equations \eqref{eqn:wealth} and \eqref{eqn:market_wealth}, together with It\^o's quotient rule applied to compute the dynamics of $\mu$. We now apply It\^o's formula to the quotient $W^Q/W^{\mathcal{M}}$ to get
\begin{align*}
    dV^Q(t) & = V^Q(t)\bigg(\frac{dW^Q(t)}{W^Q(t)} - \frac{dW^\mathcal{M}(t)}{W^\mathcal{M}(t)} + \frac{d[W^\mathcal{M}](t)}{(W^\mathcal{M}(t))^2} - \frac{d[W^Q,W^\mathcal{M}](t)}{W^Q(t)W^\mathcal{M}(t)} \bigg) 
    \\
    & = \frac{1}{W^\mathcal{M}(t)}\bigg(dW^Q(t)- W^Q(t)\frac{dW^\mathcal{M}(t)}{W^\mathcal{M}(t)} + W^Q(t)\frac{d[W^\mathcal{M}](t)}{(W^\mathcal{M}(t))^2} - \frac{d[W^Q,W^\mathcal{M}](t)}{W^\mathcal{M}(t)}\bigg).
\end{align*}   
Next, we use \eqref{eqn:numeraire_identities}, \eqref{eqn:wealth} and the fact that $W^Q(t) = Q(t)^\top P(t) = \sum_{i=1}^d \frac{Q_i(t)}{N_i}\bfp_i(t)$ to obtain
\begin{align*} 
V^Q(T) &= 1+ \int_0^T\frac{\overline \bfp(0)}{w\overline \bfp(t)}\bigg( \sum_{i=1}^d Q_i(t)dP_i(t) - \sum_{i=1}^d\frac{Q_i(t)}{N_i}\bfp_i(t)\bigg(\frac{d\overline \bfp(t)}{\overline \bfp(t)} -\frac{d[\overline \bfp](t)}{\overline \bfp^2(t)}\bigg)\\
& \hspace{3cm}- \sum_{i=1}^d \frac{Q_i(t)}{N_i}\frac{d[\bfp_i,\overline \bfp](t)}{\overline \bfp(t)}\bigg) + \frac{1}{2}\sum_{i=1}^d\int_0^T\frac{\overline \bfp(0)}{w\overline \bfp(t)}d[I_i,Q_i](t) \\
& = 1 + \int_0^T \sum_{i=1}^d \frac{\bfp_i(t)}{\overline \bfp(t)}\frac{\overline \bfp(0)Q_i(t)}{wN_i}\bigg(\frac{d\bfp_i(t)}{\bfp_i(t)} -  \frac{d\overline \bfp(t)}{\overline \bfp(t)} + \frac{d[\overline \bfp](t)}{\overline \bfp^2(t)} - \frac{d[\bfp_i,\overline \bfp](t)}{ \bfp_i(t)\overline \bfp(t)}\bigg) \\
& \hspace{1.9cm} + \frac{1}{2}\sum_{i=1}^d\int_0^T\frac{\overline \bfp(0)}{w\overline \bfp(t)}d[I_i,Q_i](t)\\
& = 1 + \int_0^T\sum_{i=1}^d\mu_i(t)\frac{\overline \bfp(0)Q_i(t)}{wN_i}\bigg(\frac{d\mu_i(t)}{\mu_i(t)}\bigg) + \frac{1}{2}\sum_{i=1}^d\int_0^T\frac{\overline \bfp(0)}{w\overline \bfp(t)}d[I_i,Q_i](t) \\
& = 1 +\int_0^T \sum_{i=1}^d \frac{\overline \bfp(0)Q_i(t)}{wN_i}d\mu_i(t) +  \frac{1}{2}\sum_{i=1}^d\int_0^T\frac{\overline \bfp(0)}{w\overline \bfp(t)}d[I_i,Q_i](t).
\end{align*} 
In the second equality we factored $\bfp_i(t) Q_i(t)/N_i$ and in the third we used that $\mu_i = \bfp_i/\overline \bfp$ and substituted the dynamics \eqref{eqn:numeraire_identities} for the market weights $\mu_i$. This is precisely \eqref{eqn:relative_wealth}.
\end{proof}

\section{Proofs for Section~\ref{sec:functionally_generated}}
\subsection{Proof of Proposition~\ref{prop:AP}} \label{sec:proof_AP}
\begin{proof} 
To simplify the exposition, we set for $t \geq 0$, $\mu \in \Delta^{d-1}_+$,
\[\Xi(t,\mu) =-\mathrm{diag}(N)\overleftrightarrow{\nabla^2 G}^{\mu}(t,\mu)\mathrm{diag}(N), \quad  C(t,\mu,J) = \frac{w}{\overline \bfp(p^0)} \mathrm{diag}\Big(\partial_x h(t,J) \circ K(t,t)\Big)\Xi(t,\mu)\]
so that $A^P(t,p,J) = I_{d \times d} + \overline \bfp^{-1}(p)C\big(t,\mu(p),J\big)$. It also follows from the representation \eqref{eqn:arrowM} that,
for any $\varphi\in \R^d$,
    \begin{equation} \label{eqn:arrowM_PSD}
        \varphi^\top \overleftrightarrow{\nabla^2 G}^{\mu}(t,\mu)\varphi =  (\varphi - \overline \varphi\mu)^\top \nabla^2 G(t,\mu)(\varphi - \overline\varphi\mu),
    \end{equation}
    where $\overline \varphi = \varphi^\top {\bf 1}_d$.
    
    Now let $(t,p,J) \in [0,\infty) \times (0,\infty)^d \times \R^d$ be given and assume that $G(t,\mu(p))$ satisfies \eqref{eqn:quadratic_form}. From \eqref{eqn:arrowM_PSD} we deduce that $\overleftrightarrow{\nabla^2 G}^{\mu(p)}(t,\mu(p))$ is a negative semidefinite matrix, since $(\varphi - \overline \varphi \mu)^\top {\bf 1}_d = 0$ for any $\varphi \in \R^d$. It then follows that $\Xi(t,\mu(p))$ is a positive semidefinite matrix. Since the product of two positive semidefinite matrices has nonnegative eigenvalues \cite[Corollary~7.6.2]{horn2013matrix} we have that $C(t,\mu(p),J)$ has nonnegative eigenvalues. From here it follows that all of the eigenvalues of $A^P(t,p,J)$ are real and not smaller than one. Indeed, $(\lambda,u)$ is an eigenvalue-eigenvector pair for $A^P(t,p,J)$ if and only if 
    \[A^P(t,p,J) u = \lambda u \iff C(t,\mu(p),J)u = \overline \bfp(p)(\lambda - 1)u,\]
    establishing the claim. 

    For the converse we assume that $G(t^*,\mu^*)$ does not satisfy \eqref{eqn:quadratic_form} for some $t^* \geq 0$ and $\mu^* \in \Delta^{d-1}_+$. Then there exists $x \in \R^d$ with $x^\top {\bf 1}_d = 0$ satisfying $x^\top \nabla^2 G(t^*,\mu^*)x > 0$. Setting $\varphi = x$ in \eqref{eqn:arrowM_PSD} we see that $\overleftrightarrow{\nabla^2 G}^{\mu^*}(t^*,\mu^*)$ fails to be negative semidefinite. We now establish that  $C(t^*,\mu^*,J^*)$ has a negative eigenvalue. Let $D = \mathrm{diag}(\partial_x h(t^*,J^*) \circ K(t^*,t^*))$ and note that $D$ is a diagonal matrix with strictly positive entries since $\partial_x h_i(t^*,J_i^*) > 0$ for every $i$ by assumption. Hence $C(t^*,\mu^*,J^*) = D^{1/2}(D^{1/2}\frac{w}{\overline \bfp(p^0)}\Xi(t^*,\mu^*)D^{1/2})D^{-1/2}$ is similar to the matrix $\widetilde \Xi(t^*,\mu^*) := D^{1/2}\Xi(t^*,\mu^*) D^{1/2}$ and therefore shares the same eigenvalues. But $\widetilde \Xi(t^*,\mu^*)$ is not positive semidefinite because $\Xi(t^*,\mu^*)$ is not. Indeed, taking $x \in \R^d$ such that $x^\top \Xi(t^*,\mu^*)x < 0$ and setting $y = D^{-1/2}x$ we have that $y^\top \widetilde \Xi(t^*,\mu^*) y < 0$. Hence,  $\widetilde \Xi(t^*,\mu^*)$ and, consequently, $C(t^*,\mu^*,J^*)$ must have a negative eigenvalue $\lambda^* < 0$. We denote by $u^*$ the corresponding eigenvector. 
    
    We are now ready to establish noninvertibility of $A^P$. To this end set $p^*_i = -\lambda^*\mu^*_i/N_i > 0$ so that $\overline \bfp(p^*) = -\lambda^*>0$ and $\mu(p^*) = \mu^*$. We have that
    \[A^P(t^*,p^*,J^*)u^* = u^* + \frac{1}{-\lambda^*} C(t^*,\mu^*,J^*)u^* = u^*-\frac{1}{\lambda^*}(\lambda^*u^*) = 0 \]
    proving the claim.
\end{proof}
\subsection{Deriving the SDE \texorpdfstring{\eqref{eqn:SDE}}{}}  \label{app:SDE_derivation}
We start by formally deriving \eqref{eqn:SDE}. To carry this out we temporarily suppose that $P, Q$ and $\mu$ are semimartingales so that equation \eqref{eqn:frictionless_additive_holdings}, with $\Gamma$ given by \eqref{eqn:Gamma} as derived in Section~\ref{sec:heuristic} holds. To simplify the exposition, in this section for any semimartingale $X$ we will write $X = X(0) + X^{FV} + X^M$ for its canonical semimartingale decomposition, where $X^{FV}$ is the finite variation part and $X^M$ is the local martingale part. 

We first derive $Q^M$. Applying It\^o's formula to the function $\partial_i G(t,\mu) + G(t,\mu)- \nabla G(t,\mu)^\top \mu$ appearing in the right-hand side of \eqref{eqn:frictionless_additive_holdings} and collecting the local martingale terms gives
\[dQ_i^M(t) = \frac{wN_i}{\overline \bfp(0)}\sum_{k=1}^d \Big(\partial_{ik}G(t,\mu(t))-\sum_{\ell=1}^d \partial_{k\ell} G(t,\mu(t))\mu_\ell(t)\Big)d\mu_k^M(t).\]
    Recalling the dynamics \eqref{eqn:numeraire_identities} for $\mu$ in terms of $P$, then gives
\begin{align} dQ^M_i(t)  & = \sum_{j=1}^d\frac{wN_iN_j}{\overline \bfp(0)\overline \bfp(t)}\sum_{k=1}^d\Big(\partial_{ik}G(t,\mu(t))-\sum_{\ell=1}^d \partial_{k\ell}G(t,\mu(t))\mu_\ell(t)\Big)(\delta_{jk}-\mu_k(t))dP^M_j(t) \nonumber \\
& = \sum_{j=1}^d A^Q_{ij}(t,P(t))dP^M_j(t), \label{eqn:Q^M}
\end{align}
where we recall that $A^Q$ is given by \eqref{eqn:AQ}.
Next, recalling \eqref{eqn:price} and \eqref{eqn:dJ} we have that
\[dP^M_i(t) = dS_i^M(t) + \partial_x h_i(t,J_i(t))dJ_i^M(t) = dS_i^M(t) + \partial_x h_i(t,J_i(t))K_i(t,t)dQ_i^M(t). \] Substituting in \eqref{eqn:Q^M}, writing in vector form, and collecting the $dP^M(t)$ terms on the left-hand side yields $A^P(t,P(t),J(t))dP^M(t) = dS^M(t)$, which gives
\begin{equation} \label{eqn:P^M}
    dP^M(t) = \beta^P(t,P(t),J(t))dS^M(t),
\end{equation}
where $\beta^P = (A^P)^{-1}$. In particular, by substituting into \eqref{eqn:Q^M}, we obtain the SDE coefficient $\beta^Q$ and by computing the quadratic variation from \eqref{eqn:P^M}, we obtain
\begin{equation} \label{eqn:[P]}
d[P](t) = \beta^P(t,P(t),J(t))d[S](t)\beta^P(t,P(t),J(t))^\top.
\end{equation}

Next, we compute $Q^{FV}$ by applying It\^o's formula to the right-hand side of \eqref{eqn:frictionless_additive_holdings}. In the following, we omit the function evaluations for brevity,
\begin{align*}dQ_i^{FV} & = \frac{wN_i}{\overline \bfp(0)}\bigg(\partial_{it}G dt + \sum_{j=1}^d \partial_{ij}Gd\mu_j^{FV} + \frac{1}{2}\sum_{j,k=1}^d \partial_{ijk}Gd[\mu_j,\mu_k] + \partial_t G dt + \sum_{j=1}^d \partial_{j}Gd\mu_j^{FV} \\
& \qquad + \frac{1}{2}\sum_{j,k=1}^d \partial_{jk}Gd[\mu_j,\mu_k] -\sum_{\ell=1}^d \partial_{\ell t}G\mu_\ell dt - \sum_{j=1}^d\Big(\sum_{\ell=1}^d \partial_{j\ell}G\mu_\ell +\partial_j G\Big)d\mu_j^{FV}\\
& \qquad - \frac{1}{2}\sum_{j,k=1}^d\Big(\sum_{\ell=1}^d\partial_{jk\ell}G\mu_\ell + 2 \partial_{jk}G\Big)d[\mu_j,\mu_k] -\partial_tGdt -\frac{1}{2}\sum_{j,k=1}^d\partial_{jk}Gd[\mu_j,\mu_k] \\
& \qquad+ \frac{1}{2}\sum_{\ell=1}^d \frac{wN_\ell^2}{\overline \bfp(0)\overline \bfp}\partial_xh_\ell K_\ell\sum_{j,k=1}^d \Big(\partial_{j\ell}G - \sum_{m=1}^d \partial_{j m}G\mu_m\Big)\Big(\partial_{k\ell}G - \sum_{m=1}^d \partial_{km}G\mu_m\Big)d[\mu_j,\mu_k]\bigg) \\
& = \frac{wN_i}{\overline \bfp(0)}\bigg(\Big(\partial_{it}G  - \sum_{\ell=1}^d \partial_{\ell t}G\mu_\ell\Big)dt + \sum_{k=1}^d \Big(\partial_{ik}G - \sum_{\ell=1}^d \partial_{k\ell}G\mu_\ell\Big)d\mu_k^{FV} \\
 &\qquad \quad + \frac{1}{2}\sum_{j,k=1}^d \Big(\partial_{ijk}G - 2\partial_{jk}G- \sum_{\ell=1}^d \partial_{jk\ell}G \mu_\ell  \\
 & \qquad \qquad \quad   + \sum_{\ell=1}^d \frac{wN_\ell^2}{\overline \bfp(0) \overline \bfp}\partial_xh_\ell K_\ell \Big(\partial_{j\ell}G - \sum_{m=1}^d \partial_{j m}G\mu_m\Big)\Big(\partial_{k\ell}G - \sum_{m=1}^d \partial_{k m}G\mu_m\Big)\Big) d[\mu_j,\mu_k]\bigg) \\
 & = \frac{wN_i}{\overline \bfp(0)}\bigg((\overrightarrow{\nabla \partial_tG})^{\mu}_idt + \sum_{k=1}^d (\overrightarrow{\nabla \partial_kG})^{\mu}_id\mu_k^{FV} + \frac{1}{2}\sum_{j,k=1}^d(\overrightarrow{\nabla \partial_{jk}G})^{\mu}_id[\mu_j,\mu_k]\\
 & \qquad \qquad - \frac{1}{2}\sum_{j,k=1}^d\Big(2\partial_{jk}G -\sum_{\ell=1}^d \frac{wN_\ell^2}{\overline \bfp(0)\overline \bfp}\partial_x h_\ell K_\ell(\overrightarrow{\nabla \partial_jG})^{\mu}_\ell(\overrightarrow{\nabla \partial_kG})^{\mu}_\ell\Big)d[\mu_j,\mu_k]\bigg).
 \end{align*}
Using the identities
 \begin{align*}
     d\mu_k^{FV} & = \frac{N_kdP_k^{FV}}{\overline \bfp} -\mu_k\sum_{j=1}^d\frac{N_jdP^{FV}_j}{\overline \bfp} + \mu_k \frac{d[\overline \bfp]}{\overline \bfp^2} -\frac{d[\bfp_k,\overline \bfp]}{\overline \bfp^2}, \\
     d[\mu_j,\mu_k] & = \frac{1}{\overline \bfp^2}\big(d[\bfp_j,\bfp_k] - \mu_jd[\bfp_k,\overline \bfp]-\mu_kd[\bfp_j,\overline \bfp]+\mu_j\mu_kd[\overline \bfp]\big),
 \end{align*} we obtain
 \begin{align}
     dQ_i^{FV} & =  \frac{wN_i}{\overline \bfp(0)}(\overrightarrow{\nabla \partial_tG})^{\mu}_idt + \sum_{j=1}^d\frac{wN_iN_j}{\overline \bfp(0)\overline \bfp}\sum_{k=1}^d(\overrightarrow{\nabla \partial_kG})^{\mu}_i(\delta_{jk}-\mu_k)dP^{FV}_j   \nonumber \\
     & \quad  +\frac{wN_i}{\overline \bfp(0)} \sum_{k=1}^d\Big(\mu_k (\overrightarrow{\nabla \partial_kG})^{\mu}_i\sum_{\ell, m=1}^d\frac{N_\ell N_m}{\overline \bfp^2}d[P_\ell,P_m] - (\overrightarrow{\nabla \partial_kG})^\mu_i\sum_{\ell=1}^dN_kN_\ell\frac{d[P_k,P_\ell]}{\overline \bfp^2}\Big) \nonumber\\
     & \qquad \qquad  
     + \frac{wN_i}{2\overline \bfp(0)\overline \bfp^2}\sum_{j,k=1}^d \Big((\overrightarrow{\nabla \partial_{jk}G})^{\mu}_i - 2\partial_{jk}G + \sum_{\ell=1}^d \frac{wN_\ell^2}{\overline \bfp(0)\overline \bfp}\partial_x h_\ell K_\ell(\overrightarrow{\nabla \partial_jG})^{\mu}_\ell(\overrightarrow{\nabla \partial_kG})^{\mu}_\ell\Big) \nonumber\\
     & \qquad \qquad \qquad \qquad \qquad \qquad \times \bigg(N_jN_kd[P_j,P_k] - \sum_{m=1}^d\mu_jN_kN_md[P_k,P_m] \nonumber \\
     & \qquad \qquad \qquad \qquad \qquad \qquad \qquad  -\sum_{m=1}^d\mu_kN_jN_md[P_j,P_m] + \sum_{m,n=1}^d\mu_j\mu_kN_mN_nd[P_m,P_n]\bigg) \nonumber\\
     &  =   \frac{wN_i}{\overline \bfp(0)}(\overrightarrow{\nabla \partial_tG})^{\mu}_idt + \sum_{j=1}^dA^Q_{ij}dP^{FV}_j  + \sum_{j,k=1}^d\Upsilon_{ijk}d[P_j,P_k], \label{eqn:Q^FV}
     \end{align}
     where in the last equality we group like terms and recall the notation introduced in \eqref{eqn:AQ} and \eqref{eqn:upsilon}.
     We can now obtain the full dynamics of $P$,
     \begin{align*}
     dP_i(t) & = dS_i(t) + dI_i(t) = dS_i(t) + \partial_t h_i(t,J_i(t))dt + \partial_x h_i(t,J_i(t))dJ_i(t) + \frac{1}{2}\partial_{xx}h_i(t,J_i(t))d[J_i](t) \\
     & = dS_i(t) + \big(\partial_t h_i(t,J_i(t)) +\partial_x h_i(t,J_i(t))b_i^J(t,Q_i^{[0,t]})\big)dt + \partial_x h_i(t,J_i(t))K_i(t,t)dQ_i(t) \\
     & \qquad \qquad \qquad \qquad \qquad\qquad \qquad \qquad \qquad \qquad + \frac{1}{2}\partial_{xx}h_i(t,J_i(t))K_i^2(t,t)d[Q_i](t),
     \end{align*}
     where we used \eqref{eqn:dJ} to expand the dynamics of $J$. 
     Substituting $dQ_i$ and $d[Q_i]$, using \eqref{eqn:Q^M} and \eqref{eqn:Q^FV}, and collecting $dP$ terms on the left-hand side, we obtain
\begin{align*}
         \sum_{j=1}^dA^P_{ij}(t,P(t),J(t))dP_j(t) & =  dS_i(t) + \partial_x h_i(t,J_i(t)) K_i(t,t) \sum_{j,k=1}^d  \Upsilon_{ijk}(t,P(t),J(t))d[P_j,P_k](t) \\
         & \hspace{-3.5cm}  + \Big(\partial_t h_i(t,J_i(t)) + \partial_x h_i(t,J_i(t))b^J_i(t,Q_i^{[0,t]}) + \frac{wN_i}{\overline \bfp(0)} \partial_x h_i(t,J_i(t)) K_i(t,t)  (\overrightarrow{\nabla \partial_tG})^{\mu(t)}_i(t,\mu(t))\Big)dt \\
         &  + \frac{1}{2}\partial_{xx}h_i(t,J_i(t))K_i^2(t,t)\Big(A^Q(t,P(t),J(t))d[P](t)A^Q(t,P(t),J(t))^\top\Big)_{ii}.
\end{align*}
Finally, substituting $d[P]$ from \eqref{eqn:[P]} and multiplying both sides by $\beta^P = (A^P)^{-1}$ yields \eqref{eqn:SDEP}. Inserting the derived expressions for $dP^M$, $dP^{FV}$ and $d[P]$ into \eqref{eqn:Q^M} and \eqref{eqn:Q^FV} yields \eqref{eqn:SDEQ}. Lastly, \eqref{eqn:SDEJ} follows by substituting the dynamics for $dQ$ into \eqref{eqn:dJ}.

\subsection{Proof of Theorem~\ref{thm:main}} \label{sec:main_proof}
Now that we have formally derived the SDE \eqref{eqn:SDE}, we are ready to establish Theorem~\ref{thm:main}.
\begin{proof}[Proof of Theorem~\ref{thm:main}]
We start by noting that the SDE \eqref{eqn:SDE} can be written in the standard form
\[dX(t) = \Phi(t,X^{[0,t]})dZ(t),\]
where $ X= (P,Q,J)$ is a $3d$-dimensional process, $Z = (t,S,[S])$ is the driving $(1+d+d^2)$-dimensional semimartingale and $\Phi:[0,\infty) \times C(\R^{3d}) \to  \R^{3d \times (1+d+d^2)}$ can be explicitly computed in terms of the SDE coefficients \eqref{eqn:SDE_coefficients}. 

Define the domains 
\[D_n = \{p \in (0,\infty)^d: p_i > 1/n \text{ and }\mu_i(p) > 1/n \text{ for every } i\} \times \{q \in \R^d: |q_i| < n \text{ for every } i\} \times \R^d ,\] and the corresponding exit times 
\[\xi_n = \inf\{t \geq 0: X(t) \not \in D_n\}\]
for every $n \in \N$. By inspecting the coefficients \eqref{eqn:SDE_coefficients}, recalling Assumption~\ref{ass:impact_inputs}, which ensures boundedness of the derivatives of $h$, as well as the hypothesis on $G$ given in the theorem statement, it is easily verified that $\Phi$ is bounded and Lipschitz continuous on $[0,T]\times \mathcal{D}_n$ for every $T \geq 0$, where
\[\mathcal{D}_n = \{ \omega \in C(\R^{3d}): \omega(t)\in D_n \text{ for every }t \geq 0\}.\]It\^o's existence and uniqueness theorem (see e.g., \cite[Theorem~IX.2.1]{revuz1999continuous}) together with a standard localization argument yields a pathwise unique solution to \eqref{eqn:SDE} on $[0,T \land \xi_n)$. Sending $T$ and $n$ to infinity yields a pathwise unique solution on $[0,\xi)$, where $\xi = \lim_{n \to \infty} \xi_n$. To complete the proof of \ref{item:SDE}, it suffices to establish the claimed representation for $\xi$. We first establish items \ref{item:right_process} and \ref{item:Q_additive} and then return to this step.

That $P,Q$ and $J$ are indeed the price, holdings and impact state processes from our financial model follows directly by reversing the calculations in Appendix~\ref{app:SDE_derivation} and Section~\ref{sec:heuristic}. These are now rigorous on the lifetime $[0,\xi)$, as we start from the processes constructed in item \ref{item:SDE} and through manipulations mostly involving It\^o's formula, obtain the representations $P = S + I$, where $I$ is given by \eqref{eqn:impact} and $Q$ is given by \eqref{eqn:frictionless_additive_holdings} with $\Gamma$ defined in \eqref{eqn:Gamma}. It is clear that the solution to \eqref{eqn:SDE} is the unique such semimartingale process $(P,Q,J)$. Indeed, by the calculations in Section~\ref{sec:heuristic} and Appendix~\ref{app:SDE_derivation}, any triple $(P,Q,J)$ satisfying \eqref{eqn:price}, \eqref{eqn:impact} and \eqref{eqn:frictionless_additive_holdings} must also satisfy the SDE \eqref{eqn:SDE}, which by pathwise uniqueness must then coincide with the processes obtained in \ref{item:SDE}. This proves item~\ref{item:right_process}. 

Item \ref{item:Q_additive} follows by noting that the holdings process $Q_i$  given by \eqref{eqn:frictionless_additive_holdings} is of the form $Q_i(t) = \frac{wN_i}{\overline \bfp(0)}(\partial_i  G(t,\mu(t)) + C(t))$ for $t \in [0,\xi)$, where $C(t) = 1 + G(t,\mu(t)) - G(0,\mu(0)) - \nabla G(t,\mu(t))^\top \mu(t) + \Gamma(t)$ is a scalar process common to all $i$. Substituting into the relative wealth equation \eqref{eqn:relative_wealth} and using the fact that $\sum_{i=1}^d d\mu_i(t) = 0$ yields
\[dV^Q(t) = \nabla G(t,\mu(t))^\top d\mu(t) + \frac{1}{2}\sum_{i=1}^d\frac{\overline \bfp (0)}{w\overline \bfp(t)}d[I_i,Q_i](t).\] The ensuing computations in Section~\ref{sec:heuristic} now rigorously lead to \eqref{eqn:additive_master_impact}, establishing that $Q$ is additively functionally generated by $G$ on $[0,\xi)$.

We now return to establishing the representation for $\xi$ given by \eqref{eqn:xi}.
From the definition of $\xi_n$ and $D_n$ above, it is clear that $\xi$ corresponds to the first time that $P_i(t) = 0$ or $\mu_i(t) = 0$ or $|Q_i(t)| = + \infty$. Since $\mu_i(t) = N_iP_i(t)/\overline \bfp(t)$, we have the equality of events
\[\{\mu_i(t) = 0 \text{ or } P_i(t) = 0\} = \{\mu_i(t) = 0 \text{ or } \overline \bfp(t) = 0\},\]
and we also note that if $\overline \bfp(t) = 0$ then $\mu_i(t)$ (understood in this degenerate case as $\limsup_{s \uparrow t} \mu_i(s)$) may be positive.
 As such, to establish the representation of $\xi$ and complete the proof of item~\ref{item:SDE}, we will show that if $|Q_i(\xi)| = +\infty$, then either $\mu_j(\xi) = 0$ for some index $j$ or $\overline \bfp(\xi) = 0$.

We begin by expanding the quadratic variation terms
 in \eqref{eqn:Gamma}. Note that since the diffusion coefficient $\beta^P$ is a function of $t, P(t)$ and $J(t)$, but not of $Q^{[0,t]}$, it follows that $d[\mu](t) = \Psi(t,P(t),J(t))d[S](t)$ for some $\R^{d\times d}$-valued function $\Psi$. Hence, from \eqref{eqn:frictionless_additive_holdings} and the expansion of the quadratic variation terms in \eqref{eqn:Gamma}, we obtain that $Q_i$ has the representation
 \begin{align*} 
 Q_i(t) = & \frac{wN_i}{\overline \bfp(0)}\bigg(\partial_i G(t,\mu(t)) + 1 + G(t,\mu(t)) - G(0,\mu(0)) - \nabla G(t,\mu(t))^\top \mu(t) \\
 & \qquad - \int_0^t \partial_t G(s,\mu(s))ds + \int_0^t\sum_{j,k=1}^d\widetilde\Psi_{jk}(s,P(s),J(s))d[S_j,S_k](s)\bigg), 
\end{align*} 
where $\widetilde \Psi$ is some $\R^{d \times d}$ valued function which aggregates all of the quadratic variation terms. Although $\widetilde \Psi$ can be explicitly computed similarly to the derivation of \eqref{eqn:Q^FV}, its explicit form is not needed here. We only note, by the boundedness of $G$ and its derivatives on $[0,T] \times \{\mu: \mu_i \geq 1/n \text{ for every } i\}$, as well as the boundedness of $\overline \bfp^{-1}(p)$ on $\{p:p_i \geq 1/n\}$, that each element $\widetilde \Psi_{jk}$ is bounded on $[0,T] \times D_n$. Hence we obtain the estimate
\begin{align*}
 \sup_{t \leq T \land \xi_n}|Q_i(t)| & \leq \frac{wN_i}{\overline \bfp(0)}\bigg(1+ \sup_{t \leq T \land \xi_n}\big|\partial_i G(t,\mu(t))\big| + 2\sup_{t \leq T \land \xi_n}\big|G(t,\mu(t))\big| + \sup_{t \leq T \land \xi_n}\big|\nabla G(t,\mu(t))^\top \mu(t)\big| \\
    & \qquad  + \sup_{t \leq T \land \xi_n}\big|\partial_t G(t,\mu(t))\big|T + \sum_{j,k=1}^d \sup_{t \leq T \land \xi_n}\big|\widetilde \Psi_{jk}(t,P(t),J(t))\big|\sqrt{[S_j](T)[S_k](T)}\bigg),
\end{align*}
where we used the Kunita--Watanabe Inequality \cite[Corollary~IV.1.16]{revuz1999continuous}  to handle the quadratic covariation of $S$ terms. The right-hand side only (possibly) tends to $+\infty$ as $n \to \infty$ if $\mu_j(t) \to 0$ for some $j$ or $\overline \bfp(t) \to 0$. Hence $\lim_{n \to \infty} \sup_{t \leq T \land \xi_n} |Q_i(t)| = \infty \implies \mu_j(\xi) = 0$ for some $j$ or $\overline \bfp(\xi) = 0$. This completes the proof of item~\ref{item:SDE}. 

To establish item~\ref{item:long-only}, we recall that the $\Gamma$ process of this additively functionally generated portfolio is given by \eqref{eqn:Gamma}. 
The first two terms on the right-hand side of \eqref{eqn:Gamma} are nonnegative because $G(\cdot,\mu)$ is nonincreasing and $G(t,\cdot)$ is concave. The final term is always nonnegative, yielding nonnegativity of $\Gamma$. Since $G(0,\mu(p^0)) \leq 1$, we have $V^Q(t) \geq G(t)$, establishing that $V^Q$ is nonnegative if $G$ is nonnegative. The nonnegativity of $Q_i$ now follows by the same argument as in the proof of Proposition~\ref{prop:additive_frictionless}, establishing item~\ref{item:long-only}. 

It remains to prove item \ref{item:positive_price}. Recalling the function $F_i$ defined in item~\ref{item:positive_price}, it follows from  \eqref{eqn:frictionless_additive_holdings} that for every $i$
\begin{align*}
    Q_i(t)-Q_i(s)  = \frac{wN_i}{\overline \bfp(0)}\big(F_i(\mu(t))-F_i(\mu(s)) + \Gamma(t) - \Gamma(s)\big),
\end{align*} 
 Hence, setting $\tau_n = \inf\{t \geq 0: F_i(\mu(t)) \geq\overline F_i - 1/n\}$, we see on the set $\{\tau_n < \infty\}$ that
\[Q_i(\tau_n) - \sup_{s \leq \tau_n} Q_i(s) \geq \frac{wN_i}{\overline \bfp(0)}\Big(F_i\big(\mu(\tau_n)\big)-\sup_{s \leq \tau_n} F_i\big(\mu(s)\big) + \Gamma(\tau_n) - \sup_{s \leq \tau_n} \Gamma(s)\Big) = 0,\] where the final equality follows from the definition of $\tau_n$ and the fact that $\Gamma$ is increasing. Let $\sigma = \lim_{n \to \infty} \tau_n$ and note that by condition \eqref{eqn:positive_condition} we have $\P$-a.s.\ that $\sigma \leq \inf\{t \geq 0: \mu_i(t) = 0 \}$. Since $P_i(t) = 0 \iff \mu_i(t) = 0$ on the set $\{0 < \overline \bfp(t) < \infty\}$, it follows from Proposition~\ref{prop:positive} and Remark~\ref{rem:positive} that $P_i$, and hence $\mu_i$, remains positive on $\{0 < \overline \bfp < \infty\}$.   This yields the characterization of the explosion time $\xi$ and completes the proof.
\end{proof}

\section{Proofs for Section~\ref{sec:relarb}} \label{app:relarb_proofs}
The purpose of this section is to prove Lemma~\ref{lem:relarb},
which we decompose into several smaller lemmas. Throughout this section, we let the setup and assumptions of Lemma~\ref{lem:relarb} hold and freely use the notation defined therein.
 We also introduce additional notation, including the constant $\overline \nu$ appearing in the statement of Lemma~\ref{lem:relarb}. In this section $\|x\|_p$ denotes the $\ell^p$ norm of a vector $x \in \R^d$ for $p \in [1,\infty]$. When $p = 2$, we simply write $\|x\|$ without the subscript. Note that $\widehat h_i(\cdot)$ is nondecreasing and, as such, admits a pseudoinverse that we take to be \[\widehat h_i^{-1}(x) = \begin{cases}
    \inf\{y \geq 0: h_i(y) = x\}, & x \geq 0, \\
    \sup\{y \leq 0: h_i(y) = x\}, & x < 0. 
\end{cases}\] Additionally, Assumption~\ref{ass:impact_inputs} ensures that $\widehat h_i'(\cdot)$ is bounded, which implies that $\overline {\widehat h'_i} := \sup_{ x \in \R} \widehat h_i'(x) < \infty$.
We now define
\begin{equation} \label{eqn:bar_nu}
    \overline \nu = \min_{i \in \{1,\dots,d\}}\min\{\nu^1_i,\nu^2_i,\nu^3\},
\end{equation} 
where 
\begin{align*}
 \nu^1_i & = -\frac{\overline \bfp(0)}{wN_i\overline K_i(\overline H + \overline F_i)}\widehat h^{-1}_i\Big(\frac{-\kappa_S\ell_S}{2\overline \phi_iN_i}\frac{\delta_S}{1-\delta_S/2}\Big), \\
 \nu^2_i & = \frac{\overline \bfp(0)}{wN_i\overline K_i(CT_0 + 2\overline H + \overline F_i)}\widehat h^{-1}_i\Big(\frac{\kappa_P \delta_S}{4\overline \phi_iN_i}\Big), \\
 \nu^3 & =\frac{1}{2}\Big(\frac{8w\|N\|_1\|N\|_\infty\|\overline{\widehat h'}\|_\infty \|\overline K\|_\infty \overline D}{\overline \bfp(0) \kappa_P}\Big)^{-1}\Big(1+\max_{i \in \{1,\dots,d\}}\nu^1_i\frac{8w\|N\|_1\|\overline{\widehat h'}\|_\infty \|\overline K\|_\infty \overline D}{\overline \bfp(0) \kappa_P}\Big)^{-1},
\end{align*}
with 
\[
 C  = \bigg(d\overline D +\frac{4d^2\overline D^2(\min_{i\in \{1,\dots,d\}}\nu^1_i)w\|N^2\|_\infty\|\overline \phi\|_\infty\|\overline{\widehat h'}\|_\infty\|\overline K\|_\infty }{\overline \bfp(0)\kappa_P}\bigg) \frac{d^2(1+d)\sigma_S^2\|\frac{1}{N}\|_1^2\|N\|^2 \kappa_S^2}{\kappa_P^2}.
\]
Finally, we define 
\begin{align*}
 \overline I_i & = \overline \phi_i\widehat h_i\Big( \nu\frac{wN_i\overline K_i}{\overline \bfp(0)}(CT_0 + 2\overline H + \overline F_i)\Big),  && \overline \bfi = \sum_{i=1}^d N_i\overline I_i,  \\
 \underline I_i & = - \overline \phi_i \widehat h_i\Big(- \nu\frac{wN_i\overline K_i}{\overline \bfp(0)}(\overline H + \overline F_i)  \Big) && \underline \bfi = \sum_{i=1}^d N_i\underline I_i.
 \end{align*}
Note that $\underline I_i$ is nonnegative since $\widehat h_i(x) \leq 0$ whenever $ x \leq 0$.

We begin the proof by establishing lower bounds on the impact and capitalization processes.
\begin{lem} \label{lem:impact_bound}
   The bounds
    \[\overline \bfp(t) \geq \kappa_P, \qquad I_i(t) \geq - \underline I_i \]
    hold for all $t \in [0,\xi)$ and $i=1,\dots,d$, where $\xi$ is the explosion time of the SDE \eqref{eqn:SDE}.
\end{lem}

\begin{proof}
From \eqref{eqn:frictionless_additive_holdings} we have that \begin{align*}
    Q_i(t) & = \frac{wN_i}{\overline \bfp(0)}\big(1 + \nu\psi(t)F_i(\mu(t))+\Gamma(t)\big),
\end{align*}  
where $\Gamma$ is given by \eqref{eqn:Gamma}.
From \eqref{eqn:impact} we have 
\begin{align}
    J_i(t) & = \frac{wN_i}{\overline \bfp(0)}\int_0^t K_i(t,s)d\Gamma(s) + \frac{wN_i\nu}{\overline \bfp(0)}\int_0^t K_i(t,s)d(\psi F_i(\mu))(s),\nonumber \\
    &  = \frac{wN_i}{\overline \bfp(0)}\int_0^t K_i(t,s)d\Gamma(s) + \frac{wN_i\nu}{\overline \bfp(0)}\bigg(K_i(t,t)\psi(t)F_i(\mu(t)) - \int_0^t \partial_s K_i(t,s)\psi(s)F_i(\mu(s))ds\bigg), \label{eqn:J_relarb}
\end{align}
where in the final step we used $\psi(0) = 0$.

We now work towards bounding $J_i$ from below. Recalling the concavity of $H$ and that $0 \leq \psi(t) \leq 1$, we see that all the terms making up $\Gamma$ in \eqref{eqn:Gamma} are increasing, with the exception of the term $t \mapsto- \int_0^{t\land T_0} \partial_t G(s,\mu(s))ds$. As such, keeping only this term in the integral against $\Gamma$ and using the bounds $0 < \underline F_i \leq F_i(\mu) \leq \overline F_i$ to estimate the second integral in \eqref{eqn:J_relarb}, we obtain the lower bound
\begin{align*} 
J_i(t) & \geq -\frac{wN_i\nu}{\overline \bfp(0)}\bigg(\int_0^{t \land T_0} K_i(t\land T_0,s)\psi'(s)H(\mu(s))ds - \overline F_i\int_0^t \partial_s K_i(t,s)ds\bigg) \\
& \geq  -\frac{wN_i\nu}{\overline \bfp(0)}\big(\overline K_i\overline H(\psi(t\land T_0) - \psi(0)) - \overline F_i(K_i(t,t) - K_i(t,0))\big)  \\
& \geq -\overline \nu\frac{wN_i \overline K_i }{\overline \bfp(0)}(\overline H + \overline F_i). 
\end{align*}
Applying $\widehat h_i(\cdot)$ to both sides and bounding $\phi_i(t)$ by $\overline \phi_i$ gives the lower bound  $I_i(t) \geq -\underline I_i$. 
As a consequence, we see that 
$P_i(t) \geq S_i(t) -\underline I_i$. Since $\nu \leq \overline \nu \leq \nu^1_i$ we have that
$\underline I_i \leq \frac{\kappa_S\ell_S\delta_S}{2N_i(1-\delta_S/2)}$,
from which it follows that 
\[\underline \bfi = \sum_{i=1}^d N_i\underline I_i \leq \frac{d\kappa_S\ell_S\delta_S}{2(1-\delta_S/2)} \leq \frac{\kappa_S\delta_S}{2(1-\delta_S/2)},\]
where we used the fact that $d\ell_S \leq 1$.
Hence, we see that 
\begin{align}
    \overline \bfp(t) & = \overline \bfs(t) +\sum_{i=1}^d N_iI_i(t) \geq \kappa_S - \underline \bfi 
     \geq \kappa_S - \frac{\kappa_S\delta_S}{2(1-\delta_S/2)} = \kappa_S\frac{1-\delta_S}{1-\delta_S/2} = \kappa_P.      \label{eqn:bfp_lower_bound_S}
\end{align}
 This completes the proof. 
    \end{proof}
 
\begin{lem} \label{lem:Gamma_bounds}
    The process $\Gamma$ satisfies the bounds 
    \begin{equation}
        \label{eqn:Gamma_bounds}
        -\nu \overline H \leq \Gamma(t) \leq \nu(\overline H + CT_0), \qquad \text{for }t \in [0,T_1 \land \xi).
    \end{equation}
        Moreover, $\Gamma(t) \leq \nu(\overline H+Ct)$ for all $t \in [0,\xi)$.

    \end{lem}
\begin{proof}
    As in the proof of Lemma~\ref{lem:impact_bound}, note that the terms making up $\Gamma$ in \eqref{eqn:Gamma} are all increasing with the exception of $t \mapsto -\int_0^{t \land T_0}\partial_t G(s,\mu(s))ds$, which is decreasing. As such, we have the lower bound
    \begin{equation} \label{eqn:Gamma_lower_bound} \Gamma(t) \geq - \int_0^{t \land T_0} \partial_t G(s,\mu(s))ds = -\nu \int_0^{t \land T_0} \psi'(s)H(\mu(s))ds \geq - \nu \overline H(\psi(t \land T_0) - \psi(0)) \geq - \nu \overline H,
    \end{equation}
    which establishes the lower bound. 

    For the upper bound, we instead work with the other terms in \eqref{eqn:Gamma},
    \begin{align} 
    \Gamma(t) & \leq - \nu  \int_{T_1}^{T_1 \lor t} \psi'(s)H(\mu(s))ds  - \frac{\nu}{2}\int_0^t\sum_{i,j=1}^d \psi(s)\partial_{ij}H(\mu(s))d[\mu_i,\mu_j](s) \nonumber \\
& \qquad + \frac{\nu^2}{2}\int_0^t \sum_{i=1}^d \frac{wN_i^2}{\overline \bfp(0) \overline \bfp(s)}\phi_i(s)\widehat h_i'(J_i(s))K_i(s,s)\psi^2(s)d[F_i(\mu)](s) \nonumber  \\
\leq &  -\nu \overline H\int_{T_1}^T \psi'(s)ds + \frac{\nu \overline D }{2}\sum_{i,j=1}^d [\mu_i,\mu_j](t) + \nu^2 \frac{w\|N^2\|_\infty\|\overline \phi\|_\infty\|\overline{\widehat h'}\|_\infty\|\overline K\|_\infty}{2\overline \bfp(0)\kappa_P}\sum_{i=1}^d[F_i(\mu)](t) \nonumber \\
\leq & \nu\Big( \overline H + \frac{d\overline D}{2}\sum_{i=1}^d [\mu_i](t) + \frac{(\min_{i\in \{1,\dots,d\}} \nu_i^1) w\|N^2\|_\infty\|\overline \phi\|_\infty\|\overline{\widehat h'}\|_\infty\|\overline K\|_\infty}{2\overline \bfp(0)\kappa_P}\sum_{i=1}^d[F_i(\mu)](t)\Big),\label{eqn:Gamma_QV_bound}
\end{align}
 where in the second inequality we used $|\partial_{ij}H| \leq \overline D$ to bound the second term, and bounded each factor in the integrand of the third term by its supremum. In the final inequality we factored out $\nu$, noted that $-\int_{T_1}^T \psi'(s)ds = \psi(T_1) - \psi(T) = 1$, used the bound $\sum_{i,j=1}^d [\mu_i,\mu_j](t) \leq d \sum_{i=1}^d [\mu_i](t)$, which is a consequence of the Cauchy--Schwarz inequality, and additionally bounded the remaining $\nu$ in front of the third term by $\min_{i  \in \{1,\dots,d\}}\nu^1_i$.

Next, recalling the notation \eqref{eqn:arrowM}, we compute 
\begin{equation} \label{eqn:QV_Fi_bound}
\begin{aligned}
\sum_{i=1}^d [F_i(\mu)](t) = \sum_{i,j,k=1}^d\int_0^t  (\overrightarrow{\nabla \partial_jH})^{\mu(s)}_i(\mu(s))(\overrightarrow{\nabla \partial_kH})^{\mu(s)}_i(\mu(s))d[\mu_j,\mu_k](s) & \leq 4\overline D^2 \sum_{i,j,k=1}^d[\mu_j,\mu_k](t) \\
& \leq 4d^2\overline D^2 \sum_{i=1}^d [\mu_i](t),
\end{aligned}
\end{equation}
where in the first inequality we used the bound $|\partial_{ij}H(\mu(s))| \leq \overline D$ to bound each second partial derivative appearing in $(\overrightarrow{\nabla \partial_jH})^{\mu(s)}_i(\mu(s))$ and $(\overrightarrow{\nabla \partial_kH})^{\mu(s)}_i(\mu(s))$, and in the second inequality we again used $\sum_{j,k=1}^d[\mu_j,\mu_k](t) \leq d \sum_{i=1}^d [\mu_i](t)$.

Next, we bound the quadratic variation of the market weights, \begin{equation} \label{eqn:QV_mu_bound}
    d[\mu_i]= d\bigg[\int_0^\cdot \frac{d\bfp_i}{\overline \bfp} - \int_0^\cdot \frac{\mu_i}{\overline \bfp}d\overline \bfp\bigg] \leq \frac{2}{\overline \bfp^2}(d[\bfp_i] + \mu_i^2d[\overline \bfp]) \leq \frac{2}{\overline \bfp^2}(N_i^2d[P_i] + d\times \mu_i^2\sum_{j=1}^d N_j^2d[P_j]).
\end{equation}
Recalling the dynamics of $P$ given by \eqref{eqn:SDEP}, we see, omitting function evaluations for brevity, that
\begin{align*}
    d[ P_i]& \leq \sum_{j=1}^d d[P_j]=  \mathrm{Tr}(\beta^Pd[S](\beta^P)^\top) \\ 
    & = \mathrm{Tr}\Big((\beta^P)^\top \beta^P\mathrm{diag}\Big(\frac{1}{N}\Big)\mathrm{diag}(\bfs (t))d[\log S](t)\mathrm{diag}(\bfs (t))\mathrm{diag}\Big(\frac{1}{N}\Big)\Big)  \\
    & \leq \mathrm{Tr}((\beta^P)^\top\beta^P)\overline \bfs^2\Big(\sum_{j=1}^d \frac{1}{N_j}\Big)^2\sigma^2 dt  \leq d\sigma^2\Big\|\frac{1}{N}\Big\|_1^2 \overline \bfs^2dt,
    \end{align*}
where in the second equality we used the cyclical property of trace and rewrote $d[S]$ in terms of $d[\log S]$. In the final two inequalities we used the bound $\mathrm{Tr}((\beta^P)^\top\beta^P) \leq d$, which follows from the fact that the eigenvalues of $(\beta^P)^\top\beta^P$ are bounded by one courtesy of Proposition~\ref{prop:AP}. Substituting this into \eqref{eqn:QV_mu_bound} and writing the expression in integral form yields the bound 
\begin{align}
[\mu_i](t) & \leq 2d\sigma^2\Big\|\frac{1}{N}\Big\|_1^2\bigg(N_i^2\int_0^t \frac{\overline \bfs^2(s)}{\overline \bfp^2(s)}ds +\|N\|^2d \int_0^t \mu_i^2(s)\frac{\overline \bfs^2(s)}{\overline \bfp^2(s)}ds\bigg) \nonumber \\
& \leq 2d(1+d)\sigma^2\Big\|\frac{1}{N}\Big\|_1^2\|N\|^2\int_0^t \frac{\overline \bfs^2(s)}{\overline \bfp^2(s)}ds  
 \leq 2d(1+d)\sigma^2\Big\|\frac{1}{N}\Big\|_1^2\|N\|^2\int_0^t \frac{\overline \bfs^2(s)}{(\overline \bfs(s)-\underline \bfi)^2}ds \nonumber  \\
& \leq  2d(1+d)\sigma^2\Big\|\frac{1}{N}\Big\|_1^2\|N\|^2 \frac{\kappa_S^2}{(\kappa_S-\underline \bfi)^2}t \leq  2d(1+d)\sigma^2\Big\|\frac{1}{N}\Big\|_1^2\|N\|^2 \frac{\kappa_S^2}{\kappa_P^2}t \label{eqn:QV_mu_final_bound},
\end{align}    
where in the second-to-last inequality we used that the function $x \mapsto \frac{x}{x-c}$ for $c > 0$ is decreasing on $(c,\infty)$ and in the final inequality we used \eqref{eqn:bfp_lower_bound_S}. From \eqref{eqn:Gamma_QV_bound}, \eqref{eqn:QV_Fi_bound} and \eqref{eqn:QV_mu_final_bound}, we obtain the bound 
\[\Gamma(t) \leq \nu(\overline H +Ct),\] 
as required. Moreover, it is clear that for $t \in [0,T_0]$, this implies the second inequality in \eqref{eqn:Gamma_bounds}. However, by definition of $T_1$, we also have that $\Gamma(t) \leq 0$ for $t \in (T_0,T_1]$. Since the right-hand side of \eqref{eqn:Gamma_bounds} is positive, this completes the proof.
\end{proof}

\begin{lem} \label{lem:nonexplosive} We have $I_i(t) \leq \overline I_i \leq \frac{\delta_S\kappa_P}{4N_i}$ and $\mu_i(t) \geq \ell_P$ for every $i=1,\dots,d$ and $t \in [0,T_1]$. 
  Moreover, Assumption~\ref{ass:nonexplosion} is satisfied.
\end{lem}

\begin{proof} 
We start by obtaining an upper bound for $J_i$.  From \eqref{eqn:J_relarb} and the fact that all terms in $\Gamma$ of \eqref{eqn:Gamma} are increasing except $t \mapsto -\int_0^{t \land T_0} \partial_t G(s,\mu(s))ds$, which is decreasing, we obtain
\begin{align*} 
J_i(t) & =   \frac{wN_i}{\overline \bfp (0)}\bigg(\int_0^t K_i(t,s)d\Big(\Gamma + \int_0^{\cdot \land T_0} \partial_t G(\cdot,\mu)\Big)(s) - \int_0^{t\land T_0} K_i(t,s)\partial_t G(s,\mu(s))ds\bigg) \nonumber \\
& \qquad +  \frac{wN_i\nu}{\overline \bfp(0)}\bigg(K_i(t,t)\psi(t)F_i(\mu(t)) - \int_0^t \partial_s K_i(t,s)\psi(s)F_i(\mu(s))ds\bigg) \nonumber \\
& \leq \frac{wN_i \overline K_i}{\overline \bfp(0)}\Big(\Gamma(t) + \int_0^{t \land T_0} \partial_t G(s,\mu(s))ds\Big) +\frac{wN_i  \nu\overline K_i \overline F_i}{\overline \bfp(0)}, 
\end{align*}
where we also used the nonnegativity of $F_i(\cdot)$ and $\partial_s K(t,s)$. From the estimates in \eqref{eqn:Gamma_lower_bound} we see that  $\int_0^{t \land T_0} \partial_t G(s,\mu(s))ds \leq \nu \overline H$, from which we deduce that
\begin{equation} \label{eqn:Ji_upper_bound}
    J_i(t) \leq \frac{wN_i \overline K_i}{\overline \bfp(0)}(\Gamma(t) + \nu \overline H + \nu \overline F_i).
\end{equation}
From here and the upper bound on $\Gamma(t)$ provided by Lemma~\ref{lem:Gamma_bounds}, we obtain
\[I_i(t) \leq \overline \phi_i \widehat h_i\Big(\nu\frac{wN_i  \overline K_i}{\overline \bfp(0)}(CT_0 + 2\overline H +  \overline F_i)\Big)= \overline I_i \leq \frac{\delta_S\kappa_P}{4N_i}, \qquad t \in [0,T_1],\] 
where in the last inequality we also used that $\nu \leq \overline \nu \leq \nu^2_i$. For general $t$, by substituting the bound $\Gamma(t) \leq \nu(\overline H + Ct)$ from Lemma~\ref{lem:Gamma_bounds}, into the right-hand side of \eqref{eqn:Ji_upper_bound} we can obtain a time-dependent upper bound on $I_i(t)$ for all $t \in [0,\xi)$. In particular, this yields that $\overline \bfp(t) \leq \overline \bfs(t) +  f(t) < \infty$, for an explicit function $f(t)$ whose precise form is not needed here. This establishes that $\overline \bfp$ cannot blow up in finite time.  

We now turn our attention to establishing nonexplosion of the SDE \eqref{eqn:SDE}. Since the total market capitalization $\overline \bfp(t)$ is bounded from below, from the representation \eqref{eqn:xi}, we see that the proof will be complete once we show that $\mu_i$ is positive almost surely for every $i$.
To this end, we use the market weight floor for the frictionless market to see that
    \begin{equation} \label{eqn:mu_lower_bound}
        \mu_i(t) = \frac{\bfs_i(t) + N_iI_i(t)}{\overline \bfp(t)}  = \frac{\bfs_i(t) + N_iI_i(t)}{\overline \bfs(t)}\frac{\overline \bfs(t)}{\overline \bfp(t)} \geq \Big(\ell_S - \frac{N_i\underline I_i}{\kappa_S}\Big)\frac{\overline \bfs(t)}{\overline \bfp(t)}.
    \end{equation}
   Since $\overline \nu \leq \nu_i^1$ we have that
   \begin{equation} \label{eqn:intermediate_bound}
       \frac{N_i\underline I_i}{\kappa_S} \leq \frac{\delta_S}{2(1-\delta_S/2)}\ell_S < \ell_S
   \end{equation} so that $\ell_S - \frac{N_i\underline I_i}{\kappa_S} > 0$. We previously established that $\overline \bfp$ cannot blow up in finite time, so this establishes $\P(\xi = \infty) = 1$.

   It remains to show that $\mu_i(t) \geq \ell_P$ for $t \in [0,T_1]$. Using \eqref{eqn:mu_lower_bound}, the first inequality in \eqref{eqn:intermediate_bound} and that $I_i(t) \leq \overline I_i \leq \frac{\delta_S\kappa_P}{4N_i}$, we estimate that
   \[\mu_i(t) \geq \Big(\ell_S - \frac{N_i\underline I_i}{\kappa_S}\Big)\frac{\overline \bfs(t)}{\overline \bfs(t) + \sum_{j=1}^d N_jI_j(t)} \geq \ell_S\Big(\frac{1-\delta_S}{1-\delta_S/2}\Big)\frac{\kappa_S}{\kappa_S + \overline \bfi} \geq \frac{\ell_S\kappa_P}{\kappa_S + d\delta_S\kappa_P/4} = \ell_P, \]
   where we also used that the map $x \mapsto x/(x+c)$ is increasing for $c > 0$. This completes the proof.
\end{proof}

\begin{lem} \label{lem:diversity}
   The observed market is diverse on $[0,T_1]$ with diversity constant $\delta_P$. 
\end{lem}
\begin{proof}
Using diversity of the frictionless market and Lemma~\ref{lem:impact_bound} we have for $t \in [0,T_1]$, that
\begin{align}
    \mu_i(t)  = \frac{\bfs_i(t) + N_iI_i(t)}{\overline \bfs(t) + \sum_{j=1}^d N_jI_j(t)} &  \leq \frac{(1-\delta_S)\overline \bfs(t) + N_i \overline I_i}{\overline \bfs(t)  - \underline \bfi} \leq \frac{(1-\delta_S)\kappa_S + N_i \overline I_i}{\kappa_S -\underline \bfi} \nonumber \\
    & \leq \frac{(1-\delta_S)\kappa_S + N_i \overline I_i}{\kappa_P} = 1- \frac{\delta_S}{2} + \frac{N_i\overline I_i}{\kappa_P}, \label{eqn:mu_bound}
\end{align}
where in the second-to-last inequality we used the fact that $x \mapsto (ax+b)/(x-c)$ is a decreasing function on $(c,\infty)$ whenever $a,b,c > 0$ and in the final inequality we used the estimates obtained in \eqref{eqn:bfp_lower_bound_S}. 
Using the bound $N_i\overline I_i \leq  \delta_S\kappa_P/4$, courtesy of Lemma~\ref{lem:nonexplosive}, we see from
\eqref{eqn:mu_bound} that $\mu_i(t) \leq 1-\delta_S/4$ for every $i \in \{1,\dots,d\}$ and $t \in [0,T_1]$, which completes the proof.
\end{proof}

\begin{lem} \label{lem:nondegeneracy} The observed market is nondegenerate on $[0,T_1]$ with constant $\epsilon_P$.
\end{lem}
\begin{proof}
    From the dynamics \eqref{eqn:SDEP} and the identity $a^P_{ij}(t)dt = d[\log P_i,\log P_j](t) = \frac{d[P_i,P_j](t)}{P_i(t)P_j(t)}$,
    we compute that
   \begin{equation} \label{eqn:ap}
       a^P(t) = \widetilde \beta^P(t)a^S(t)\widetilde \beta^P(t)^\top,
    \end{equation}
    where
    \[\widetilde \beta^P(t) = \mathrm{diag}(N) \mathrm{diag}\Big(\frac{1}{\bfp(t)}\Big)\beta^P(t,P(t),J(t))\mathrm{diag}(\bfs (t))\mathrm{diag}\Big(\frac{1}{N}\Big).\]

Next, we estimate the largest eigenvalue of $A^P(A^P)^\top$, where $A^P$ is given by \eqref{eqn:AP}. To facilitate this computation, set
\[X_{ij}(t) = \frac{wN_iN_j\phi_i(t)\widehat h_i'(J_i(t))K_i(t,t)}{\overline \bfp(0)\overline \bfp(t)}\] and write $A^P(t)$ for $A^P(t,P(t),J(t))$. With our choice of $G$, we have 
\[A^P_{ij}(t) = \begin{cases} 
    1 + \nu \psi(t) X_{ii}(t)\overleftrightarrow{\nabla^2 H}_{ii}^{\mu(t)}(t,\mu(t)), & \text{if } i = j, \\
    \nu \psi(t) X_{ij}(t)\overleftrightarrow{\nabla^2 H}_{ij}^{\mu(t)}(t,\mu(t)), & \text{if } i \ne j.
\end{cases}\] From here it follows that 
\begin{equation} \label{eqn:AP_squared}(A^P(A^P)^\top)_{ij}(t) = \begin{cases}
    \begin{aligned} 1 + & 2\nu \psi(t) X_{ii}(t)\overleftrightarrow{\nabla^2 H}_{ii}^{\mu(t)}(t,\mu(t)) \\
    & + \nu^2 \psi^2(t)\sum_{k=1}^d X^2_{ik}(t)(\overleftrightarrow{\nabla^2 H}_{ik}^{\mu(t)}(t,\mu(t)))^2,
    \end{aligned} & \text{if } i =j, \\
    \begin{aligned} & \nu \psi(t)(X_{ji}(t)+X_{ij}(t))\overleftrightarrow{\nabla^2 H}_{ij}^{\mu(t)}(t,\mu(t)) \\
    & \quad  + \nu^2\psi^2(t) \sum_{k=1}^d X_{ik}(t)X_{jk}(t)\overleftrightarrow{\nabla^2 H}_{ik}^{\mu(t)}(t,\mu(t)) \overleftrightarrow{\nabla^2 H}_{jk}^{\mu(t)}(t,\mu(t)), 
    \end{aligned}  & \text{if } i \ne j.
\end{cases}
\end{equation}
Looking at the sum of the off-diagonal entries, we obtain
\begin{equation} \label{eqn:R_bound}
    R_i(t) := \sum_{j\ne i} \big|(A^P(A^P)^\top)_{ij}(t)\big| \leq \nu N_i\Big(\frac{8w\|N\|_1\|\overline{\widehat h'}\|_\infty \|\overline K\|_\infty \overline D}{\overline \bfp(0) \kappa_P}\Big)\Big(1+\nu^1_i\frac{8w\|N\|_1\|\overline{\widehat h'}\|_\infty \|\overline K\|_\infty \overline D}{\overline \bfp(0) \kappa_P}\Big),
\end{equation}
where we used the bound $|\overleftrightarrow{\nabla^2 H}_{ij}^{\mu(t)}(t,\mu(t))| \leq 4 \overline D$ and bounded one instance of $\nu$ by $\nu^1_i$.
By the Ger\v{s}gorin circle theorem \cite[Theorem~6.1.1]{horn2013matrix}, we have that each eigenvalue of $A^P(t)$ is within radius $R_i(t)$ of the diagonal entry $A^P_{ii}(t)$ for some $i \in \{1,\dots,d\}$.
Bounding the diagonal entries of \eqref{eqn:AP_squared} using the same approach as for $R_i$ in \eqref{eqn:R_bound}, we obtain
\begin{align*}
    \lambda_{\max}(A^P(A^P)^\top(t)) & \leq \max_{i=1,\dots,d}((A^P(A^P)^\top)_{ii}(t) + \max_{i=1,\dots,d} R_i(t)) \\
     & \hspace{-1.75cm} \leq  1 + 2\nu \Big(\frac{8w\|N\|_1\|N\|_\infty\|\overline{\widehat h'}\|_\infty \|\overline K\|_\infty \overline D}{\overline \bfp(0) \kappa_P}\Big)\Big(1+\max_{i\in \{1,\dots,d\}}\nu^1_i\frac{8w\|N\|_1\|\overline{\widehat h'}\|_\infty \|\overline K\|_\infty \overline D}{\overline \bfp(0) \kappa_P}\Big) \leq 2,
\end{align*}where the last inequality used $\nu \leq \overline \nu \leq \nu^3$.
As such, we have
\[\lambda_{\min}\big(\beta^P(\beta^P)^\top(t,P(t),J(t))\big) = \lambda_{\max}^{-1}(A^P(A^P)^\top (t)) \geq \frac{1}{2}.\]
With this estimate in hand, we are ready to obtain a lower bound on the smallest eigenvalue of $a^P(t)$. Indeed, from \eqref{eqn:ap} and the inequality  
$\lambda_{\min}(BCB^\top) \geq \lambda_{\min}(C)\lambda_{\min}(BB^\top)$, which holds for any matrix $B$ and for a positive semidefinite matrix $C$, we have
\begin{align*}
    \lambda_{\min}(a^P(t)) & \geq \frac{\lambda_{\min}(a^S(t))}{\|N\|^2_\infty\|\frac{1}{N}\|_\infty^2} \Big(\min_{i \in \{1,\dots,d\}}\bfs_i^2(t)\Big)\lambda_{\min}\big(\beta^P(\beta^P)^\top (t,P(t),J(t))\big)\bigg(\min_{i \in \{1,\dots,d\}}\frac{1}{\bfp_i^2(t)}\bigg) \\
    & \geq \frac{\epsilon_S}{2\|N\|_\infty^2\|\frac{1}{N}\|_\infty^2} \Big(\min_{i \in \{1,\dots,d\}}\mu_i^S(t)\Big)^2\bigg(\min_{i \in \{1,\dots,d\}}\frac{1}{\mu_i^2(t)}\bigg) \frac{\overline \bfs^2(t)}{\overline \bfp^2(t)}\\
    & \geq \frac{\epsilon_S \ell_S^2}{2\|N\|^2_\infty\|\frac{1}{N}\|^2_\infty}\frac{\overline \bfs^2(t)}{\big(\overline \bfs(t) + \sum_{i=1}^d N_iI_i(t)\big)^2}  \\
    & \geq \frac{ \epsilon_S \ell_S^2}{2\|N\|^2_\infty\|\frac{1}{N}\|^2_\infty}\frac{\kappa_S^2}{(\kappa_S + \overline \bfi)^2} \geq  \frac{ \epsilon_S \ell_S^2}{2\|N\|^2_\infty\|\frac{1}{N}\|^2_\infty}\frac{\kappa_S^2}{(\kappa_S +  d\delta_S\kappa_P/4)^2}  = \epsilon_P,
\end{align*}
where in the last line we used the fact that $x \mapsto x/(x+c)$ is increasing for $c > 0$ and the bound $\sum_{i=1}^d N_iI_i(t) \leq \overline \bfi \leq  d\delta_S\kappa_P/4$, which holds for $t \in [0,T_1]$ by Lemma~\ref{lem:nonexplosive}. 
This completes the proof.
\end{proof}

\begin{lem} \label{lem:stopping_time}
    We have that $\P(T_1 < T) = 1$.
\end{lem}

\begin{proof}Note that if $T_1 = T_0$, then clearly $T > T_1$ holds since $T > T_0$. It therefore suffices to consider the case $T_1 > T_0$, in which case, by continuity of $\Gamma$, we have $\Gamma(T_1) = 0$ on the set $\{T_1 < \infty\}$.
We argue by contradiction by assuming that  $\P(T \leq T_1) > 0$ and we work on this event in what follows.
    
We proceed in a similar fashion to the proof of Theorem~\ref{thm:relarb_implicit}. First, note that the process $t \mapsto -\frac{1}{2}\int_0^\cdot \sum_{i,j=1}^d \partial_{ij}G(\cdot,\mu)d[\mu_i,\mu_j]$ is strictly increasing on $(0,T_0)$, due to the strict positivity of $\psi$, the strong concavity of $H$ and the ellipticity of $d[\mu]$. As such, we obtain the strict inequality
    \begin{align*}
        \Gamma(T) & > -\int_0^{T_0} \partial_t G(s,\mu(s))ds -\frac{1}{2} \int_{T_0}^{T} \sum_{i,j=1}^d \partial_{ij}G(s,\mu(s))d[\mu_i,\mu_j](s) \\
        & \geq - \nu \overline H + \frac{\nu m}{2}\sum_{i=1}^d [\mu_i](t),
    \end{align*}
    where in the first inequality we also dropped the nonnegative impact term appearing in the second line of \eqref{eqn:Gamma}. The second inequality used \eqref{eqn:Gamma_lower_bound} for the first term and strong concavity of $H$ for the second, as we did previously in \eqref{eqn:VQ_lower_bound}. By Lemmas~\ref{lem:diversity} and \ref{lem:nondegeneracy}, the market is diverse and nondegenerate on $[0,T_1]$. Since we are assuming $T \in [0,T_1]$, arguing as in \eqref{eqn:QV_estimates} and the preceding estimates, we obtain
    \[\Gamma(T) > - \nu \overline H + \frac{\nu m \epsilon_P\delta_P^2}{2d}(T- T_0) = - \nu \overline H + \frac{\nu m \epsilon_P\delta_P^2}{2d}T^* =   0,\]
    where the final equality used the definition of $T^*$.
    This contradicts the definition of $T_1$ and establishes that $\P(T_1 < T) = 1$.
\end{proof}
Lemmas~\ref{lem:impact_bound}--\ref{lem:stopping_time} now prove Lemma~\ref{lem:relarb}.

\bibliographystyle{plain}
\setlength{\bibsep}{3pt}
\small
% \footnotesize
% \scriptsize
\bibliography{reference}

\end{document}